\documentclass[sigconf]{acmart}

\usepackage[linesnumbered,ruled,vlined]{algorithm2e}
\usepackage{amsthm}
\usepackage{multirow}
\usepackage{enumitem}
\usepackage{subfigure} 
\usepackage{soul}

\newtheorem{problem}{Problem}

\newtheorem{theorem}{Theorem}
\newtheorem{lemma}{Lemma}
\newtheorem{example}{Example}

\newcommand{\cheng}{}
\newcommand{\chengB}{}
\newcommand{\Yui}{}

\if 0
\newcommand{\chengC}{}
\newcommand{\chengD}{\color{red}}
\newcommand{\YuiR}{\color{cyan}}
\newcommand{\YuiRR}{\color{cyan}}
\newcommand{\kaixin}{\color{brown}}
\newcommand{\laks}[1]{\textcolor{magenta}{#1}}

\fi
\newcommand{\eat}[1]{}
\newcommand{\revision}{}
\newcommand{\chengC}{}
\newcommand{\chengD}{}
\newcommand{\chengE}{}
\newcommand{\YuiR}{}
\newcommand{\kaixin}{}
\newcommand{\laks}{}

\newcommand{\YuiRR}{}

\AtBeginDocument{%
  \providecommand\BibTeX{{%
    \normalfont B\kern-0.5em{\scshape i\kern-0.25em b}\kern-0.8em\TeX}}}

\setcopyright{acmcopyright}
\copyrightyear{2018}
\acmYear{2018}
\acmDOI{XXXXXXX.XXXXXXX}

\acmPrice{15.00}
\acmISBN{978-1-4503-XXXX-X/18/06}

\begin{document}
\sloppy
%\include{Revision/X.response}
%%
%% The "title" command has an optional parameter,
%% allowing the author to define a "short title" to be used in page headers.
\title{Fast Maximum Common Subgraph Search: A Redundancy-Reduced Backtracking Approach}
%\author{Anonymous Authors}
%\affiliation{\institution{ }}
%\email{ }

%%
%% The "author" command and its associated commands are used to define
%% the authors and their affiliations.
%% Of note is the shared affiliation of the first two authors, and the
%% "authornote" and "authornotemark" commands
%% used to denote shared contribution to the research.
\author{Kaiqiang Yu}
\orcid{0000-0003-1153-2902}
\affiliation{%
  \institution{Nanyang Technological University}
  \country{Singapore}
}
\email{kaiqiang002@e.ntu.edu.sg}

\author{Kaixin Wang}
\orcid{0000-0002-6650-2850}
\affiliation{%
  \institution{Beijing University of Technology}
  \city{Beijing}
  \country{China}
}
\email{kaixin001@e.ntu.edu.sg}

\author{Cheng Long}
\orcid{0000-0001-6806-8405}
\authornote{Corresponding author (c.long@ntu.edu.sg)}
\affiliation{%
  \institution{Nanyang Technological University}
  \country{Singapore}
}
\email{c.long@ntu.edu.sg}

\author{Laks Lakshmanan}
\orcid{0000-0002-9775-4241}
\affiliation{%
  \institution{The University of British Columbia}
  \city{Vancouver}
  \country{Canada}
}
\email{laks@cs.ubc.ca}

\author{Reynold Cheng}
\orcid{0000-0002-9480-9809}
\affiliation{%
  \institution{The University of Hong Kong}
  \city{Hong Kong}
  \country{China}
}
\email{ckcheng@cs.hku.hk}

%%
%% By default, the full list of authors will be used in the page
%% headers. Often, this list is too long, and will overlap
%% other information printed in the page headers. This command allows
%% the author to define a more concise list
%% of authors' names for this purpose.

%%
%% The abstract is a short summary of the work to be presented in the
%% article.
\begin{abstract} 
    \laks{Given two input graphs, finding the largest subgraph that occurs in both, i.e., finding the maximum common subgraph,} is a fundamental operator for evaluating the similarity between two graphs in graph data analysis. Existing works for solving the problem are of either theoretical or practical interest, but not both. {\chengB Specifically}, the algorithms with a theoretical guarantee on the running time are known to be not practically efficient; algorithms following the recently proposed backtracking framework called \texttt{McSplit}, run fast in practice but do not have any theoretical guarantees. In this paper, we propose a new backtracking algorithm called \texttt{RRSplit}, which at once achieves better practical efficiency and provides a non-trivial theoretical guarantee on the worst-case running time.  {\chengB To achieve} the former, we develop a series of reductions and upper bounds for reducing redundant computations, i.e., the time  for exploring some unpromising branches \laks{of exploration} that hold no maximum common subgraph. {\chengB To achieve} the latter, we formally prove that \texttt{RRSplit} incurs a worst-case time complexity which matches {\chengB the best-known complexity for the problem.} Finally, we conduct extensive experiments on {\chengE several} benchmark graph collections, and the results demonstrate that our algorithm outperforms the practical state-of-the-art by several orders of {\chengB magnitude}.
\end{abstract}

%%
%% The code below is generated by the tool at http://dl.acm.org/ccs.cfm.
%% Please copy and paste the code instead of the example below.
%%




%%
%% This command processes the author and affiliation and title
%% information and builds the first part of the formatted document.
\maketitle

\section{Introduction}
\label{sec:intro}

Graphs have been increasingly adopted to capture the relationships among entities in various domains, including social media, biological networks, communication networks, collaboration networks  etc. As a result, graph data analysis has gained great attention in the recent years. 
One of the most fundamental problems in graph analysis is \emph{maximum common subgraph search}, which is widely used to measure the similarity of two graphs~\cite{mcgregor1982backtrack,mccreesh2016clique,vismara2008finding,zhoustrengthened,liu2020learning,liu2023hybrid,mccreesh2017partitioning,choi2012efficient,rutgers2010approximate,xiao2009generative,zanfir2018deep,bai2021glsearch}. 
More precisely, a common subgraph between two graphs $Q$ and $G$ refers to a subgraph  that  appears in both $Q$ and $G$. 
% Conceptually, it is defined by a pair of induced subgraphs (i.e., an induced subgraph $q$ of $Q$ and an induced subgraph $g$ of $G$) and a bijection mapping (i.e., $\phi: V_q\rightarrow V_{g}$ from vertices in $q$ to vertices in $g$) such that $q$ and $g$ are \emph{isomorphic} to each other under the bijection $\phi$, which we denote by $\langle q,g,\phi \rangle$. 
\laks{Conceptually, it is defined by {\chengB a pair of subgraphs, $q = (V_q, E_q)$ of $Q$ and  $g = (V_g, E_g)$ of $G$, and a bijection mapping $\phi: V_q\rightarrow V_{g}$} such that $q$ and $g$ are \emph{isomorphic}  under $\phi$, which we denote by $\langle q,g,\phi \rangle$.
Given two graphs $Q$ and $G$, the problem asks for the common subgraph $\langle q,g,\phi \rangle$ with the maximum number of vertices in $q$ (equiv. $g$).} 

The maximum common subgraph search problem 
{\cheng has many applications} 
% is significant 
across various disciplines, 
{\chengB 
{\cheng and} 
has been widely studied~\cite{abu2014maximum,levi1973note,krissinel2004common,suters2005new,mcgregor1982backtrack,mccreesh2016clique,vismara2008finding,zhoustrengthened,liu2020learning,liu2023hybrid,mccreesh2017partitioning,choi2012efficient,rutgers2010approximate,xiao2009generative,zanfir2018deep,bai2021glsearch}}. 
% from both the theoretical 
% % interest
% {\cheng perspective}
% ~\cite{abu2014maximum,levi1973note,krissinel2004common,suters2005new} and the practical perspective~\cite{mcgregor1982backtrack,mccreesh2016clique,vismara2008finding,zhoustrengthened,liu2020learning,liu2023hybrid,mccreesh2017partitioning,choi2012efficient,rutgers2010approximate,xiao2009generative,zanfir2018deep,bai2021glsearch} in the past. 
To be specific, it offers an operator for evaluating the similarity between graphs in graph database systems~\cite{yan2005substructure} and thus 
% {\chengD It} 
has found a wide range of real applications, including cheminformatics~\cite{antelo2020maximum,schmidt2020disconnected}, communication networks~\cite{nirmala2016vertex}, software analysis~\cite{park2011deriving,sun2021effective},  biochemistry~\cite{bonnici2013subgraph,larsen2017cytomcs,ehrlich2011maximum}, and image segmentation~\cite{hati2016image}. 
% As an example of drug discovery and analysis, it would help
{\YuiRR For example, the similarity between two molecules 
% $Q$ and $G$ 
is calculated based on the maximum common subgraph between 
% $Q$ and $G$
{\chengD them}
% %, i.e., $Sim(Q,G)=(|V_q|+|E_q|)/((|V_Q|+|V_G|)\times (|E_Q|+|E_G|))$
~\cite{ehrlich2011maximum}.}
{\chengC Therefore, in drug discovery and analysis, it is used}
to quickly identify a small group of compounds with similar substructures (which tend to {\YuiR preserve} similar properties) for further analysis, 
% which reduces 
{\chengC so as to reduce}
the manual labor and shorten the cycle of discovery~\cite{ehrlich2011maximum}.
Besides, the maximum common subgraph search problem generalizes the well-studied \emph{subgraph matching} problem~\cite{bhattarai2019ceci,ullmann1976algorithm,sun2020rapidmatch,sun2020subgraph,shang2008taming,kim2023fast,han2013turboiso,han2019efficient,cordella2004sub,bi2016efficient,arai2023gup,jin2023circinus,sun2023efficient}. 
%%%%%%%%%%%%%%%%%%%%%%%%% 
\eat{ 
In specific, given two graphs $Q$ and $G$, subgraph matching aims to find from one graph (say, data graph $G$) all embeddings of the other (say, query graph $Q$), i.e., subgraphs of $G$ that are isomorphic to a query graph $Q$. 
Here, an embedding of $Q$ in $G$ is a common subgraph (i.e., a pair of subgraph $p$ and $q$) satisfying the constraint $q=Q$.
{\cheng Note that the embedding of $Q$ in $G$ (if it exists) corresponds to the maximum common subgraph between $G$ and $Q$.}
We note that subgraph matching is too restrictive in some real applications due to the data quality issues and/or potential requirements of the fuzzy search (e.g., no result will be returned if such an embedding does not exist). 
%
% Motivated by all the above, 
{\cheng Therefore,}
we focus on finding the maximum common subgraph between two graphs in this paper.} 
%%%%%%%%%%%%%%%%%%%%%%%%% 
\laks{More specifically, given a data graph $G$ and a query graph $Q$, subgraph matching seeks to find if there is an embedding\footnote{Not to be confused with graph embeddings.} from $Q$ to $G$, where an embedding means a 1-1 function from the nodes of $Q$ to those of $G$ which preserves edges, i.e., the embedding is a witness that $Q$ is isomorphic to a subgraph of $G$. Embedding or subgraph matching is a strong requirement. In real applications data quality is a real concern \cite{chiang2007coverage,balasundaram2011clique} and thus an embedding from $Q$ to $G$ may fail to exist, with subgraph matching yielding no results. In such circumstances, finding a maximum common subgraph is a graceful relaxation of the problem which may still yield useful results. {\YuiRR If the largest common subgraph is very similar to $Q$, it can serve as an approximation to the embedding of $Q$.} 
Given this, in this paper, we focus on finding the maximum common subgraph between two given graphs. }

\smallskip
\noindent\textbf{Challenges and {\chengC existing} methods}. Detecting the maximum common subgraph is quite  challenging as noted in the literature. Specifically, it is NP-hard~\cite{lewis1983michael} and is shown to be {\YuiR at least as hard to approximate as the maximum clique problem: it does not admit any $r$-approximate algorithm that runs in polynomial time (unless $P=NP$), for any $r\geq 1$~\cite{kann1992approximability}.}
%
%to be hard to approximate: it does not admit any $O(n^{\epsilon})$-approximate algorithm that runs in polynomial time (unless $P=NP$), \laks{for any}  $\epsilon>0$, \laks{where} $n$ is {\YuiR the number of vertices in $Q$ and $G$}~\cite{kann1992approximability}. 
%
{\Yui Existing works for solving the problem are of either theoretical or practical interest. On the one hand, some algorithms are designed 
% from the theoretical perspective
{\chengB to improve theoretical time complexity}
~\cite{abu2014maximum,levi1973note,krissinel2004common,suters2005new}. {\revision Assume that $Q$ has the number of vertices no larger than $G$, i.e., $|V_Q|\leq |V_G|$.} They have gradually improved the worst-case time complexity from $O^*(1.19^{|V_Q||V_G|})$~\cite{levi1973note} to $O^*(|V_Q|^{(|V_G|+1)})$~\cite{krissinel2004common}, and to $O^*((|V_Q|+1)^{|V_G|})$~\cite{suters2005new}, 
% which is the \emph{theoretical} state-of-the-art to the best of our knowledge. 
{\chengB which is \laks{the} best-known worst-case time complexity for the problem.} Here, $O^*$ suppresses polynomials.
However, these algorithms are of theoretical interest only and not efficient in practice. {\YuiRR This is mainly because their theoretical results rely on some sophisticated data structures while maintaining them introduces a huge amount of time and/or memory overhead in practice, e.g., the intermediate graph structure built from $Q$ and $G$ has $|V_Q|\times |V_G|$ vertices and could be very large if $G$ is large (e.g., with million nodes).}
On the other hand, quite a few algorithms have been developed {\YuiR towards improving the practical performance}~\cite{levi1973note,mcgregor1982backtrack,mccreesh2016clique,vismara2008finding,zhoustrengthened,liu2020learning,liu2023hybrid,mccreesh2017partitioning}. They are all backtracking (also known as branch-and-bound) algorithms, among which the recent works~\cite{zhoustrengthened,liu2020learning,liu2023hybrid} are based on a newly proposed backtracking framework called \texttt{McSplit}~\cite{mccreesh2017partitioning}. \texttt{McSplit} recursively partitions the search space (i.e., the set of possible common subgraphs) to multiple sub-spaces via a process of branching. Each sub-space corresponds to a branch. 
%
%The redundant computation during the backtracking refers to the time costs for exploring those branches that do not hold the maximum common subgraph to be returned. 
%
Furthermore, \texttt{McSplit} uses the 
% upper bounds of each branch (i.e., 
{\chengB upper bound on the size of common subgraphs that could be found within a branch} for reducing  redundant computations \laks{associated with exploring those branches that do not \laks{lead to finding}  the maximum common subgraph.} 
% The rationale is to 
{\chengB Specifically, it prunes} those branches that have upper bounds no larger than the size of the largest common subgraph seen so far. However, \textit{these algorithms provide no {\chengE non-trivial} theoretical guarantee on the worst-case time complexity.}
%and (2) still suffer from efficiency issues in practice.
}

{
We note that maximum common subgraph is closely related to a well-known graph similarity measure, namely graph edit distance (GED)~\cite{bunke1997relation}. The GED between two graphs $Q$ and $G$ is the  cost of the  least-cost edit path, i.e., a sequence of edit operations that transform $Q$ to $G$. Under a special cost function  which does not charge for edge insertion/deletion and charges an infinite cost for node/edge substitution, the GED computation has been shown to be equivalent to finding the maximum common subgraph~\cite{bunke1997relation}. However, recent state-of-the-art approaches~\cite{chen2019efficient,gouda2016csi_ged,piao2023computing,chang2020speeding,kim2023efficient} for computing GED \eat{cannot be used for providing an efficient solution to our problem since they study computing GED under} assume  different cost functions where the above equivalence does not hold. \eat{In specific, the cost function used in~\cite{bunke1997relation} ignores the costs for edge deletion/insertion and takes the costs for vertex/edge substitution as infinity, while those used in recent algorithms} Specifically, these cost functions do not ignore edge deletion/insertion costs and usually assign the same positive finite costs to various edit operators. Consequently, state of the art algorithms for GED cannot lead to efficient solutions for computing the maximum common subgraph. 

}

\smallskip
\noindent\textbf{Our method}. In this paper, we develop an efficient backtracking algorithm called \texttt{RRSplit}, which \laks{leverages} newly-designed reductions and new upper bounds for decreasing the redundant computations. 
% In particular, with the proposed reductions, 
{\revision {\chengE With} $|V_Q|\leq |V_G|$,} \texttt{RRSplit} achieves a worst-case time complexity of $O^*((|V_G|+1)^{|V_Q|})$, \laks{matching} 
% , 
% to our best knowledge, the state-of-the-art
{\chengB the best-known worst-case time complexity for the problem}~\cite{suters2005new}. 
We note that this theoretical result is remarkable since (1) the algorithm {\chengD with the best-known time complexity  in~\cite{suters2005new}} is of theoretical interest only and is not {\cheng practically} efficient and (2) {\chengD the algorithms with the best-known practical performance, i.e., \texttt{McSplit}~\cite{mccreesh2017partitioning}  and its variants,} do not have any theoretical guarantee on the worst-case time complexity.
{\chengB Specifically}, \texttt{RRSplit} combines the following two kinds of reductions: (i) {\chengB vertex-equivalence-based and (ii) maximality-based}, and a  vertex-equivalence-based upper bound \laks{that we establish}. %We remark that the proposed reductions and upper bound are orthogonal to the existing upper bound techniques.

Vertex-equivalence-based reductions reduce the redundant computations induced by \emph{common subgraph isomorphism (cs-isomorphism)}. 
{\YuiR
{\chengD Given} two common subgraphs $\langle q,g,\phi \rangle$ and $\langle q',g',\phi' \rangle$ of graphs $G$ and $Q$, they are said to be \textit{common subgraph isomorphic} (\textit{cs-isomorphic} for short) if and only if $q$ is isomorphic to $q'$ (or equivalently, $g$ is isomorphic to $g'$). All cs-isomorphic common subgraphs evidently share the same structural information, and exploring all of them is clearly redundant.
%Consider a common subgraph $\langle q,g,\phi \rangle$ {\cheng (where $q$ is an induced subgraph of $Q$ and $g$ is an induced subgraph of $G$)} and a subgraph $q'$ of $Q$ that is isomorphic to $q$ under bijection $f: V_{q'}\rightarrow V_q$. We observe that $q'$ is isomorphic to $g$ under bijection $\phi \circ f$ (since graph isomorphism is an equivalence relation) and thus $\langle q',g, \phi \circ f\rangle$ is a common subgraph. Clearly, $\langle q,g,\phi \rangle$ and $\langle q',g, \phi \circ f\rangle$ carry the same structural information (i.e., $q$, $q'$ and $g$ are isomorphic), for which they are said to be self-{\cheng isomorphic}. Therefore, if $\langle q,g,\phi \rangle$ is found during the backtracking, redundant computations will occur when exploring those branches that hold all common subgraphs self-isomorphic to $\langle q,g,\phi \rangle$. 
% Unfortunately
%{\cheng Furthermore}, given a common subgraph $\langle q,g,\phi \rangle$, the number of its cs-isomorphic common subgraphs $\langle q',g',\phi' \rangle$ {\kaixin could be} \emph{exponentially large}, i.e., $O(|V_Q|^{|V_q|}\cdot |V_q|!)$. \eat{This is because (1) there could be $O(|V_Q|^{|V_q|})$ different subgraphs of $Q$ with the number of vertices \eat{inside} equal to $|V_q|$ and (2) for each subgraph $q'$, there could be $O(|V_q|!)$ isomorphic bijections between $q$ and $q'$.}
%
{\kaixin To reduce this redundancy, we take two sufficient conditions into consideration {\chengD when deciding whether we can prune a branch}. Specifically, } for {\chengC any} common subgraph $\langle q,g,\phi \rangle$ to be found in a branch, if there exists another that is cs-isomorphic to $\langle q,g,\phi \rangle$ (Condition 1) and has been found before (Condition 2), we can safely {\chengC prune the branch.}
% ignore the common subgraph $\langle q,g,\phi \rangle$. 
To facilitate the reduction, 
we will leverage the \emph{vertex equivalence} property \cite{nguyen2019applications} and an \emph{auxiliary data structure} for verifying Condition 1 and Condition 2, respectively {\chengC (details in Section~\ref{subsec:VE-reduction})}.}
%Motivated by this, vertex-equivalence-based reductions aim to prune those formed sub-branches that have all common subgraphs inside self-isomorphic to the one found before, which are designed based on the \emph{structural equivalence} among vertices.  

Maximality-based reductions capture the redundant computations induced by \emph{maximality}. Specifically, a common subgraph $\langle q,g,\phi\rangle$ is maximal if and only if there does not exist any other common subgraph $\langle q',g',\phi' \rangle$ such that $q$ and $g$ are subgraphs of $q'$ and $g'$, respectively. Therefore, the maximum common subgraph is a maximal common subgraph, and those branches that hold only non-maximal ones {\chengD would} incur redundant computations. To reduce them,  we observe one necessary condition for a branch to hold the largest common subgraphs {\chengC (details  in Section~\ref{subsec:maximality-reduction})}. As a result, we can safely prune those branches that \laks{violate} the condition.

{\YuiR
Furthermore, we leverage the vertex equivalence property to derive a new vertex-equivalence-based upper bound, which is tighter than the existing one~\cite{mccreesh2017partitioning} and thus can help to prune more branches 
{\chengC (details in Section~\ref{subsec:upper-bound}).}
% (note that a branch can be pruned if its upper bound is no larger than the largest common subgraph seen so far).
}

\smallskip
\noindent\textbf{Contributions}. We make the following  contributions.
\begin{itemize}[leftmargin=*]
    \item We introduce vertex-equivalence-based reductions for reducing the redundant computation induced by cs-isomorphism (Section~\ref{subsec:VE-reduction}). We further propose  maximality-based reductions for pruning those branches that hold non-maximal common {\chengE subgraphs} only  (Section~\ref{subsec:maximality-reduction}). {\YuiR Finally, we develop a new vertex-equivalence-based upper bound for pruning more branches  (Section~\ref{subsec:upper-bound}}).
    
    \item We propose a new backtracking algorithm called \texttt{RRSplit}, which is based on the newly-designed reductions. It has a worst-case time complexity\eat{of $O^*((|V_G|+1)^{|V_Q|})$, {\chengB which is the same as}} \laks{that matches} the best-known time complexity (of the algorithms that are of theoretical interest only) (Section~\ref{subsec:summary}). 
    
    \item We conduct an extensive empirical evaluation on {\chengE several} benchmark graph collections. Our experiments reveal that our algorithm \texttt{RRSplit} runs several orders of {\cheng magnitude} faster than the state-of-the-art \texttt{McSplitDAL} on the majority of the tested input instances (Section~\ref{sec:exp}). 
\end{itemize}

Section~\ref{sec:preli} provides a formal statement of the problem studied. Section~\ref{sec:sota} reviews the existing framework \texttt{McSplit} and its state-of-the-art variant \texttt{McSplitDAL}. In Section~\ref{sec:related} we discuss related work and conclude the paper in Section~\ref{sec:conclusion}. {
\YuiR Due to space limits, we omit some proofs and sketch others, and all proofs can be found in the 
\ifx \CR\undefined
Appendix. 
\else
technical report~\cite{TR}. 
\fi
}

\section{Preliminaries}
\label{sec:preli}
In this paper, we focus on undirected and unweighted simple {\cheng graphs} without self-loops and parallel edges. {\cheng For ease of presentation, we focus on  graphs without vertex labels, but} 
% We note that 
our methods \eat{{\cheng to be introduced}} can be easily adapted to vertex-labeled graphs. Consider two graphs $Q=(V_Q,E_Q)$ and $G=(V_G,E_G)$\eat{, with vertex sets $V_Q$ and $V_G$ and edge sets $E_Q$ and $E_G$}. For simplicity, we let $u$ and $v$ (and their primed or index variants) denote a vertex in $Q$ and $G$ respectively. Given a vertex set $X\subseteq V_Q$, we use $Q[X]$ to denote the subgraph of $Q$ induced by $X$, i.e., $Q[X]=(X,\{(u,u')\in E_Q \mid u, u'\in X\})$. All subgraphs in this paper are induced subgraphs. We let $q = (V_q, E_q)$ denote an arbitrary induced subgraph of $Q$. 
Given $u\in V_Q$, we denote by $N(u,V_Q)$ (resp. $\overline{N}(u,V_Q)$) the set of neighbours (resp. non-neighbours) of $u$ in $Q[V_Q]$.  \eat{We have the symmetric definitions for a vertex set $Y$ and each vertex $v$ in $G$.} \laks{We use a similar notation for neighbours and non-neighbours of vertices in $G$.}  

We  review the definition of graph isomorphism for  simple graphs without labels.

\begin{definition}[Graph isomorphism~\cite{mcgregor1982backtrack}]
    \label{def:GS}
    $Q$ is said to be isomorphic to $G$ if and only if there exists a \textbf{bijection} $\phi: V_Q\rightarrow V_G$ such that
    \begin{equation}
    \label{eq:isomorphic}
        \forall u, u' \in V_Q: (u,u')\in E_Q \iff (\phi(u),\phi(u'))\in E_G.
    \end{equation}
\end{definition}

\laks{Note that} two isomorphic graphs are structurally equivalent, {\chengC and thus we have} $|V_Q|=|V_G|$ and $|E_Q|=|E_G|$. We next review induced subgraph isomorphism for  graphs. 

\begin{definition}[Induced subgraph isomorphism~\cite{mcgregor1982backtrack}]
    $Q$ is said to be (induced) subgraph isomorphic to $G$ if and only if there exists an \textbf{injection} $\phi: V_Q\rightarrow V_G$ such that
    \begin{equation}
    %\label{eq:isomorphic}
        \forall u, u' \in V_Q: (u,u')\in E_Q {\YuiR \iff} (\phi(u),\phi(u'))\in E_G. %\Longrightarrow 
    \end{equation}
    %that satisfies Equation~(\ref{eq:isomorphic}).
\end{definition}

Notice that  induced subgraph isomorphism is a special case of  graph isomorphism. The injection mapping $\phi: V_Q\rightarrow V_G$ is also known as \emph{embedding} of $Q$ into $G$, {\chengC and thus we have} $|V_Q|\leq |V_G|$ and $|E_Q|\leq |E_G|$. The subgraph matching problem aims to find all embeddings of a small query graph $Q$ in a large data graph $G$. The common induced subgraph is defined as follows. 

\begin{definition}[Common induced subgraph~\cite{mcgregor1982backtrack}]
    \label{def:CIS}
    A common subgraph of $Q$ and $G$, {\cheng denoted by $\langle q,g,\phi \rangle$, is defined as} a \eat{set of vertex pairs} \laks{triple   consisting of an induced subgraph $q$ of $Q$, an induced subgraph $g$ of $G$, and a bijection $\phi: V_q\rightarrow V_g$, such that $q$ is isomorphic to $g$ under   $\phi$.} %{\cheng Formally, we have}
    %\begin{equation}
    %    \langle q,g,\phi \rangle :=\{\langle u,\phi(u) \rangle \mid \forall u\in V_q\}.
   % \end{equation}
\end{definition}

%A common subgraph $\langle q,g,\phi \rangle$ is said to be contained in (or a subgraph of) another common subgraph $\langle q',g',\phi' \rangle$ if and only if $\{\langle u,\phi(u) \rangle \mid \forall u\in V_q\}\subseteq \{\langle u,\phi'(u) \rangle \mid \forall u\in V_{q'}\}$.
By the size of a common subgraph $\langle q, g,\phi\rangle$ we mean the number of {\YuiR vertices in $q$ or $g$}. %vertex pairs. 
Clearly, the size of a common subgraph is at most $\min\{|V_Q|,|V_G|\}$. 
{\YuiR Sometimes, for  ease of presentation, we represent a common subgraph $\langle q,g,\phi \rangle$ by a set of vertex pairs $\{\langle u,\phi(u) \rangle \mid  u\in V_q\}$.}

\begin{example}
Consider the input graphs in Figure~\ref{fig:Input_graph}. The graphs $q := Q[u_1,u_2,u_3,u_4,u_7]$ and  $g := G[v_3,v_4,v_5,v_6,v_7]$ form a common subgraph with size 5 under the bijection $\phi:=\{u_1\!\rightarrow\! v_5,u_2\!\rightarrow\! v_6, u_3\!\rightarrow\! v_4, u_4\!\rightarrow\! v_3, u_7\!\rightarrow\! v_7\}$.%, which can be represented by a set of vertex pairs, i.e., $\{\langle u_1,v_5 \rangle,\langle u_2,v_6 \rangle,\langle u_3,v_4\rangle,\langle u_4,v_3\rangle,\langle u_7,v_7\rangle\}$.
\end{example}
We are ready to formulate the problem studied in this paper.

\begin{problem}[Maximum Common Subgraph~\cite{lewis1983michael}]
    Given two graphs $Q$ and $G$, the Maximum Common Subgraph (MCS) problem aims to find the maximum common subgraph {\cheng of $Q$ and $G$}, i.e., a common subgraph  with the largest number of vertices.
\end{problem}

%We use MaxCS as a shorthand of the maximum common subgraph thoughout this paper.
\eat{If {\chengB we further require} the size of found maximum common subgraph {\cheng to be} at least $|V_Q|$, then the MCS problem would reduce to  subgraph matching (i.e., {\cheng it finds} one embedding of $Q$ in $G$). Therefore, the MCS problem is a natural generalization of the subgraph matching problem. } 
\laks{Note that the MCS problem is a generalization of subgraph matching: there is a MCS between $Q$ and $G$ whose size is $|Q|$ iff $Q$ is isomorphic to a subgraph of $G$. It is well known that  MCS  is NP-hard~\cite{lewis1983michael} and is hard to approximate, i.e., {\YuiR there is no $r$-approximate PTIME algorithm for the problem for any $r\geq 1$~\cite{kann1992approximability}, \laks{unless P=NP}.} %$O(n^{\epsilon})$-approximate PTIME algorithm for the problem for any $\epsilon>0$, where $n$ is the input instance size~\cite{kann1992approximability}. 
} 

\eat{ 
\smallskip
\noindent\textbf{Hardness.} We remark that the problem of finding the maximum common subgraph is NP-hard~\cite{lewis1983michael}. Besides, 
% it is shown to be 
{\cheng the problem is}
hard to approximate, {\cheng e.g.,} it does not admit any $O(n^{\epsilon})$-approximate algorithm that runs in polynomial time (unless P=NP), where $\epsilon>0$ and $n$ is the size of an input instance~\cite{kann1992approximability}.  
} 

\begin{figure}[]
		\subfigure[\textsf{Input graph $Q$}]{
			\includegraphics[width=2.7cm]{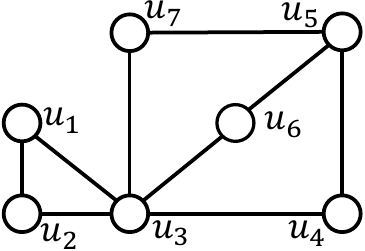}
		}
        \hspace{0.4in}
		\subfigure[\textsf{Input graph $G$}]{
			\includegraphics[width=2.7cm]{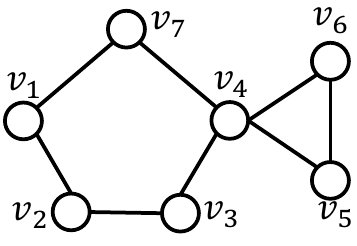}
		}
  \vspace{-0.2in}
	\caption{Input graphs used throughout the paper}
	\label{fig:Input_graph}
 \vspace{-0.2in}
\end{figure}
\section{The Basic Framework: \texttt{McSplit}}
\label{sec:sota}

% We first build necessary background of a basic framework called \texttt{McSplit}~\cite{mccreesh2017partitioning} for finding the maximum common subgraph. 
{\cheng \noindent\textbf{Overview.} \texttt{McSplit},  a \emph{backtracking} (aka \textit{branch-and-bound}) algorithm, \eat{also known as \emph{branch-and-bound} algorithm,} has been widely adopted for the MCS problem in recent years and has achieved the state-of-the-art performance in practice~\cite{bai2021glsearch,liu2020learning,zhoustrengthened,liu2023hybrid}.}
The \laks{key idea} of \texttt{McSplit} is to recursively expand a partial solution $S$ (which is {\cheng the largest} common subgraph {\cheng seen so far}) via a process of \emph{branching}. Specifically, the branching process partitions the current problem instance of finding the maximum common subgraph into several subproblem instances. Each problem instance {\cheng corresponds} to a \emph{branch} {\cheng and} is denoted by $(S,C)$. Here, $S$ is a \emph{partial solution} {\Yui (i.e., set of vertex pairs)} and $C$ is the \emph{candidate set} consisting of \emph{candidate pairs} (i.e., $\langle u, v\rangle$) used to expand the partial solution $S$. Solving the instance or branch $(S,C)$ means finding the largest common subgraph $S^*$ in the branch; a common subgraph is said to be in a branch $(S,C)$ if and only if \emph{it contains $S$ and is within the set $S\cup C$, {\Yui i.e., $S\subseteq S^*\subseteq S\cup C$}}.
%{\cheng which contains} $S$ and {\cheng is} within the set $S\cup C$.
{\Yui Given two vertex subsets $V_q\subseteq V_Q$ and $V_g\subseteq V_G$, we consider all pairs of   vertices from $V_q$ and $V_g$, i.e.,
$V_q\times V_g = \{\langle u,v \rangle\mid u\in V_q, v\in V_g\}$.
%
%Let $\langle V_Q\times V_G \rangle$ denote the set of all possible vertex pairs (i.e., $\langle V_Q\times V_G \rangle=\{\langle u,v \rangle\mid u\in V_Q, v\in V_G\}$). 
%
Note that solving the initial branch $(\emptyset,V_Q\times V_G)$ finds the largest common subgraph of $Q$ and $G$.}
%, where $\langle V_Q\times V_G \rangle$ denotes the set of all possible vertex pairs (i.e., $\langle V_Q\times V_G \rangle=\{\langle u,v \rangle\mid u\in V_Q, v\in V_G\}$).

To solve a branch $(S,C)$, the branching process selects a vertex $u$ {\cheng appearing} in $C$ as a \emph{branching vertex},  and then creates two groups of new sub-branches by either including $u$ into the solution \eat{(this corresponds to the first one)} or discarding $u$ from the candidate set {\Yui and thus also from the solution}. \eat{(this corresponds to the second one).} Specifically, \textbf{in the first group}, each formed branch includes into $S$ one candidate pair containing $u$ and excludes from $C$ all \laks{pairs} containing $u$ (note that a common subgraph has each vertex {\cheng appearing} in at most one pair); consequently, for each candidate pair that contains $u$, i.e., $\langle u,v \rangle$, \laks{we form a new branch corresponding to} $(S\cup\{\langle u,v \rangle\}, C\backslash u\backslash v)$,  
%
%{\kaixin $(S\cup\{\langle u,v \rangle\}, C\setminus (\{\langle u, \cdot \rangle\}\cup \{\langle \cdot, v \rangle\}))$}
%
where $C\backslash u\backslash v$ denotes the set obtained by removing from $C$ all  candidate pairs  containing $u$ or $v$, formally,
\begin{equation}
    C\backslash u\backslash v:= C\backslash ((\{u\}\times V_g) \cup ( V_q\times\{v\})).
\end{equation}
\textbf{In the second group}, we form only one branch by excluding from $C$ all candidate pairs containing $u$, {\cheng i.e.,} $(S,C\backslash u)$. Clearly, \eat{solving all the formed branches solves} the solution to  $(S,C)$ \eat{(since its solution} \laks{is the largest one among those found {\cheng from} the branches above.} We illustrate this next. 
\laks{
\begin{example}
\label{ex:split} 
Consider \eat{an example of {\cheng the} {\chengB branching} process on} the given pair of  input graphs in  Figure~\ref{fig:Input_graph}.  The splitting process is partially depicted in Figure~\ref{fig:example_branching} (\laks{ignore the ``$D$'' terms in the figure for now)}). For the initial branch $B_0=(\emptyset,\{u_1,u_2,...,u_7\}\times \{v_1,v_2,...,v_7\} )$, McSplit  selects the branching vertex $u_1$, and then creates the first group of branches  $B_i=(\{\langle u_1,v_i \rangle\},\{u_2,...,u_7\}\times (\{v_1,v_2,...,v_7\}\backslash \{v_i\}))$ for $1\leq i\leq 7$, each of which includes one candidate pair $\langle u_1,v_i \rangle$ into the solution, and the second group of a single branch, namely $B_8=(\emptyset,\{u_2,...,u_7\}\times \{v_1,v_2,...,v_7\})$, which excludes $u_1$ from the solution.
\end{example} 
}

To improve the efficiency, \laks{McSplit}  further applies a \emph{reduction rule} and an \emph{upper-bound-based pruning} rule for a newly formed branch $(S,C)$. Specifically, the reduction rule {\cheng removes} from the candidate set $C$ those candidate pairs $\langle u,v\rangle$ that cannot form a common subgraph with $S$, i.e., $S\cup \{\langle u,v \rangle$\} cannot be a common subgraph, thus narrowing the search space: note that any supergraph of a non-common subgraph cannot be a common subgraph and thus we can remove them safely. The upper-bound-based pruning rule {\cheng computes} an upper bound on the \laks{size of the} largest common subgraph in the branch and {\cheng prunes} the branch if the upper bound is no larger than the \laks{size of the largest common subgraph found so far} {\YuiR (details will be discussed in Section~\ref{subsec:upper-bound})}. 
{\chengB Below, \laks{we formally state the}  reduction rule.}
\begin{figure}[]
		\includegraphics[width=0.45\textwidth]{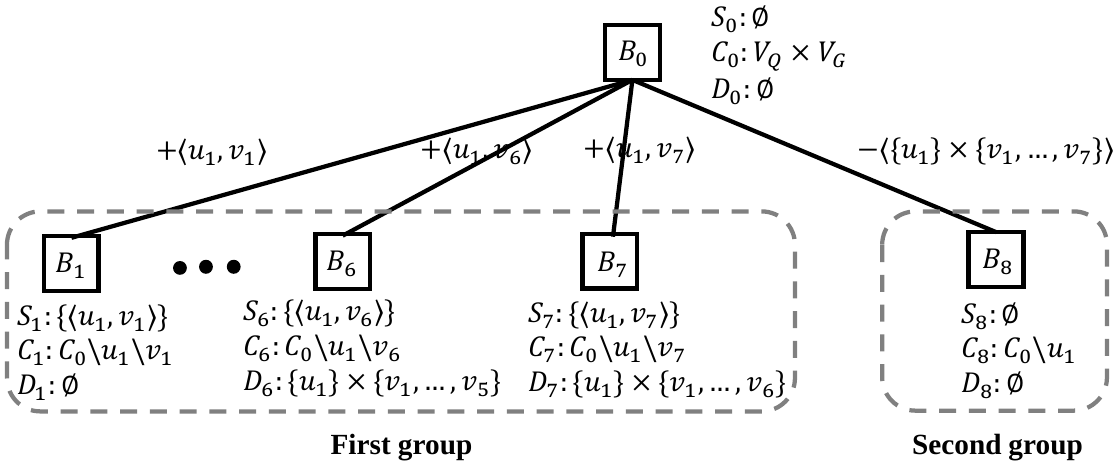}
  \vspace{-0.15in}
	\caption{Illustrating the backtracking process (``+'' means to {\chengB move} vertex pairs from $C$ to $S$ and ``-'' means to remove vertex pairs from $C$)}
 \vspace{-0.2in}
	\label{fig:example_branching}
\end{figure}

{\Yui
\noindent\textbf{Reduction {\chengC rule}.} \emph{Consider a branch $(S,C)$. A candidate pair $\langle u,v \rangle$ in $C$ cannot form any common subgraph with $S$ if there exists a vertex pair $\langle u',v' \rangle$ in $S$ such that $u$ and $v$ are not simultaneously adjacent or non-adjacent to $u'$ and $v'$, respectively.} 

\laks{The soundness of this rule} \eat{can be verified based on} \laks{immediately follows from}  Definition~\ref{def:CIS}. 
{
%Based on the partial set $S$, the refined candidate set $C$ can be split as several subsets, i.e.,
%\begin{equation}
%    C=X_1\times Y_1 \cup ... \cup X_c \times Y_c,
%\end{equation}
%where $c$ is a positive integer and every vertex pair $\langle u,v \rangle$ in $X_i\times Y_i$ $(1\leq i\leq c)$ 
%
We note that the above reduction rule can be applied in a recursive way and \emph{the refined candidate set $C$ can be split as several subsets, i.e., $C=X_1\times Y_1 \cup ... \cup X_c \times Y_c$ where $c$ is a positive integer and $X_i$ and $X_j$ (resp. $Y_i$ and $Y_j$) are disjoint}. We show this as follows.
More precisely, starting from the basis case of the initial branch $B_0$ with $S_0=\emptyset$ and $C_0= V_Q\times V_G $, none of the candidate pairs in $C_0$ can be pruned by the reduction rule since $S_0$ is empty. Then, consider the recursive case of a branch $B=(S,C)$ with $C=X_1\times Y_1 \cup ... \cup X_c \cup Y_c $ where $c$ is a positive integer and $X_i$ and $X_j$ (resp. $Y_i$ and $Y_j$) are disjoint for $1\leq i\leq c$. For one sub-branch $(S\cup\{\langle u,v\rangle\},C\backslash u\backslash v)$ formed in the first group by including $\langle u,v\rangle$ to $S$, the candidate set $C\backslash u \backslash v$ can be refined as
\begin{eqnarray}
    \label{eq:update_candidate_set}
    \bigcup_{i=1}^c  N(u,X_i)\times N(v,Y_i)  \cup  \overline{N}(u,X_i\backslash \{u\})\times \overline{N}(v, Y_i\backslash\{v\}). 
   % \{N(u,X_i)\times N(v,Y_i) \mid 1\leq i \leq c\} \cup \nonumber\\ \{\overline{N}(u,X_i\backslash \{u\})\times \overline{N}(v, Y_i\backslash\{v\}) \mid 1\leq i \leq c\}
\end{eqnarray}
since those vertex pairs in $N(u,X_i)\times \overline{N}(v, Y_i\backslash\{v\})  \cup  \overline{N}(u,X_i\backslash \{u\})\times N(v,Y_i)$ can be pruned by the above reduction. Clearly, $N(u,X_i)$ and $\overline{N}(u,X_i\backslash \{u\})$ (resp. $N(v,Y_i)$ and $\overline{N}(v, Y_i\backslash\{v\}$) are disjoint for $1\leq i\leq c$. For one sub-branch $(S,C\backslash u)$ formed in the second group, none of the candidate pairs in $C\backslash u$ can be pruned by the reduction rule since $S$ remains unchanged. Suppose $u\in X_i$, then the refined candidate set $C\backslash u$ can be represented by $X_1\times Y_1\cup ... \cup (X_i\backslash\{u\})\times Y_i \cup ...\cup X_c\times Y_c$ where any two subsets are disjoint as well.

We remark that all candidate sets $C$ mentioned in the following sections refer to the ones refined by the reduction rule and thus can be represented by $C=X_1\times Y_1\cup ...\cup X_c\times Y_c$, \laks{for some $c$.} Given this, we define
\begin{equation}
    \mathcal{P}(C)=\{X_i\times Y_i \mid 1\leq i \leq c\}.
\end{equation}
%We assume that $X_i$ and $X_j$ (resp. $Y_i$ and $Y_j$) are disjoint, which holds for the base case and   

%Consider an immediate sub-branch of $B_0$ which is formed by including candidate pair $\langle u,v \rangle$ into the partial solution. Those candidate pairs in $ N(u,V_Q)\times \overline{N}(v,V_G)$ and those in $\overline{N}(u,V_Q)\times N(v,V_G)$ can be pruned by the reduction rule. As a result, the refined candidate set $C\backslash u\backslash v$ is $ N(u,V_Q)\times N(v,V_G) \cup  \overline{N}(u,V_Q\backslash\{u\})\times \overline{N}(v,V_G\backslash\{v\}) $, which is \emph{split} as two subsets. Consider an immediate sub-branch of $B_0$ which is formed by excluding $u$, 
}

\begin{example}
    \label{exp:branching}
   Consider branch $B_6$ in the first group in Figure~\ref{fig:example_branching}. We have $X_1=N(u_1,V_Q)=\{u_2,u_3\}$, $X_2=\overline{N}(u_1,V_Q\backslash\{u_1\})=\{u_4,u_5,u_6,u_7\}$, $Y_1=\{v_4,v_5\}$ and $Y_2=\{v_1,v_2,v_3,v_7\}$. Therefore, the candidate set becomes $\{u_2,u_3\}\times \{v_4,v_5\}\cup\{u_4,u_5,u_6,u_7\}\times \{v_1,v_2,v_3,v_7\}$. Then, consider a sub-branch of $B_6$ formed by further including $\langle u_7,v_7 \rangle$. We can deduce the candidate set by splitting $X_1\times Y_1 $ to $\{u_3\}\times \{v_4\} \cup \{u_2\}\times \{v_5\} $ and splitting $ X_2\times Y_2$ to $\{u_5\}\times \{v_1\}\cup  \{u_4,u_6\}\times \{v_2,v_3\}$.
\end{example}
}

\noindent\textbf{\laks{Algorithm Outline.}} We summarize the details of \texttt{McSplit} in Algorithm~\ref{alg:mcsplit}. It maintains the currently found largest common subgraph $S^*$ (Line 4) and terminates the branch by upper-bound-based pruning (Line 5). Besides, it branches by selecting (from $\mathcal{P}(C)$) a vertex $u$ and the corresponding subset $X\times Y$, called \emph{branching subset}, that $u$ belongs to (Line 6, note that all candidate pairs {\cheng containing} $u$ are within $X\times Y$), and creates two groups of branches as discussed before (Lines 8-12). 
For the first group, the ordering of formed branches depends on that of the candidate pairs to be included into $S$, which is specified by a policy (Line 9).
We note that \texttt{McSplit} \textit{adopts heuristic policies for selecting $X\times Y$, $u$, and $v$}. %Specifically, it {\Yui first selects from $C$ one subset} $\langle X\times Y \rangle$ with the smallest value of $|X|\times |Y|$, {\Yui then selects from $X$ one vertex} $u$ with the smallest vertex ID, {\Yui and finally iteratively selects from $Y$ each vertex $v$ with the smallest vertex ID for forming a sub-branch $(S\cup\{\langle u,v \rangle\}, C\backslash u\backslash v)$} (note that each vertex in $Q$ or $G$ is assigned an unique integral ID, i.e., $\{0,1,...,|V_Q|-1\}$ for $Q$ and $\{0,1,...,|V_G|-1\}$ for $G$). Therefore, $u$ and $v$ {\cheng are} selected based on the given orderings of vertices in $Q$ and $G$, which are fixed during the recursion.

\begin{algorithm}[t]
\small
\caption{An existing framework: \texttt{McSplit}~\cite{mccreesh2017partitioning}}
\label{alg:mcsplit}
\KwIn{Two graphs $Q=(V_Q,E_Q)$ and $G=(V_G,E_G)$}
\KwOut{The maximum common subgraph}
$S^*\leftarrow \emptyset$; \tcp{Global variable}
\texttt{McSplit-Rec}$(\emptyset,V_Q\times V_G)$; \textbf{Return} $S^*$;\\
\SetKwBlock{Enum}{Procedure \texttt{McSplit-Rec}$(S,C)$}{}
%\SetKwBlock{update}{Procedure \texttt{Update}$(S,C)$}{}
\Enum{
    \lIf{$|S|>|S^*|$}{$S^*\leftarrow S$}
    \tcc{Termination}
    \lIf{$C=\emptyset$ or the upper bound is no larger than $|S^*|$}{\textbf{return}}
    \tcc{Branching}
    {\Yui Select a branching vertex $u$ and a branching subset $X\times Y$ from $\mathcal{P}(C)$  based on a policy\;}
    $Y_{temp}\leftarrow Y$\;
    \tcc{First group: branches formed by including $u$}
    \For{$i=1,2,...,|Y|$}{
        Select and {\chengB move} a vertex $v$ from $Y_{temp}$ based on a policy\;
        Create a candidate set $C_i$ based on $\langle u,v\rangle$ and Equation~(\ref{eq:update_candidate_set})\;
        \texttt{McSplit-Rec}($S\cup\{\langle u,v\rangle\},C_i$)\;
    }
    \tcc{Second group: one branch formed by excluding $u$}
    \texttt{McSplit-Rec}($S,C\backslash u$)\;
}
\end{algorithm}

% \smallskip
% \noindent\textbf{Variants of \texttt{McSplit}.} Quite a few studies adopt \texttt{McSplit} for solving the MCS problem, and 

{\cheng Existing algorithms that are based on \texttt{McSplit} differ in the strategies of optimizing} the policies of selecting vertices in  line 6 and line 9 (e.g., via some learning-based techniques) to find the largest  \laks{common}  subgraph as soon as possible during the recursion~\cite{zhoustrengthened,liu2020learning,liu2023hybrid}. 
{\Yui However, these algorithms (1) provide no theoretical guarantee on the worst-case time complexity and (2) still suffer from efficiency issues in practice due to  significant redundant computations.}
%
%We note that the learned policies dynamically select the next vertex at line 6 and line 9 according to the running-time contexts.
%
%The rationale is that the earlier the largest common subgraph is found, the more branches will be pruned by the upper-bound-based pruning (Line 5).

\section{Redundancy-Reduced Splitting: \texttt{RRSplit}}
\label{sec:RRSplit}

In this part, we present our backtracking algorithm called \texttt{RRSplit}. 
First, we propose a vertex-equivalence-based reduction for pruning those redundant branches that hold all common subgraphs cs-isomorphic to {\chengC one that has been already found} (Section~\ref{subsec:VE-reduction}).
Second, we introduce a newly-designed maximality-based reduction for pruning those redundant branches that hold only non-maximal common subgraphs (Section~\ref{subsec:maximality-reduction}). 
{\YuiR Third, we develop a new vertex-equivalence-based upper bound on the size of common subgraphs that can be found in a branch for further pruning those branches that hold only small common subgraphs (Section~\ref{subsec:upper-bound}).}
\laks{Finally, we summarize the \texttt{RRSplit} algorithm, which is based on the above \eat{carefully-designed} reductions, and analyze its worst-case time complexity (Section~\ref{subsec:summary}). In particular, we show \texttt{RRSplit} has a worst-case time complexity of $O^*((|V_G|+1)^{|V_Q|})$, matching the best-known worst-case time complexity of the state of the art. We will later show (Section~\ref{sec:exp}) that unlike the state of the art,  \texttt{RRSplit} is very efficient in practice.}

\subsection{Vertex-Equivalence-based Reduction}
\label{subsec:VE-reduction}
We first introduce the concept of \emph{common subgraph isomorphism (cs-isomorphism)}.
\begin{definition}[Common subgraph isomorphism] \laks{Consider two common subgraphs $\langle q,g,\phi \rangle$ and $\langle q',g',\phi' \rangle$ of graphs $G$ and $Q$. They are} 
    \eat{$\langle q,g,\phi \rangle$ is } said to be common subgraph isomorphic (cs-isomorphic) \eat{to $\langle q',g',\phi' \rangle$} if and only if $q$ is {\chengC isomorphic} to $q'$ (or equiv., $g$ is {\chengC isomorphic} to $g'$).
\end{definition}

\laks{All cs-isomorphic common subgraphs evidently  share the same structural information, and exploring all of them is clearly redundant.}  
% so is redundant for the purpose of 
\eat{{\chengC thus it would introduce redundancy if we explore both of them for}
finding the largest common subgraph.} \laks{We reduce this redundancy as follows.}  
For a common subgraph $\langle q,g,\phi \rangle$ to be found in a branch, we can safely {\chengC ignore the common subgraph $\langle q,g,\phi \rangle$}, {\revision if there exists another one $\langle q',g',\phi' \rangle$ that satisfies the following two conditions:
\begin{itemize}
    \item \textbf{Condition 1}: $\langle q',g',\phi' \rangle$ is cs-isomorphic to $\langle q,g,\phi \rangle$;
    \item  \textbf{Condition 2}: $\langle q',g',\phi' \rangle$ has been found before.
\end{itemize}}
%if there exists another one that is cs-isomorphic to $\langle q,g,\phi \rangle$ {\chengC (Condition 1)} and has been found before {\chengC (Condition 2)},  

\laks{To facilitate the reduction, we make use of Condition 1 and Condition 2, for which we will leverage the \emph{vertex equivalence} {\chengC property} and an \emph{auxiliary data structure} {\chengC respectively}.} 

\smallskip
\noindent\textbf{Vertex equivalence}. 
% We start with an important graph property, namely 
{\chengC The \emph{structural equivalence} property}
% , which 
has been widely used to speed up subgraph matching tasks~\cite{nguyen2019applications,yang2023structural,kim2021versatile}. Conceptually, two vertices are  \emph{structurally equivalent}  if and only if they have the \emph{same} set of neighbours. 
% Formally, we have: 
{\chengC Formally, }

\begin{definition}[Structural equivalence~\cite{nguyen2019applications}]
    %\label{def:structural_eqv}
    Two vertices $u$ and $v$ in $Q$ are structurally equivalent, denoted  $u\sim v$, if and only if  
    \begin{equation}
        \forall u'\in V_Q, (u,u')\in E_Q \Leftrightarrow (v,u')\in E_Q. 
    \end{equation}
\end{definition}

Clearly, structural equivalence is an equivalence relation. Therefore, we can partition the vertices of graph $Q$ into equivalence classes, with the equivalence class of vertex $u\in V_Q$ defined as 
\begin{equation}
    \Psi(u)  := \{u'\in V_Q \mid u'\sim u\},
\end{equation}
where $u\in V_Q$ is a representative of class $\Psi(u)$.
We remark that this process can be done in $O(|V_Q|\delta_Q d_Q)$ {\chengB time} where $\delta_Q$ and $d_Q$ are the degeneracy and the maximum degree of the graph $Q$, respectively~\cite{nguyen2019applications,yang2023structural,kim2021versatile}.

\begin{figure}[]
		\includegraphics[width=0.45\textwidth]{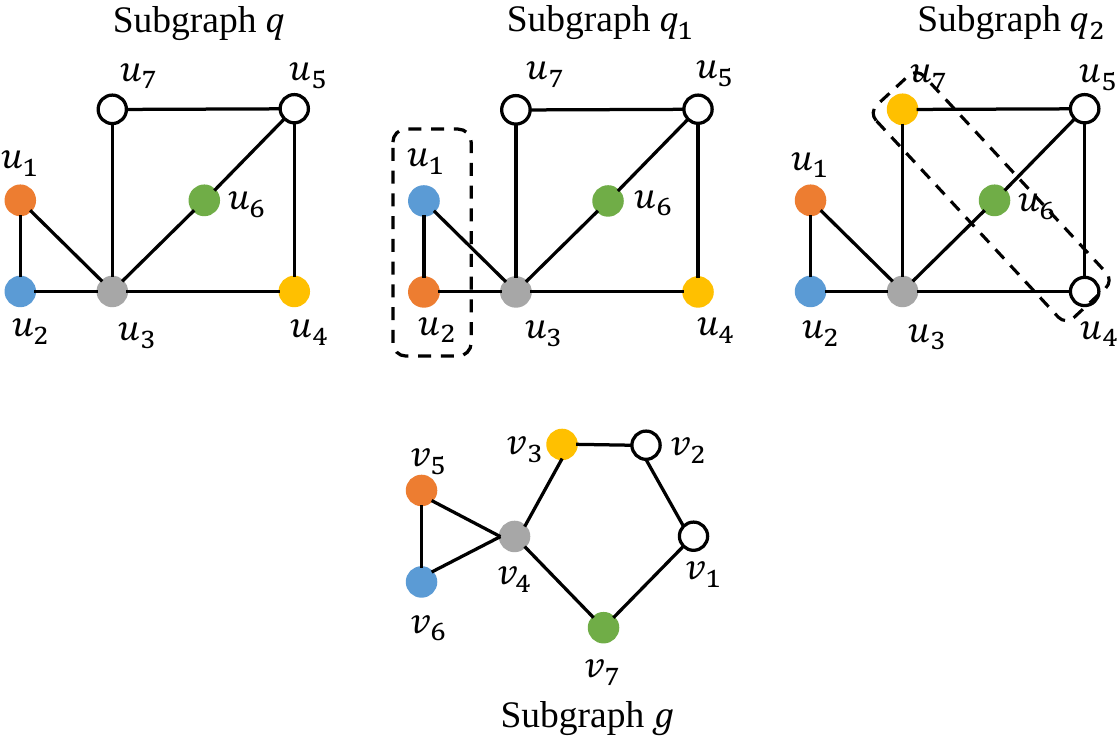}
	\caption{Illustrating cs-isomorphism and vertex equivalence (vertices, denoted by colored bullet circles, induce subgraphs $q$, $q_1$, $q_2$ and $g$; vertices in $\{u_1,u_2\}$ and $\{u_4,u_6,u_7\}$ are structurally equivalent, respectively; $\langle q,g,\phi \rangle$ is cs-isomorphic to $\langle q_1,g,\phi_1 \rangle$ (Case 1, say, exchange the mapped vertices of $u_1$ and $u_2$) and $\langle q_2,g,\phi_2 \rangle$  (Case 2, say, replace $u_4$ with $u_7$) where vertices with the same color indicate the bijection)}
 \label{fig:example_self_iso}
\end{figure}

Based on vertex equivalence, we can \laks{identify} several common subgraphs \laks{that are} cs-isomorphic to a given one $\langle q,g,\phi \rangle$ by swapping one vertex in $V_q$ with its structurally equivalent counterpart, which falls into  two cases.
{\YuiR In specific, consider a vertex $u$ in $V_q$ and one of its structurally equivalent counterparts  $u_{equ}$ in $\Psi(u)$. {\cheng We can obtain a cs-isomorphic common subgraph in two cases.} If $u_{equ}$ is also in $V_q$, {\cheng we can exchange} the mapped vertices of $u_{equ}$ and $u$, i.e., {\cheng we replace} $\langle u,\phi(u) \rangle$ and $\langle u_{equ},\phi(u_{equ}) \rangle$ with $\langle u,\phi(u_{equ}) \rangle$ and $\langle u_{equ},\phi(u) \rangle$; Otherwise, we replace $u$ with $u_{equ}$, i.e., replacing $\langle u,\phi(u)\rangle$ with $\langle u_{equ},\phi(u) \rangle$.}
Formally, we have the following lemma, {\cheng which can be easily verified} (see the examples in Figure~\ref{fig:example_self_iso} for a visual illustration of the lemma).
\begin{lemma}
    \label{lemma:cs}
    Let $S=\langle q,g\eat{p},\phi \rangle$ be a common subgraph of given graphs $Q$ and $G$, $u$ be a vertex in $V_q$ and $u'$ be a vertex in $\Psi(u)$. Then {\kaixin one of the following cases holds}. 
    \begin{itemize}[leftmargin=*]
        \item[]\textbf{Case 1: $u'\in V_q$}. $S'=S\backslash \{\langle u,\phi(u) \rangle,\langle u',\phi(u') \rangle \}\cup \{\langle u,\phi(u') \rangle,\langle u',\\\phi(u) \rangle\}$ is a common subgraph cs-isomorphic to $S$. \;\; {\kaixin\eat{OR}} 
        \item[]\textbf{Case 2: $u'\notin V_q$}. $S'=S\backslash \{\langle u,\phi(u) \rangle\}\cup \{\langle u',\phi(u) \rangle\}$ is a common subgraph cs-isomorphic to $S$.
    \end{itemize}
\end{lemma}

\smallskip
\noindent\textbf{Auxiliary data structure}. 
% Though a few cs-isomorphic common subgraphs can be constructed based on Lemma~\ref{lemma:cs} for verifying Condition 1, none of them may have been found before and thus Condition 2 fails to satisfy. 
{\chengC To facilitate the verification of Condition 2, i.e., whether a common subgraph that is cs-isomorphic to a current one has been found before}, we introduce a new data structure, namely exclusion set (denoted by $D$).
{\chengC $D$} is recursively maintained for each branch, and thus each branch is denoted by $(S,C,D)$. Specifically, $D$ is a set of vertex pairs that have been considered for expanding the partial solution and must not be included in any common subgraphs within the branch. Formally, the exclusion set is maintained as follows (\laks{illustrated in Figure~\ref{fig:example_branching} -- see the ``$D$'' terms now!}). 
\begin{itemize}[leftmargin=*]
    \item \textbf{Initialization}. The exclusion set is initialized to {\kaixin be} empty at the initial branch, i.e., $(\emptyset, V_Q\times V_G, \emptyset)$.
    \item \textbf{Recursive update}. Consider {\chengC the} branching at a branch $(S,C,D)$. For the first group where the $i^{th}$ sub-branch $(S_i,C_i,D_i)$ is formed by including $\langle u,v_i \rangle$ into $S$, we update the exclusion set to $D_i=D\cup \{\langle u,v_1 \rangle,\langle u,v_2 \rangle,\cdots ,\langle u,v_{i-1} \rangle\}$. For the second group where one sub-branch $(S',C',D')$ is formed, we set $D'=D$.
\end{itemize}
%We note that, for a vertex pair $\langle u,v \rangle$ in exclusion set $D$, there exists some common subgraphs containing $\langle u,v \rangle$ that has been found before. This will help us to verify Condition 2.
Consider a branch $(S,C,D)$ and a vertex pair $\langle u',v' \rangle$ in the exclusion set $D$, as shown in Figure~\ref{fig:exclusion_set}. There exists an {\YuiR ancestor}\eat{ascendant branch} of $(S,C,D)$, denoted by $(S_{anc},C_{anc},D_{anc})$, where $u'$ is selected as the branching vertex. Clearly, $\langle u',v' \rangle$ is not in $D_{anc}$ and will be \laks{added} to $D_{anc}$ after $B'_{anc}$ {\chengC is formed}, \laks{i.e., more precisely $D'_{anc} = D_{anc}\cup\{(u',v')\}$}. Therefore, all common subgraphs within the sub-branch $B'_{anc}$, which must contain $\langle u',v' \rangle$, have been found before $(S,C,D)$. This will help us  verify Condition 2.
\begin{figure}[]
		\includegraphics[width=0.3\textwidth]{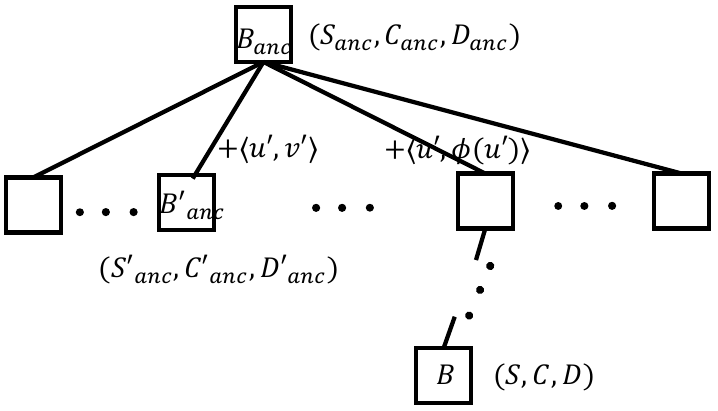}
  \vspace{-0.15in}
	\caption{Illustrating the exclusion set $D$ {\YuiR ($\langle u',v'\rangle$ is a vertex pair in $D$; $B_{anc}$ is an ancestor of $B$, where $u'$ is selected as the branching vertex)}}
 \vspace{-0.2in}
	\label{fig:exclusion_set}
\end{figure}

Based on  vertex equivalence and exclusion set, we are now ready to develop the reductions. Consider the branching process at a branch $(S=\langle q,g,\phi\rangle,C,D)$ where $X\times Y$ in $\mathcal{P}(C)$ and vertex $u$ in $X$ are selected as the branching subset and branching vertex, respectively. 
% In general, the reductions have two cases.
% {\chengC We consider }

\smallskip
\noindent\textbf{Reduction at the first group}. Consider a sub-branch formed at the first group by including one vertex pair $\langle u,v\rangle$ where $v\in Y$. We note that {\chengC \emph{each}} common subgraph $S_{sub}$ to be found in this sub-branch must include $\langle u,v \rangle$. 
We observe that if \emph{there exists a vertex pair $\langle u_{equ},v \rangle$ in $D$ such that $u_{equ}$ is structurally equivalent to $u$, i.e., $u_{equ}\in\Psi(u)$}, Conditions 1 \& 2 hold for $S_{sub}$ and thus the branch can be pruned. Below, we elaborate on the details. 

We first show that Condition 1 holds:  Clearly, $S$ contains a vertex pair $\langle u_{equ},\phi(u_{equ}) \rangle$ since otherwise $D$ will not include $\langle u_{equ},v \rangle$ according to the maintenance of $D$.
Therefore, we can construct the following common subgraph {\chengC $S_{iso}$, which is} cs-isomorphic to $S_{sub}$ based on Case 1 of Lemma~\ref{lemma:cs} {\chengC (essentially, we exchange the mapped vertices of $u_{equ}$ and $u$)}.
\begin{equation}    
\label{eq:iso1}
S_{iso}\!=\!S_{sub}\backslash\{{\chengC \langle u,v\rangle,\langle u_{equ},\phi(u_{equ})\rangle}\} \!\cup\! \{\langle u_{equ},v\rangle,\langle u,\phi(u_{equ}) \rangle\}
\end{equation}
{\chengC Clearly, $S_{iso}$ is cs-isomorphic to $S_{sub}$ given that $u$ and $u_{equ}$ are structurally equivalent.}

We next show that $S_{iso}$ has been found before and thus Condition 2 holds: In specific, consider an ancestor \eat{ascendant branch} of $(S,C,D)$, denoted by $(S_{anc},C_{anc},D_{anc})$, where $u_{equ}$ is selected as the branching vertex, as visually illustrated in Figure~\ref{fig:exclusion_set}. {\chengC Since} $S_{iso}$ contains $\langle u_{equ},v\rangle$, {\chengC we check} whether it has been found in the sub-branch of $(S_{anc},C_{anc},D_{anc})$, which is formed by including $\langle u_{equ},v\rangle$. The answer is interestingly positive. The rationale behind is that (1) $S_{iso}\backslash\{\langle u,\phi(u_{equ})\rangle\}$ is a subset of $S_{anc}\cup C_{anc}$ (this can be easily verified by the facts that $S_{sub} \subseteq S_{anc}\cup C_{anc}$ and $\langle u_{equ},v \rangle\in C_{anc}$) and (2) $\langle u,\phi(u_{equ}) \rangle$ constructed by ours is in $C_{anc}$ (this holds due to the vertex equivalence between $u$ and $u_{equ}$
% , which will be discussed later in Lemma~\ref{lemma:reduction_for_VE1}
). Therefore, $S_{iso}$ is a common subgraph in $(S_{anc},C_{anc},D_{anc})$, i.e., $S_{anc}\subseteq S_{iso} \subseteq S_{anc}\cup C_{anc}$, and has been found before $(S,C,D)$ since it contains $\langle u_{equ},v\rangle$. 

In summary, we have the following lemma and reduction rule.
\begin{lemma}
    \label{lemma:reduction_for_VE1}
    Let $(S,C,D)$ be a branch. Common subgraph $S_{iso}$ {\chengC defined in} Equation~(\ref{eq:iso1}) has been found before the formation of $(S,C,D)$.
\end{lemma}
\begin{proof}
    %Omitted for lack of space. See the anonymous technical report~\cite{TR} for details. 
    {\revision
    We note that the recursive branching process forms a recursion tree where each tree node corresponds to a branch. Consider the path from the initial branch $(\emptyset, V_Q\times V_G,\emptyset)$ to $(S,C,D)$, there exists an ascendant branch of $(S,C,D)$, denoted by $B_{asc}=(S_{asc},C_{asc},D_{asc})$, where $u_{equ}$ is selected as the branching vertex, since $\langle u_{equ},\phi(u_{equ}) \rangle$ is in $S$.
    We can see that there exists one sub-branch $B'_{asc}=(S'_{asc},C'_{asc},D'_{asc})$ of $B_{asc}$ formed by including $\langle u_{equ},v \rangle$, and all common subgraphs within $B'_{asc}$ have been found before the formation of $(S,C,D)$, since $\langle u_{equ},v \rangle$ is in $D$. We then show that common subgraph $S_{iso}$ can be found within $B_{asc}'$, i.e., $S'_{asc}\subseteq S_{iso}\subseteq S'_{asc}\cup C'_{asc}$. 
    \underline{First}, we have $S'_{asc}\subseteq S_{iso}$ since (1) $S'_{asc}=S_{asc}\cup\{\langle u_{equ},v \rangle\}$, (2) $S_{sub}$ is a common subgraph in $B_{asc}$ and thus $S_{asc}\subseteq S_{sub}$, (3) $S_{asc}$ does not include $\langle u_{equ},\phi(u_{equ})\rangle$ or $\langle u,v\rangle$ since they are in $C_{asc}$ and will be included to the partial solution at $B_{asc}'$ and $(S\cup \{\langle u,v \rangle\}, C\backslash u\backslash v)$, and thus (4) by combining all the above, we have  $S'_{asc}=S_{asc}\cup\{\langle u_{equ},v \rangle\}\subseteq S_{sub}\cup\{\langle u_{equ},v \rangle\} \backslash \{\langle u_{equ},\phi(u_{equ})\rangle,\langle u,v\rangle\}\subseteq S_{iso}$. 
    \underline{Second}, we have $S_{iso}\subseteq S'_{asc}\cup C'_{asc}$ based on the following two facts. 

    \begin{itemize}
        \item \textbf{Fact 1.} $S_{sub}\backslash\{\langle u_{equ},\phi(u_{equ})\rangle,\langle u,v\rangle\}\subseteq S'_{asc}\cup C'_{asc}$.
        \item \textbf{Fact 2.} $\langle u_{equ},v \rangle\in S'_{asc}$ and $\langle u,\phi(u_{equ}) \rangle\in C'_{asc}$.
    \end{itemize}
    
    The correctness of the above facts can be verified accordingly, for which we put the details in the 
    \ifx \CR\undefined
Appendix. 
\else
technical report~\cite{TR}. 
\fi
    }
\end{proof}

%\medskip
\noindent\fbox{%
    \parbox{0.47\textwidth}{%
       \textbf{Vertex-Equivalence-based reduction at the first group}. Let $B=(S,C,D)$ be a branch. For a sub-branch of $B$ formed by including a candidate pair $\langle u,v \rangle$ in the first group, it can be pruned if there exists a vertex pair $\langle u',v \rangle$ in $D$ such that $u'\in \Psi(u)$. 
    }%
}
%\medskip

\begin{example}
        Consider again the branching process at branch $B_6=(S_6,C_6,D_6)$ in Figure~\ref{fig:example_branching}. \laks{Suppose $u_2$ is selected as the branching vertex}. Then, we can see that $u_1$ is in $\Psi(u_2)$ and $D_{6}= \{u_1\} \times \{v_1,v_2,v_3,v_4,v_5\}$. {\YuiRR Recall that $C_6=\{u_2,u_3\}\times \{v_4,v_5\}\cup\{u_4,u_5,u_6,u_7\}\times\{v_1,v_2,v_3,v_7\}$. Thus, $B_6$ has two sub-branches which are formed by including $\langle u_2,v_4 \rangle$ or $\langle u_2,v_5 \rangle$, and they can be pruned based on the above reduction. }
\end{example}

\smallskip
\noindent\textbf{Reduction at the second group}. {\YuiR Recall that vertex $u$ is selected as the branching vertex.} Consider the sub-branch formed  {\YuiR by excluding $u$} in the second group. We note that {\chengC each} common subgraph $S_{sub}=\langle q_{sub},g_{sub},\phi_{sub} \rangle$ to be found in this sub-branch must exclude vertex $u$. We observe that if \emph{$S_{sub}$ contains a vertex $u_{equ}$, {\chengC which} is in $C\backslash u$ and is structurally equivalent to $u$}, Condition 1\&2 holds for $S_{sub}$. 

In specific, Condition 1 holds {\chengC since we can construct the} following common subgraph {\chengC which is} cs-isomorphic to $S_{sub}$ {\chengC based on} Case 2 of Lemma~\ref{lemma:cs} {\chengC (essentially, we replace $u_{equ}$ with $u$)}.
\begin{equation}
    \label{eq:iso2}
    S_{iso}= S_{sub}\backslash\{\langle u_{equ}, \phi_{sub}(u_{equ})\rangle\}\cup\{ \langle u,\phi_{sub}(u_{equ}) \rangle\}
\end{equation}

Besides, we note that $S_{iso}$ contains $\langle u,\phi_{sub}(u_{equ}) \rangle$ and common subgraphs found in the first group must include $u$. We thus {\chengC check} whether $S_{iso}$ has been found in one sub-branch formed in the first group. The answer is also positive. The rationale behind is that (1) $\langle u,\phi_{sub}(u_{equ}) \rangle$ constructed by ours exists in $C$ (this holds due to the vertex equivalence between $u$ and $u_{equ}$
% , which will be discussed in Lemma~\ref{lemma:reduction_for_VE2}
) and (2) thus there exists a sub-branch formed by including $\langle u,\phi_{sub}(u_{equ}) \rangle$ in the first group, where $S_{iso}$ has been found since it includes $\langle u,\phi_{sub}(u_{equ}) \rangle$. 

In summary, we have the following lemma and reduction rule.

\begin{lemma}
\label{lemma:reduction_for_VE2}
Let $(S,C,D)$ be a branch where $u$ is selected as the branching vertex. Common subgraph $S_{iso}$ {\chengC defined in} Equation~(\ref{eq:iso2}) has been found before the formation of $(S,C\backslash u,D)$ at the second group.   
\end{lemma}
\begin{proof}
{\revision
    \underline{First}, we note that $\langle u_{equ},\phi_{iso}(u_{equ}) \rangle$ is in $C\backslash u$ and also in $C$ since otherwise $S_{sub}$ cannot include $\langle u_{equ},\phi_{iso}(u_{equ}) \rangle$. This is because (1) $S\subseteq S_{sub}\subseteq S\cup C\backslash u$ since $S_{sub}$ is a common subgraph in the sub-branch $(S,C\backslash u,D)$ and (2) $S$ does not include $\langle u_{equ},\phi_{iso}(u_{equ}) \rangle$ since $u_{equ}$ appears in $C\backslash u$.
    \underline{Second}, we note that $\phi_{iso}(u_{equ})$ is in $Y$. Recall that $X\times Y$ is the branching set at $(S,C,D)$. This is because (1) $u_{equ}$ is in the same subset $X$ as $u$ since $u_{equ}$ and $u$ are structurally equivalent and thus have the same set of neighbours and non-neighbours in $q$, and (2) $\langle u_{equ},\phi_{iso}(u_{equ}) \rangle$ is in $C$ as discussed before.
    \underline{Third}, we can derive that there exists a sub-branch $(S\cup \{\langle u,\phi_{iso}(u_{equ}) \rangle\},C\backslash u\backslash \phi_{iso}(u_{equ}),D')$, which is formed at branch $(S,C,D)$  by including $\langle u,\phi_{iso}(u_{equ}) \rangle$ before the formation of $(S,C\backslash u,D)$, since $\phi_{iso}(u_{equ})\in Y$.
    \underline{Forth}, we show that $S_{iso}$ is in $(S\cup \{\langle u,\phi_{iso}(u_{equ}) \rangle\},C\backslash u\backslash \phi_{iso}(u_{equ}),D')$, formally, $S\cup \{\langle u,\phi_{iso}(u_{equ}) \rangle\} \subseteq S_{iso}\subseteq S\cup \{\langle u,\phi_{iso}(u_{equ}) \rangle\}\cup (C\backslash u\backslash \phi_{iso}(u_{equ}))$.
    We have $S\subseteq S_{sub}\subseteq S\cup (C\backslash u)$ since $S_{sub}$ is a common subgraph in $(S,C\backslash u,D)$. Let $S'=S\cup \{\langle u,\phi_{iso}(u_{equ}) \rangle\}$, it can be proved as below.
    \begin{eqnarray}
        && S\subseteq S_{sub}\subseteq S\cup (C\backslash u)\\
        && \Rightarrow S\subseteq S_{sub}\backslash\{\langle u_{equ}, \phi_{iso}(u_{equ})\rangle\}\subseteq S\cup (C\backslash u\backslash \phi_{iso}(u_{equ})) \label{eq:lemma_eq_1}\\
        && \Rightarrow S' \subseteq S_{iso}\subseteq S'\cup (C\backslash u\backslash \phi_{iso}(u_{equ})) \label{eq:lemma_eq_2}
    \end{eqnarray}
    Note that Equation~(\ref{eq:lemma_eq_1}) holds since $\langle u_{equ},\phi_{iso}(u_{equ}) \rangle$ is in $S$; Equation~(\ref{eq:lemma_eq_2}) is derived by including the vertex pair $\langle u, \phi_{iso}(u_{equ})\rangle$.}
\end{proof}

\noindent\fbox{%
    \parbox{0.47\textwidth}{%
       \textbf{Vertex-Equivalence-based reduction at the second group}. Let $B=(S,C,D)$ be a branch and $(S,C\backslash u,D)$ be  the sub-branch formed in the second group by excluding all candidate pairs that consist of $u$. For a vertex $u'$ appearing in $C\backslash u$, if $u'$ is structurally equivalent to $u$, i.e., $u'\in \Psi(u)$, all candidate pairs that consist of $u'$ can be pruned from  $C\backslash u$. 
    }%
}
%\medskip

\begin{example}
    Consider the branching process at $B_0$ where $u_1$ is the branching vertex in Figure~\ref{fig:example_branching}. For sub-branch $B_8$ (which is the sub-branch at the second group), we can see that $\Psi(u_1)=\{u_1,u_2\}$ and thus the candidate set of $B_8$ can be reduced to $ (V_Q\backslash\{u_1,u_2\})\times V_G$, {\YuiR i.e., to $\{u_3,u_4,u_5,u_6,u_7\}\times \{v_1,v_2,\cdots,v_7\}$}.
\end{example}

\subsection{Maximality-based Reduction}
\label{subsec:maximality-reduction}

We introduce the redundancies induced by \emph{non-maximality}. Clearly, a maximum common subgraph must be a maximal common subgraph. Therefore, exploring those branches that hold non-maximal common subgraphs only will incur redundant computations. 
Consider a current branch $B=(S,C,D)$. {\chengC Note that there might exist multiple common subgraphs with the largest number of vertices in the branch.}
We observe that \emph{there exists one largest common subgraph in $B$ that must contain one specific candidate vertex pair $\langle u,v \rangle$}. 
{\YuiR Given this, we can remove this candidate vertex pair $\langle u,v\rangle$ from $C$ to $S$, thereby producing one immediate sub-branch i.e., $(S\cup\{\langle u,v \rangle\},C\backslash u\backslash v)$. Clearly, solving the resulting sub-branch is enough to find the largest common subgraph (since it holds all those common subgraphs that contain $\langle u,v\rangle$ in $B$). As a result,  we can safely prune all other sub-branches for which the partial set and the candidate set do not include the candidate pair $\langle u,v\rangle$.}
%
%As a result, we can safely prune all other sub-branches \laks{for which} all common subgraphs \laks{in the sub-branch} \emph{exclude} these candidate vertex pairs. 
%
Below, we elaborate on the details.

To be specific, we observe that there exists one largest common subgraph, denoted by $S_{opt}$, in $B$ such that $S_{opt}$ must contain a candidate vertex pair $\langle u,v \rangle$ if for any subset $X\times Y$ in $\mathcal{P}(C)$, $u$ and $v$ are simultaneously adjacent or non-adjacent to all other vertices in $X$ and $Y$, respectively, i.e.,
\begin{eqnarray}
    \label{eq:condition}
    \forall X\times Y\in \mathcal{P}(C): {\revision\big(}N(u,X)=X\backslash\{u\} {\revision \text{ and }} N(v,Y)=Y\backslash\{v\}{\revision\big)} \text{ or }\nonumber\\  
        {\revision \big(}N(u,X)=\emptyset {\revision\text{ and }} N(v,Y)=\emptyset{\revision\big)},
\end{eqnarray}
Formally, we have the following lemma.

\begin{lemma}
\label{lemma:maximality}
    Let $B=(S,C,D)$ be a branch {\chengC and $\langle u,v \rangle$ be a candidate vertex pair that satisfies the condition in Equation~(\ref{eq:condition}).} There exists one largest common subgraph $S_{opt}$ in the branch $B$ such that $S_{opt}$ contains 
    {\chengC $\langle u,v \rangle$.}
    % a candidate vertex pair $\langle u,v \rangle$, if $\langle u,v \rangle$ satisfies the condition in Equation~(\ref{eq:condition}).
\end{lemma}

\begin{proof}
 {\revision
 This can be proved by construction. Let $S^*=(q^*,g^*,\phi^*)$ be one largest common subgraph to be found in $B$. Note that if $S^*$ contains the candidate vertex pair $\langle u,v \rangle$, we can finish the proof by constructing $S_{opt}$ as $S^*$. Otherwise, if $\langle u,v \rangle$ is not in $S^*$, we prove the correctness by constructing one largest common subgraph $S_{opt}$ to be found in $B$ that contains candidate vertex pair $\langle u,v \rangle$, i.e., $S\subseteq S_{opt} \subseteq S\cup C$,  $|S_{opt}|=|S^*|$ and $\langle u,v \rangle\in S_{opt}$. In general, there are four different cases, and the details can be found in the 
 \ifx \CR\undefined
Appendix. 
\else
technical report~\cite{TR}. 
\fi
 }
\end{proof}

Consider a branch $B=(S,C,D)$ where $ X^*\times Y^*$ in $\mathcal{P}(C)$ and $u^*$ in $X^*$ are selected as the branching subset and the branching vertex, {\chengC respectively}. 
Assume that there exists a vertex $v$ in $Y^*$ such that $\langle u^*,v \rangle$ satisfies the condition in Equation~(\ref{eq:condition}).
Based on the above lemma, there exists one largest common subgraph in the branch $B$ that contains candidate vertex pair $\langle u^*,v \rangle$. Therefore, we only need to form one sub-branch $(S\cup\{\langle u^*,v \rangle\},C\backslash u^*\backslash v,D\cup  \{u^*\}\times (Y^*\backslash\{v\}) )$ since other formed sub-branches will exclude the candidate vertex $\langle u^*,v \rangle$ from the found common subgraphs. We note that the exclusion set of the formed sub-branch can be updated by $D\cup  \{u^*\}\times (Y^*\backslash\{v\}) $ to enhance the pruning power of the proposed reduction at the first group. In summary, we obtain the following reduction.

%\begin{lemma}[Maximality-based reduction]
%    \label{lemma:maximality-reduction}
%    Let $B=(S,C,D)$ be a branch where $\langle X\times Y \rangle$ in $C$ and $u$ in $X$ are selected as the branching subset and the branching vertex. If there exists a candidate vertex pair $\langle u,v \rangle$ in the candidate set such that $\langle u,v \rangle$ satisfies Equation~(\ref{eq:condition}), only one sub-branch $(S\cup\{\langle u,v \rangle\},C\backslash u\backslash v,D\cup \langle \{u\}\times (Y\backslash\{v\}) \rangle)$ needs to be formed at $B$.
%\end{lemma}

\medskip
\noindent\fbox{%
    \parbox{0.47\textwidth}{%
       \textbf{Maximality-based reduction}. Let $B=(S,C,D)$ be a branch where $ X\times Y $ in $\mathcal{P}(C)$ and $u$ in $X$ are selected as the branching subset and the branching vertex. If there exists a candidate vertex pair $\langle u,v \rangle$ in the candidate set such that $\langle u,v \rangle$ satisfies Equation~(\ref{eq:condition}), only one sub-branch $(S\cup\{\langle u,v \rangle\},C\backslash u\backslash v,D\cup \{u\}\times (Y\backslash\{v\}))$ needs to be formed at $B$.
    }%
}

\begin{example}
Consider the branching at branch $B_6=(S_6,C_6,D_6)$ in Figure~\ref{fig:example_branching} where \laks{suppose} $u_2$ is selected as the branching vertex. Recall that $C_6=X_1\times Y_1 \cup X_2\times Y_2 =\{u_2,u_3\}\times \{v_4,v_5\} \cup \{u_4,u_5,u_6,u_7\}\times \{v_1,v_2,v_3,v_7\}$.
We note that $\langle u_2,v_5 \rangle$ satisfies Equation~(\ref{eq:condition}) since (1) $N(u_2,X_1)=X_1\backslash\{u_2\}$ and $N(v_5,Y_1)=Y_1\backslash\{v_5\}$ and (2) $N(u_2,X_2)=\emptyset$ and $N(v_5,Y_2)=\emptyset$. Therefore, we only need to explore one sub-branch $(S_6\cup\{\langle u_2,v_5\rangle\},C_6\backslash u_2\backslash v_5,D_6\cup\{\langle u_2,v_4 \rangle\})$, and other two sub-branches formed at $B_6$ can be pruned.
\end{example}

\subsection{Vertex-Equivalence-based {\chengB Upper Bound}}
\label{subsec:upper-bound}
Consider a current branch $(S,C,D)$ and the largest common subgraph $S^*$ seen so far. Clearly, we can terminate the branch $(S,C,D)$, if the upper bound on the size of common subgraphs to be found in the branch $(S,C,D)$ (or simply, the upper bound of $(S,C,D)$) is no larger than the size of $S^*$. The tighter the upper bound, the more branches we can prune. 
% To facilitate this pruning technique, we introduce the upper bound as below. 

\smallskip
\noindent\textbf{Existing upper bound.} Consider a common subgraph $S_{sub}$ to be found in the branch $(S,C,D)$. For a subset $X\times Y$ in $\mathcal{P}(C)$, we can derive
\begin{equation}
    |S_{sub}|\cap X\times Y \leq ub_{X,Y}:= \min\{|X|,|Y|\}
\end{equation}
since otherwise a common subgraph will contain two distinct vertex pairs $\langle u,v \rangle$ and $\langle u',v' \rangle$ such that $u=u'$ or $v=v'$ (which violates the definition of the bijection). Here, $ub_{X,Y}$ is the upper bound of the number of candidate pairs that are within $X\times Y$ and are in a common subgraph to be found in the branch $(S,C,D)$. Furthermore, since all subsets in $\mathcal{P}(C)$ are disjoint, 
% we can derive 
the {\chengC following} existing upper bound of branch $(S,C,D)$, denoted by $ub_{S,C}$~\cite{mccreesh2017partitioning}, 
{\chengC can be derived}.
% as below
\begin{equation}
    |S_{sub}| \leq ub_{S,C} := |S|+\sum_{ X\times Y \in \mathcal{P}(C)} ub_{X,Y}
\end{equation}

\smallskip
\noindent\textbf{Motivation.} We observe that the existing upper bound $ub_{X,Y}$ is not tight since some candidate vertex pairs in $X\times Y$ can be pruned from the candidate set $C$ {\chengC based on} the proposed vertex-equivalence-based reductions. In specific, for a candidate vertex pair $\langle u,v \rangle$, if there exists a vertex pair $\langle u',v \rangle$ in $D$ such that $u'\in \Psi(u)$, any common subgraph to be found within $(S,C,D)$ cannot include $\langle u,v \rangle$ and thus $\langle u,v \rangle$ can be pruned from the candidate set $C$. Note that this can be easily verified based on the proposed reduction at the first group. Below, we introduce our upper bound derived with the aid of the structural equivalence on vertices.

\smallskip
\noindent\textbf{New upper bound}. Consider a subset $X\times Y $ in $\mathcal{P}(C)$. Let $u$ be an arbitrary vertex in $X$. We partition $X$ and $Y$ as follows.
\begin{eqnarray}
 X_L=X\cap \Psi(u), X_R=X\backslash X_L\\
 Y_L=\{v\mid \langle u',v \rangle\in D, u'\in \Psi(u)\}, Y_R=Y\backslash Y_L,   
\end{eqnarray}
where $X_L$ consists of those vertices in $X$ that are structurally equivalent to $u$ and $Y_L$ consists of those vertices $v$ in $Y$ which appear  in a vertex pair $\langle u',v \rangle$ in $D$ where $u'\in \Psi(u)$. We then can partition $X\times Y$ as $X_L\times Y_L$, $X_L\times Y_R$, $ X_R\times Y_L$ and $X_R\times Y_R$. Clearly, all vertex pairs in $X_L\times Y_L$ can be pruned as discussed before. 
We note that (1) $S_{sub}$ contains at most $\min\{|X_R|,|Y|\}$ vertex pairs from $X_R\times Y_L$ and $X_R\times Y_R$ since otherwise there exists one vertex in $X_R\cup Y$ that appears in at least two distinct vertex pairs in $S_{sub}$ and thus $S_{sub}$ cannot be a common subgraph; and similarly (2) $S_{sub}$ contains at most $\min\{|X_L|,|Y_R|,\max\{|Y|-|X_R|,0\}\}$ vertex pairs from $X_L\times Y_R$ (note that the additional term $\max\{|Y|-|X_R|,0\}$ is used to ensure that the sum of $\min\{|X_R|,|Y|\}$ and $\min\{|X_L|,|Y_R|,\max\{|Y|-|X_R|,0\}\}$ is no larger than the existing upper bound $ub_{S,C}$). Therefore, $S_{sub}$ contains at most $ub_{X,Y,D}$ vertex pairs from $X\times Y$, where
\begin{equation}
    ub_{X,Y,D}:=\min\{|X_R|,|Y|\}+\min\{|X_L|,|Y_R|,\max\{|Y|-|X_R|,0\}\}.
\end{equation}
Then, we can derive our upper bound of a branch $(S,C,D)$, denoted by $ub_{S,C,D}$, i.e.,
\begin{equation}
    |S_{sub}|\leq ub_{S,C,D}:=|S|+\sum_{X\times Y\in \mathcal{P}(C)} ub_{X,Y,D}.
\end{equation}

In summary, we obtain our new upper bound $ub_{S,C,D}$ as {\chengC above}. It is {\chengC not difficult} to verify that our upper bound is tighter than the existing one, i.e., $ub_{S,C,D}\leq ub_{S,C}$: see Example~\ref{example:upper_bound} for an example where $ub_{S,C,D} < ub_{S,C}$.  
\begin{lemma}[Upper bound]
    \label{lemma:upper_bound}
    Let $(S,C,D)$ be a branch. All common subgraphs to be found in $(S,C,D)$ have the size at most $ub_{S,C,D}$.
\end{lemma}

\begin{example}
\label{example:upper_bound}
    Consider again the branching process at branch $B_6=(S_6,C_6,D_6)$ in Figure~\ref{fig:example_branching}. Recall that $C_6=X_1\times Y_1\cup X_2\times Y_2=\{u_2,u_3\}\times \{v_4,v_5\}\cup \{u_4,u_5,u_6,u_7\}\times \{v_1,v_2,v_3,v_7\}$ and $D_6=\{u_1\}\times\{v_1,v_2,...,v_5\}$. For $X_1\times Y_1$, based on $u_2$, 
    we have $X_{1L}=\{u_2\}$, $X_{1R}=\{u_3\}$, $Y_{1L}=\{v_4,v_5\}$ and $Y_{1R}=\emptyset$. Thus, we have $ub_{X_1,Y_1,D_6}=\min\{1,4\}+\min\{1,0,\max\{1,0\}\}=1$. For $X_2\times Y_2$, based on $u_4$, we have $X_{2L}=\{u_4\}$, $X_{2R}=\{u_5,u_6,u_7\}$, $Y_{2L}=\emptyset$ and $Y_{2R}=\{v_1,v_2,v_3,v_7\}$. Thus, we have $ub_{X_2,Y_2,D_6}=\min\{3,4\}+\min\{1,4,\max\{1,0\}\}=4$. Therefore, we have $ub_{S_6,C_6,D_6}=1+1+4=6$, which is smaller than the existing bound $ub_{S,C}=7$.
\end{example}

{\revision
\noindent\textbf{Remark.} Our new upper bound $ub_{S,C,D}$ varies {\chengE with} different choices of $u$ due to the partition of $X$ and $Y$. We can potentially obtain a tighter upper bound by exploring all possible choices of $u$. However, it {\chengE would} introduce a large amount of time costs, thus degrading the performance of \texttt{RRSplit}. Therefore, as a trade-off,  we randomly select $u$ when computing the upper bound.
}

\subsection{Summary and Analysis}
\label{subsec:summary}
\noindent\textbf{Summary.}
We summarize our algorithm, namely \texttt{RRSplit}, in Algorithm~\ref{alg:rrsplit}, which incorporates the newly proposed vertex-equivalence-based reductions, the maximality-based reduction and the vertex-equivalence-based upper bound. Specifically, \texttt{RRSplit} differs with \texttt{McSplit} in the following aspects. (1) It maintains one additional auxiliary data structure, namely exclusion set $D$, for each formed {\chengB branch}, which is initialized as 
the empty set and recursively updated as discussed. (2)  It {\chengB prunes} a branch $(S,C,D)$ if the newly proposed vertex-equivalence-based upper bound $ub_{S,C,D}$ is no larger than the \laks{largest common subgraph size}   seen so far, i.e., $|S^*|$ (Line 7). We remark that $ub_{S,C,D}$ is tighter than the existing one $ub_{S,C}$, i.e., $ub_{S,C,D}\leq ub_{S,C}$ and thus more branches can be pruned. (3) It creates only one sub-branch and prunes all others if the maximality-based reduction is triggered (Lines 9-11). (4) Based on the vertex-equivalence-based reduction, it prunes those sub-branches at the first group that hold all common subgraphs inside cs-isomorphic to the one found before (Lines 15-16), and refines the formed sub-branch at the second group by removing from the candidate set all those candidate vertex pairs consisting of a vertex in $\Psi(u)$ (Line 19).
We remark that our implementation of \texttt{RRSplit} in the experiments adopts the same heuristic policies for selecting branching subset $X\times Y$, branching vertex $u$ (Line 8) and vertex $v$ (Line 14) as \texttt{McSplit} {\chengB does}.
Besides, we can easily prove that \texttt{RRSplit} finds the maximum common subgraph based on our discussion above. Finally, we analyze the {space complexity and time complexity} of \texttt{RRSplit} as below.

\begin{algorithm}{}
\small
\caption{Our proposed algorithm: \texttt{RRSplit}}
\label{alg:rrsplit}
\KwIn{Two graphs $Q=(V_Q,E_Q)$ and $G=(V_G,E_G)$}
\KwOut{The maximum common subgraph}
$S^*\leftarrow \emptyset$; \tcp{Global variable}
%$S\leftarrow\emptyset$, $C\leftarrow\langle V_Q\times V_G \rangle$, $D\leftarrow\emptyset$ \tcp{Global data structure}
\texttt{RRSplit-Rec}$(\emptyset,V_Q\times V_G,\emptyset)$\; \textbf{Return} $S^*$;\\
\SetKwBlock{Enum}{Procedure \texttt{RRSplit-Rec}$(S,C,D)$}{}
%\SetKwBlock{update}{Procedure \texttt{Update}$(S,C)$}{}
\Enum{
    \lIf{$|S|>|S^*|$}{$S^*\leftarrow S$}
    \tcc{Termination (Lemma~\ref{lemma:upper_bound})}
    \lIf{$C=\emptyset$}{\textbf{return}}
    \lIf{ $ub_{S,C,D}\leq |S^*|$}
        {\textbf{return}} 
    \tcc{Branching}
    Select a branching vertex $u$ and a branching subset $X\times Y$ from $\mathcal{P}(C)$  based on a policy\;
    \tcc{Maximality-based reduction}
    \If{there exists a vertex $v$ in $Y$ such that $\langle u,v \rangle$ satisfies Equation~(\ref{eq:condition})}{
        \texttt{RRSplit-Rec}($S\cup\{\langle u,v \rangle\},C\backslash u\backslash v,D\cup \{u\}\times (Y\backslash\{v\})$)\;
        \textbf{return}\;
    }
    \tcc{Branching at the first group}
    $Y_{temp}\leftarrow Y$\;
    \For{$i=1,2,...,|Y|$}{
        Select and remove a vertex $v$ from $Y_{temp}$ based on a policy\;
        \If{there exists a vertex pair $\langle u',v \rangle$ in $D$ such that $u'\in \Psi(u)$}{\textbf{continue;}}
        Refine candidate set $C\backslash u\backslash v$ as $C_i$ based on Equation~(\ref{eq:update_candidate_set})\;
        \texttt{RRSplit-Rec}($S\cup\{\langle u,v\rangle\},C_i,D\cup\{u\}\times (Y\backslash Y_{temp})$);
    }
    \tcc{Branching at the second group}
    \texttt{RRSplit-Rec}($S,C\backslash \Psi(u),D$)\;
}
\end{algorithm}

\smallskip
\noindent\textbf{Space complexity}. We note that \texttt{RRSplit} recursively maintains three global data structures, namely $S$, $C$ and $D$, for each branch, which dominate the space complexity of \texttt{RRSplit}. Let $S^*$ be the largest common subgraph between  graphs $Q$ and $G$. \underline{First}, partial solution $S$ is a set of vertex pairs and its size is bounded by $O(|S^*|)$. \underline{Second}, candidate set $C$ is also a set of vertex pairs and can be partitioned as several subsets, i.e., $C= X_1\times Y_1 \cup  X_2\times Y_2\cup\cdots\cup X_c\times Y_c$ where $c$ is a positive integer, based on Equation~(\ref{eq:update_candidate_set}). We note that subsets in $X_1,X_2,...,X_c$ (resp. $Y_1,Y_2,...,Y_c$) are mutually disjoint and $X_1\cup X_2\cup .... \cup X_c=X$ (resp. $Y_1\cup Y_2\cup .... \cup Y_c=Y$), as discussed in the proof of Lemma~\ref{lemma:reduction_for_VE1}. Therefore, $C$ can be stored {\chengC as} $c$ subsets, each of which $\langle X_i,Y_i\rangle$ ($1\leq i\leq c$) consists of two sets $X_i$ and $Y_i$. Thus, the size of $C$ is bounded by $O(|V_Q|+|V_G|)$. \underline{Third}, $D$ is a set of vertex pairs and consists of at most $|S^*|\cdot |V_G|$ different vertex pairs since for a vertex pair $\langle u,v \rangle$ in $D$, (1) $u$ must {\chengB appear} in $S$ based on our maintenance of $D$ and thus has at most $|S^*|$ different values and (2) $v$ has at most $|V_G|$ different values clearly. In summary, the space complexity of \texttt{RRSplit} is $O(|V_Q|+|S^*|\times|V_G|)$. %We remark that \texttt{McSplit} has the space complexity of $O(|V_Q|+|S^*|\times|V_G|)$, which is the same as that of \texttt{RRSplit}, since \texttt{McSplit} needs to maintain the set $Y_{temp}$ for each branch.

\smallskip
\noindent\textbf{Time complexity of the proposed reductions}. \underline{First}, the reduction at the first group takes $O(|V_Q|+|V_G|)$ {\chengB time} (Lines 15-16). In specific, $D$ is organized as several disjoint subsets, i.e., $D=\{u_1\} \times A_1 \cup  \{u_2\} \times A_2\cup \cdots \cup \{u_d\} \times A_d $ where $d$ is a positive integer. 
Thus, it can be conducted in two steps: (1) for each vertex $u_i$ {\chengC appearing} in $D$, it takes $O(1)$ to check whether $u_i\in \Psi(u)$ and (2) if $u_i\in \Psi(u)$, it takes $O(|A_i|)$ to check whether $\langle u_i,v \rangle\in  \{u_i\} \times A_i$.
We note that for any two distinct vertices $u_i$ and $u_j$ appearing in $D$ such that $u_i\in \Psi(u_j)$, it is no hard to verify that $A_i\cap A_j=\emptyset$ due to the reduction at the first group (for which we put the details of the proof in the 
\ifx \CR\undefined
Appendix\else technical report~\cite{TR}\fi). 
As a result, we have $\sum_{u_i\in \Psi(u)} (|A_i|)\leq |V_G|$.  
\underline{Second}, the reduction at the second group runs in $O(|X|)$ for updating $C\backslash\Psi(u)$ at Line 19, which is bounded by $O(|V_Q|)$. In specific, it can be done by removing from $X$ all vertices in $\Psi(u)$ (note that, given all structurally equivalent classes, determining whether a vertex belongs to $\Psi(u)$ can be done in $O(1)$).
\underline{Third}, the maximality reduction runs in $O(\sum_{\langle X',Y' \rangle\in C}|X'|+|Y'|\cdot |Y|)$, which is bounded by $O(|V_Q|+|V_G|^2)$. In specific, for each vertex in $|Y|$, it needs to check the condition in Equation~(\ref{eq:condition}).
{\revision \underline{Fourth}}, the new upper bound can be obtained in $O(|V_Q|+|V_G|^2)$. In specific, the time cost is dominated by the computation of $ub_{X',Y',D}$ for each subset $\langle X',Y'\rangle$ in $C$. $ub_{X',Y',D}$ can be obtained in $O(|X'|+\sum_{u_i\in\Psi(u')} |A_i|+|Y'|)$, where $u'$ is a random vertex selected from $X'$ and $\{u_i\}\times A_i$ is a subset in $D$, which is bounded by $O(|X'|+|V_G|)$. Therefore, the new upper bound can be obtained in $O(\sum_{X'\times Y'\in C} (|X'|+|V_G|))$, which is bounded by $O(|V_Q|+|V_G|^2)$.

\smallskip
\noindent\textbf{Worst-case time complexity of \texttt{RRSplit}.} We note that the worst-case time complexity of \texttt{RRSplit} is dominated by the number of recursive calls of \texttt{RRSplit-Rec} (i.e., the number of formed branches) since \texttt{RRSplit-Rec} runs in polynomials of $|V_Q|$ and $|V_G|$. Formally, we have the following theorem.
\begin{theorem}
    Assume that $|V_Q|\leq |V_G|$. The worst-case time complexity of our proposed \texttt{RRSplit}  is $O^*((|V_G|+1)^{|V_Q|})$, where $O^*(\cdot)$ suppresses the polynomials.
\end{theorem}
\begin{proof}
    It is easy to verify that the worst-case time complexity of \texttt{RRSplit} is bounded by the number of branches. Consider a branch $B=(S,C,D)$. For all sub-branches formed at $B$ by selecting a branching vertex $u^*$, we observe that only the sub-branch in the second group has the same partial solution $S$ with $B$. Based on this, we can easily deduce that there are at most $|V_Q|$ branches which share the same partial solution. Besides,  we observe that each vertex in $V_p\cup V_q$ only appears in one pair of $S$, i.e., for any two distinct pairs $\langle u,v \rangle$ and $\langle u',v' \rangle$ in $S$, we have $u\neq u'$ and $v\neq v'$. Based on this, let $|S|=k$ where $0\leq k\leq |V_Q|$, and we can deduce that there are at most $k!\binom{|V_Q|}{k}\binom{|V_G|}{k}$ different partial solutions with the size of $k$ by applying the multiplication principle (note that $V_p$ has $\binom{|V_Q|}{k}$ different choices, $V_q$ has $\binom{|V_G|}{k}$ different choices, and the bijection $\phi$ between $V_p$ and $V_q$ has $k!$ different choices). Therefore, the number of branches is at most
    \begin{equation}
       T=|V_Q| \sum_{k=0}^{|V_Q|} k! \binom{|V_Q|}{k}\binom{|V_G|}{k}.
    \end{equation}
    We then show that $T$ is bounded by $O^*((|V_G|+1)^{|V_Q|})$ as below.
    \begin{eqnarray}
       T&=&|V_Q| \sum_{k=0}^{|V_Q|} (|V_Q|-k)! \binom{|V_Q|}{k}\binom{|V_G|}{|V_Q|-k}\\
       &=&|V_Q| \sum_{k=0}^{|V_Q|} \frac{(|V_{G}|)!}{(|V_G|-|V_Q|+k)!}\binom{|V_Q|}{k}\\
       &\leq& |V_Q|\sum_{k=0}^{|V_Q|} (|V_G|)^{|V_Q|-k} \binom{|V_Q|}{k}=|V_Q|(|V_G|+1)^{|V_Q|},
    \end{eqnarray}
    where $(|V_{G}|)!/(|V_G|-|V_Q|+k)!$ is  much smaller than $(|V_G|)^{|V_Q|-k}$ clearly and $(|V_G|+1)^{|V_Q|}$ in the last equation is derived by the binomial theorem.
\end{proof}

\smallskip
\noindent\textbf{Remark.} \laks{Note that the assumption that $|V_Q|\leq |V_G|$ is not a restrictive assumption: it is realistic in practice.} We remark that to our best knowledge, the achieved worst-case time complexity $O^*((|V_G|+1)^{|V_Q|})$ of \texttt{RRSplit} \laks{matches} 
% the state-of-the-art
{\chengB the best-known worst-case time complexity for the problem}
~\cite{suters2005new}. However, the algorithm proposed in~\cite{suters2005new} is of theoretical {\chengB interest} only and is not {\chengB practically} efficient. Besides, we note that \texttt{McSplit} and its variants~\cite{zhoustrengthened,liu2020learning,liu2023hybrid,mccreesh2017partitioning} do not have any theoretical guarantees on the worst-case time complexity. 
% Therefore, \texttt{RRSplit} is efficient in both theory and practice.

\section{Experiments}
\label{sec:exp}
\noindent\textbf{Datasets.} {\YuiR Following existing studies~\cite{liu2023hybrid,zhoustrengthened,liu2020learning,mccreesh2017partitioning,solnon2015complexity,hoffmann2017between},} we use four benchmark graph collections, namely biochemicalReactions (\textsf{BI}), images-CVIU11 (\textsf{CV}), images-PR15 (\textsf{PR}) and LV (\textsf{LV}), in the experiments. All datasets are collected from http://liris.cnrs.fr/csolnon/SIP.html and come from real-world applications in various domains, {\Yui as shown in Table~\ref{tab:my_label}}. Specifically, \textsf{BI} contains 136 unlabeled bipartite graphs, each of which corresponds to a biochemical reaction network. \textsf{CV} contains 44 pattern graphs and 146 target graphs, which are generated from segmented images. \textsf{PR} contains 24 pattern graphs and 1 target graph, which are also from segmented images. \textsf{LV} contains 112 graphs generated from biological networks. 
{\YuiR All graphs have up to thousands of vertices. We note that (1) solving our problem on two graphs with beyond 10K vertices is challenging based on the worst-case time complexity of $O^*((|V_G|+1)^{|V_Q|})$, (2) the largest graph used in previous studies~\cite{liu2023hybrid,zhoustrengthened,liu2020learning,mccreesh2017partitioning} has 6,771 vertices, which is also covered (in LV) by our experiments, and (3) finding the largest common subgraph between two graphs with thousands of vertices has found many real applications~\cite{ehrlich2011maximum}.}
{\Yui Following existing studies~\cite{liu2023hybrid,zhoustrengthened,liu2020learning,mccreesh2017partitioning,solnon2015complexity,hoffmann2017between}}, for \textsf{BI} and \textsf{LV}, we generate and test the problem instances (i.e., $Q$ and $G$) by pairing any two distinct graphs; and for \textsf{CV} and \textsf{PR} {\revision which consist of two types of graphs, namely pattern graphs and target graphs}, we test all those problem instances with one graph $Q$ from pattern graphs and the other $G$ from target graphs.

\begin{table*}[]
    \centering
    \caption{\Yui Datasets used in the experiments (``\# of solved instances'' refers to the number of instances solved by algorithms within 1,800 seconds and ``Achieved speedups'' refers to the percentage of the solved instances that \texttt{RRSplit} runs at least 5$\times$/10$\times$/100$\times$ faster than \texttt{McSplitDAL})}
    \vspace{-0.15in}
    \begin{tabular}{|c|c|c|c|c|c|c|c|c|c|c|}
        \hline
        \multirow{2}{*}{Dataset} & \multirow{2}{*}{Domain} & \multirow{2}{*}{\# of graphs} & \multirow{2}{*}{\# of instances} & \multirow{2}{*}{\# of vertices} & \multicolumn{2}{c|}{\# of solved instances} & \multicolumn{3}{c|}{Achieved speedups} \\
        \cline{6-10}
        & & & & & \texttt{RRSplit} & \texttt{McSplitDAL} & 5$\times$ & 10$\times$ & 100$\times$\\
        \hline
        \textsf{BI} & Biochemical & 136 & 9,180 & 9$\sim$ 386 & 7,730 & 4,696 & 91.3\% & 84.4\% & 69.7\% \\
        \textsf{CV} & Segmented images & 190 & 6,424 & 22$\sim$ 5,972 & 1,351 & 1,291& 76.5\% & 48.6\% & 0.2\% \\
        \textsf{PR} & Segmented images & 25& 24& 4$\sim$ 4,838  & 24 & 24 & 91.7\% & 91.7\% & 58.3\% \\
        \textsf{LV} & Synthetic & 112 & 6,216& 10$\sim$ 6,671 & 1,059 & 883 & 68.0\% & 54.7\% & 38.3\%\\
        \hline
    \end{tabular}
    
    \label{tab:my_label}
\end{table*}

\begin{table*}[]
    \centering
    \caption{\YuiR Comparison of running time on all datasets (statistics of achieved speedups in Figure~\ref{fig:all_datasets_T})}
    \vspace{-0.15in}
    \begin{tabular}{|c|c|c|c|c|c|c|}
        \hline
        \multirow{2}{*}{Dataset} & \multicolumn{3}{c|}{\texttt{RRSplit} runs faster} & \multicolumn{3}{c|}{\texttt{McSplitDAL} runs faster} \\
        \cline{2-7}
        & \% of instances & Avg. speedup & Max. speedup & \% of instances & Avg. speedup & Max. speedup\\
        \hline
        BI& 99.43\% & 3.3$\times 10^4$ & $10^6$ & 0.5\% & 24.81 & 872.37 \\
        CV& 92.15\% & 10.92 & 161 & 7.84\% & 4.96 & 38.97 \\
        PR& 95.83\% & 139.39 & 234 & 4.17\% & 1.23 & 1.23 \\
        LV& 93.48\% & 1.2$\times 10^4$ & $10^6$ & 6.51\% & 24.23 & 652.13 \\
        \hline
    \end{tabular}
    
    \label{tab:results}
\end{table*}

\smallskip
\noindent\textbf{Algorithms.} We compare the newly proposed algorithm \texttt{RRSplit} with \texttt{McSplitDAL}~\cite{liu2023hybrid}. To be specific, \texttt{McSplitDAL} is one variant of \texttt{McSplit} as introduced in Section~\ref{sec:sota}, which follows the framework of \texttt{McSplit} (i.e., Algorithm~\ref{alg:mcsplit}) and introduces some learning-based techniques for optimizing the policies of selecting vertices at line 6, line 8 and line 10 of Algorithm~\ref{alg:mcsplit}. To our best knowledge, \texttt{McSplitDAL} is the state-of-the-art algorithm and runs significantly faster than previous solutions, including \texttt{McSplitLL}~\cite{zhoustrengthened} and \texttt{McSplitRL}~\cite{liu2023hybrid}. Besides these, in order to study the effectiveness of different reductions employed in our algorithm \texttt{RRSplit}, we evaluate three variants of  \texttt{RRSplit} --  
{\YuiR \texttt{RRSplit-VE}, \texttt{RRSplit-MB}, and \texttt{RRSplit-UB}, respectively obtained by turning off vertex-equivalence based reductions, maximality based reductions,  and  vertex-equivalence based upper bound}. 
%namely \texttt{RRSplit-MR} and \texttt{RRSplit-VER},

\smallskip
\noindent\textbf{Implementation and metrics.} All algorithms are implemented in C++ and compiled with -O3 optimization. All experiments run on a Linux machine with a 2.10GHz Intel CPU and 128GB memory. Note that, for the implementation of \texttt{McSplitDAL}, we directly use the source code from the authors of~\cite{liu2023hybrid}. We record and compare the total running times of the algorithms on different problem instances (note that the measured running time excludes the I/O time of reading graphs from the disk). We set the running time limit (INF) as 1,800 seconds by default. Our data and code are available at https://github.com/KaiqiangYu/SIGMOD25-MCSS. 

\subsection{Comparison among algorithms}

\begin{figure}[]
		\subfigure[\textsf{BI}]{
			\includegraphics[width=4.0cm]{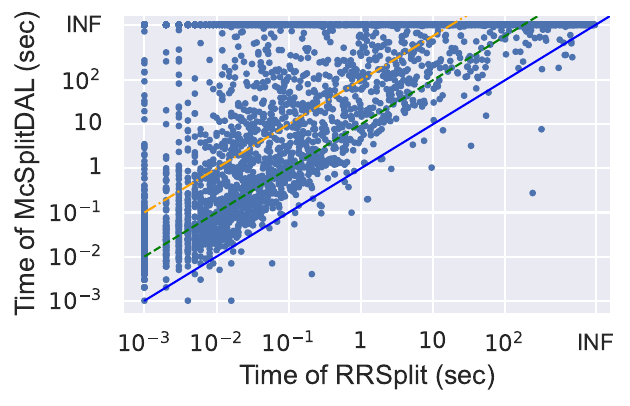}
		}
		\subfigure[\textsf{CV}]{
			\includegraphics[width=4.0cm]{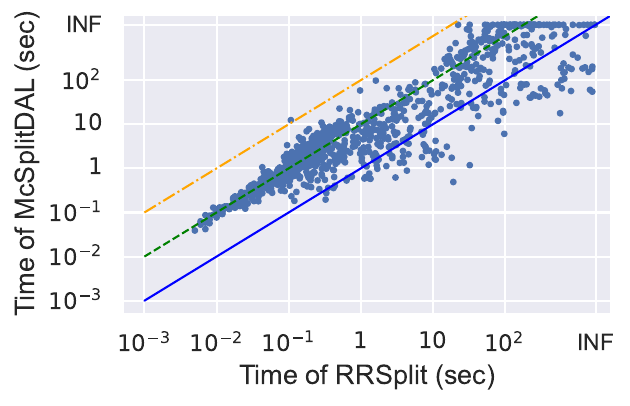}
		}
		\subfigure[\textsf{PR}]{
			\includegraphics[width=4.0cm]{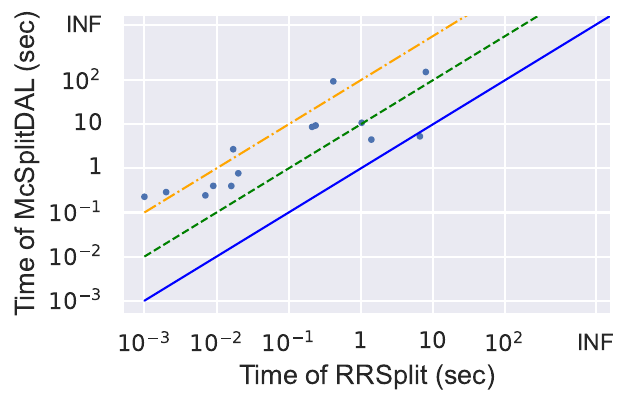}
		}	
		\subfigure[\textsf{LV}]{
			\includegraphics[width=4.0cm]{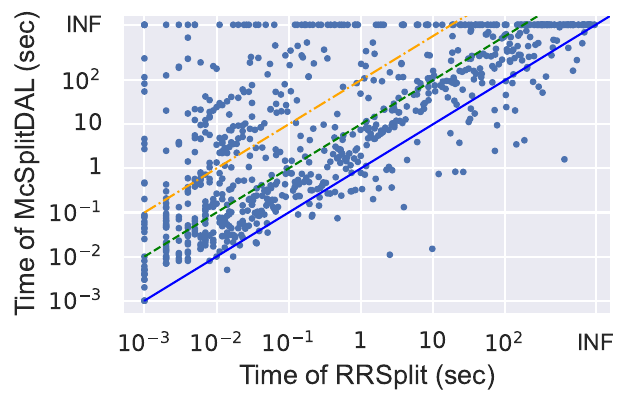}
		}
        \vspace{-0.15in}
	\caption{Running time on all datasets. {\Yui For those problem instances locating at the right side of dash line `- .' with orange color (resp. `- -' with green color),  \texttt{RRSplit} achieves at least 100$\times$ (resp. 10$\times$) speedup compared with \texttt{McSplitDAL}.}}
	\label{fig:all_datasets_T}
\end{figure}

\noindent\textbf{All datasets (running time)}. We compare our algorithm \texttt{RRSplit} with the baseline \texttt{McSplitDAL} on all graph collections. {\YuiR Following some existing works~\cite{mccreesh2016clique}}, we report the running times of the algorithms on various problem instances in Figure~\ref{fig:all_datasets_T}. 
Specifically, each dot in the scatter figures represents a problem instance, with the $x$-axis (resp. $y$-axis) corresponding to the running time of \texttt{RRSplit} (resp. \texttt{McSplitDAL}) {\chengC on the instance}. Hence, for those problem instances with small values on $x$-axis and large values on $y$-axis (which thus locate on the top left region of the figures), \texttt{RRSplit} performs better than \texttt{McSplitDAL}.
%In particular, for those problem instances locating at the right side of dash line `- .' {\Yui with orange color} (resp. `- -' with green color),  \texttt{RRSplit} achieves at least 100$\times$ (resp. 10$\times$) speedup compared with \texttt{McSplitDAL}. 
We mark the running time as INF if the problem instance cannot be solved within the default time limit.
{\YuiR Besides, we also provide some statistics in Table~\ref{tab:my_label} and Table~\ref{tab:results}.}
We observe that (1) \texttt{RRSplit} outperforms \texttt{McSplitDAL} by achieving around one to {\Yui four} orders of magnitude speedup {\YuiR (in average)} on the majority {\YuiR (above 92\%)} of the tested problem instances and (2) \texttt{McSplitDAL} cannot handle all problem instances within the time limit. 
% This fact demonstrates the efficiency of our algorithm \texttt{RRSplit}. 
We do note that \texttt{McSplitDAL} runs slightly faster on a few {\YuiR (below 8\%)} problem instances in \textsf{CV} and \textsf{LV}. {\YuiR Some possible reasons are as follows. 
First, our \texttt{RRSplit} introduces some extra time costs for conducting the proposed reductions as well as computing the upper bound. Second, the heuristic polices adopted in \texttt{RRSplit} and \texttt{McSplitDAL} for branching may have different behaviors. 
In specific, on these problem instances, the heuristic policies may help \texttt{McSplitDAL} to find a large common subgraph quickly so as to prune more unpromising branches {\revision via the upper-bound based reduction} (note that they are based on reinforcement learning {\cheng and the behaviors of the learned policy is} based on the explored branches during the running time). } 

%{\cheng One possible reason could be that} their learned heuristic policies can help to find a large common subgraph quickly so as to prune more unpromising branches {\cheng on these datasets}. 

\begin{figure}[]
		\subfigure[\textsf{BI}]{
			\includegraphics[width=4.0cm]{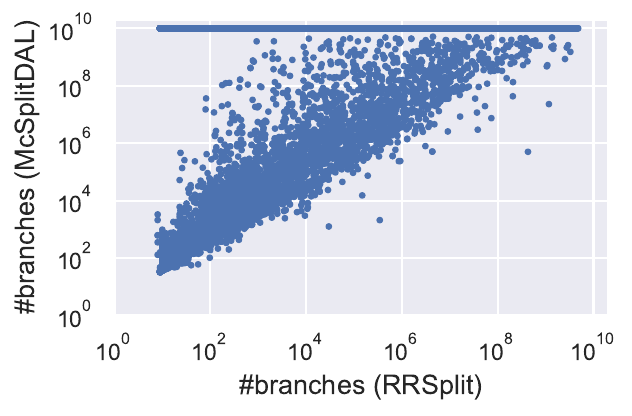}
		}
		\subfigure[\textsf{CV}]{
			\includegraphics[width=4.0cm]{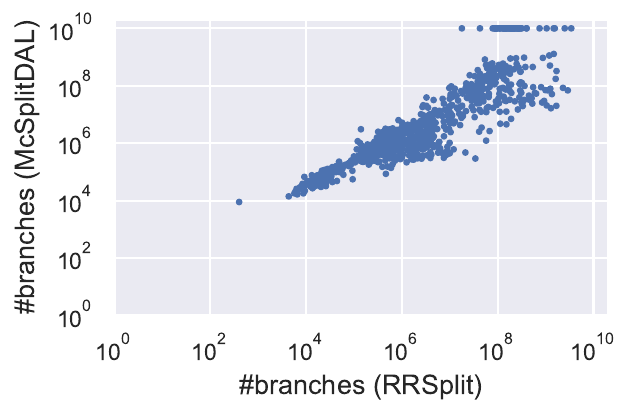}
		}
		\subfigure[\textsf{PR}]{
			\includegraphics[width=4.0cm]{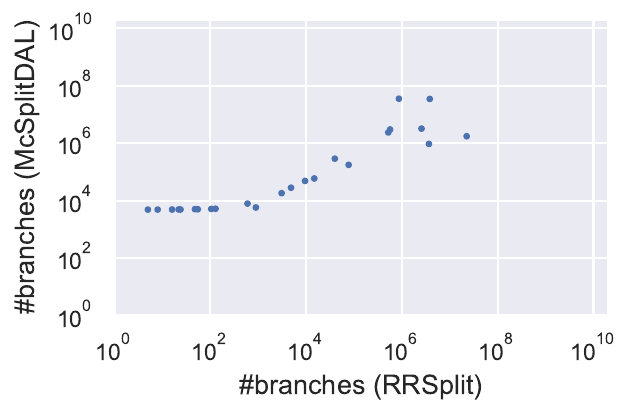}
		}	
		\subfigure[\textsf{LV}]{
			\includegraphics[width=4.0cm]{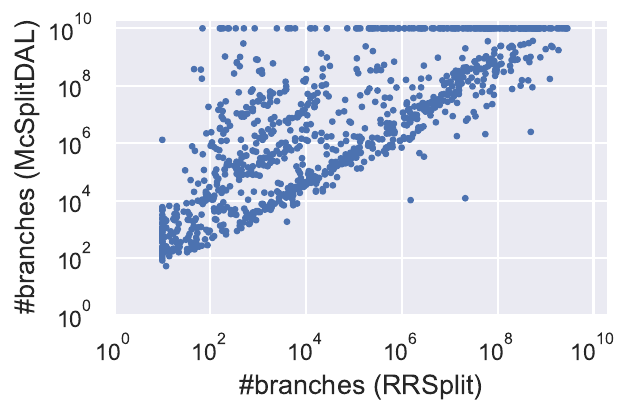}
		}
        \vspace{-0.15in}
	\caption{Number of formed branches on all datasets}
	\label{fig:all_datasets_BT}
\end{figure}

\smallskip
\noindent\textbf{All datasets (number of formed branches)}. We report the number of branches formed by the algorithms on different problem instances in Figure~\ref{fig:all_datasets_BT}. Similarly, each dot in the scatter figures represents a problem instance, with the $x$-axis (resp. $y$-axis) corresponding to the number of branches formed by \texttt{RRSplit} (resp. \texttt{McSplitDAL}) {\chengC on the instance}. We have the following observations. First, the number of branches formed by \texttt{RRSplit} is significantly {\chengC smaller} than that formed by \texttt{McSplitDAL}, e.g., the former is around 10\% - 0.01\% of the latter on the most of problem instances. This shows the effectiveness of our proposed maximality-based reductions and vertex-equivalence-based reductions.
%, and is also compatible with the theoretical results.
Second, the distribution of the number of formed branches in Figure~\ref{fig:all_datasets_BT} is consistent with that of the running time in Figure~\ref{fig:all_datasets_T}. This indicates the achieved speedups on the running time \laks{can be traced} to our newly-designed reductions.

\begin{figure}[]
		\subfigure[\textsf{BI}]{
			\includegraphics[width=4.0cm]{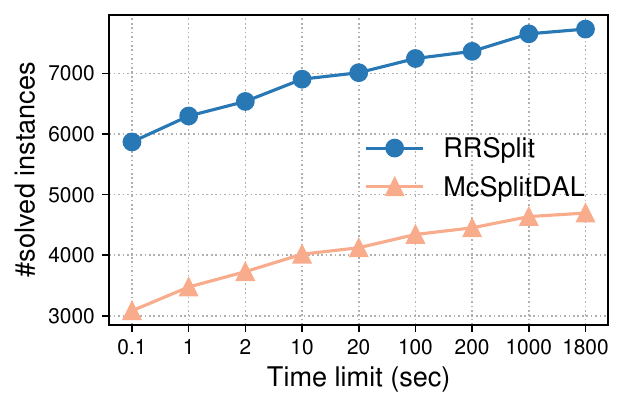}
		}
		\subfigure[\textsf{CV}]{
			\includegraphics[width=4.0cm]{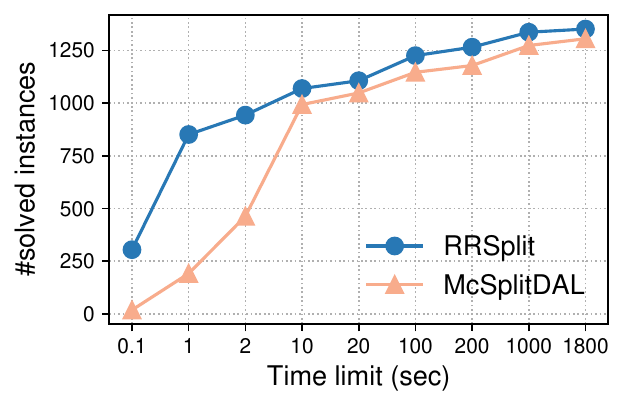}
		}
		\subfigure[\textsf{PR}]{
			\includegraphics[width=4.0cm]{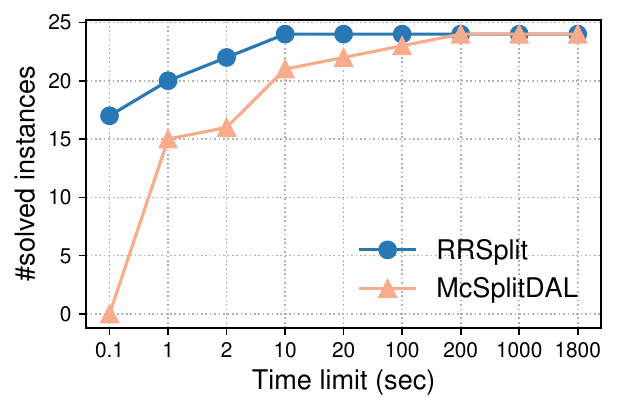}
		}	
		\subfigure[\textsf{LV}]{
			\includegraphics[width=4.0cm]{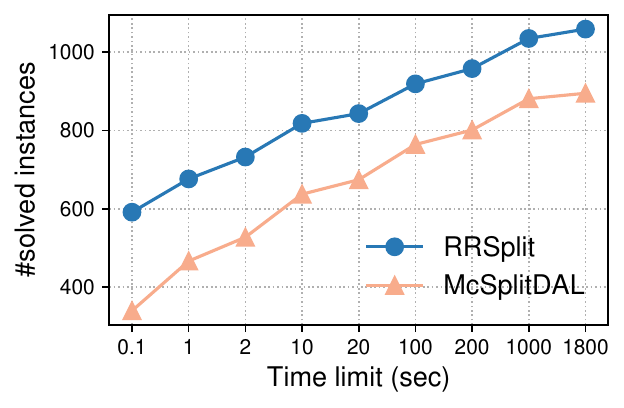}
		}
        \vspace{-0.2in}
	\caption{Comparison by varying time limits}
	\label{fig:all_vary_T}
\end{figure}

\begin{figure}[]
		\subfigure[\textsf{BI}]{
			\includegraphics[width=4.0cm]{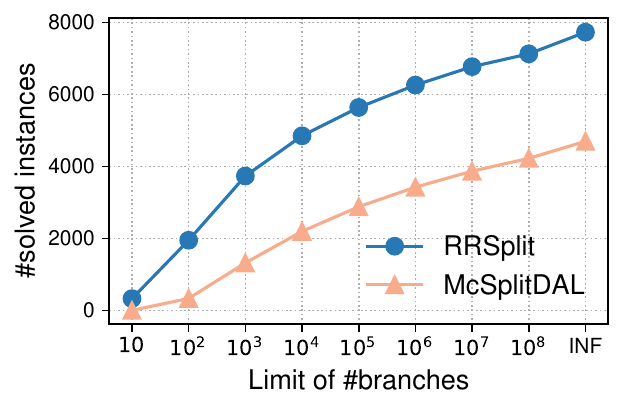}
		}
		\subfigure[\textsf{CV}]{
			\includegraphics[width=4.0cm]{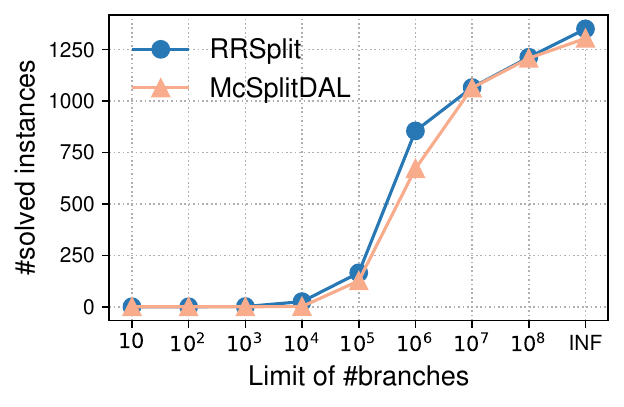}
		}
		\subfigure[\textsf{PR}]{
			\includegraphics[width=4.0cm]{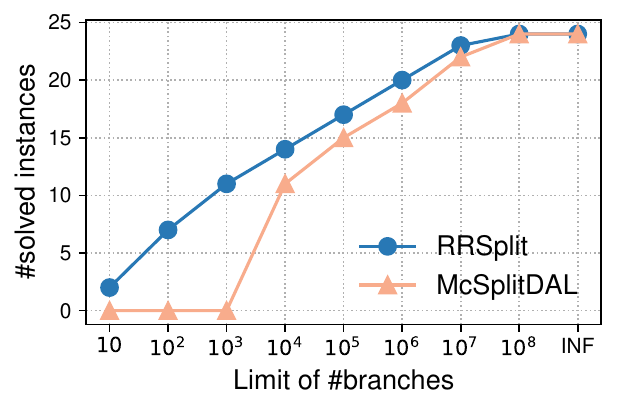}
		}	
		\subfigure[\textsf{LV}]{
			\includegraphics[width=4.0cm]{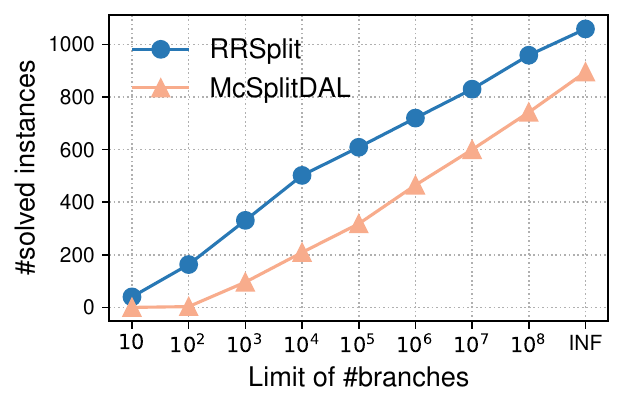}
		}
        \vspace{-0.2in}
	\caption{Comparison by varying the limit of number of formed branches}
	\label{fig:all_vary_B}
\end{figure}

\smallskip
\noindent\textbf{Varying time limits}. We report the number of solved problem instances in Figure~\ref{fig:all_vary_T} as the time limit is varied. Clearly, all algorithms solve more problem instances as the time limit increases. We observe that \texttt{RRSplit} solves more problem instances than \texttt{McSplitDAL} within the same time limit. In particular, \texttt{RRSplit} with a time limit of 1 second even solves more problem instances than \texttt{McSplitDAL} with a time limit of 10 seconds in all graph collections {\cheng except for} \textsf{CV}; and on \texttt{PR}, \texttt{RRSplit} solves all problem instances within the time limit of 10 seconds. This further demonstrates the superiority of our algorithm \texttt{RRSplit} over the baseline \texttt{McSplitDAL}. 

\smallskip
\noindent\textbf{Varying the limits of number of formed branches}. We report the number of solved problem instances in Figure~\ref{fig:all_vary_B} as the limit on  number of formed branches is varied. We note that the more branches are allowed to be formed, the more instances will be solved. We observe that (1) \texttt{RRSplit} solves more problem instances than \texttt{McSplitDAL} within the same limit of the number of formed branches and (2) the results in Figure~\ref{fig:all_vary_B} show  similar tendencies as those in Figure~\ref{fig:all_vary_T} in general. This further {\cheng explains} the practical superiority of the newly proposed reductions.

\begin{figure}[]
		\subfigure[\textsf{Running time (BI)}]{
			\includegraphics[width=4.0cm]{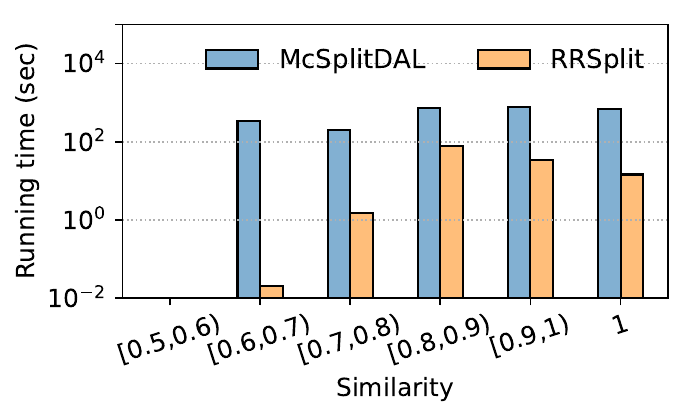}
		}	
		\subfigure[\textsf{Running time (LV)}]{
			\includegraphics[width=4.0cm]{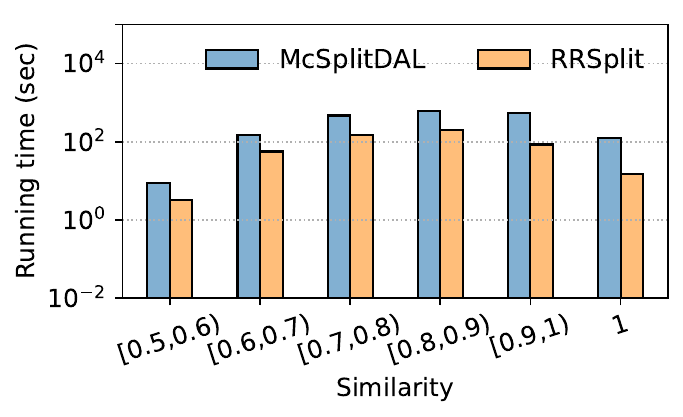}
		}
        \subfigure[\textsf{\# of branches (BI)}]{
			\includegraphics[width=4.0cm]{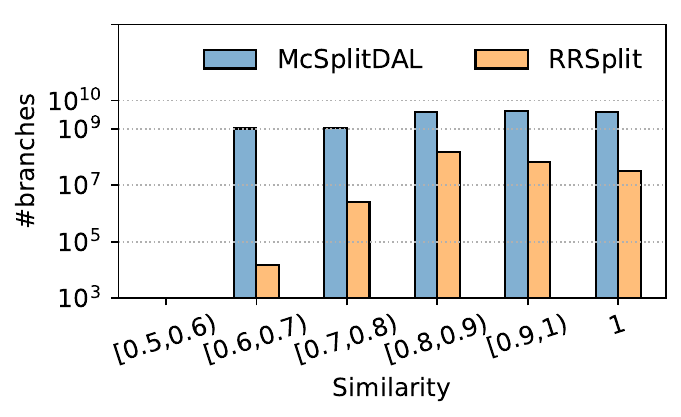}
		}	
		\subfigure[\textsf{\# of branches (LV)}]{
			\includegraphics[width=4.0cm]{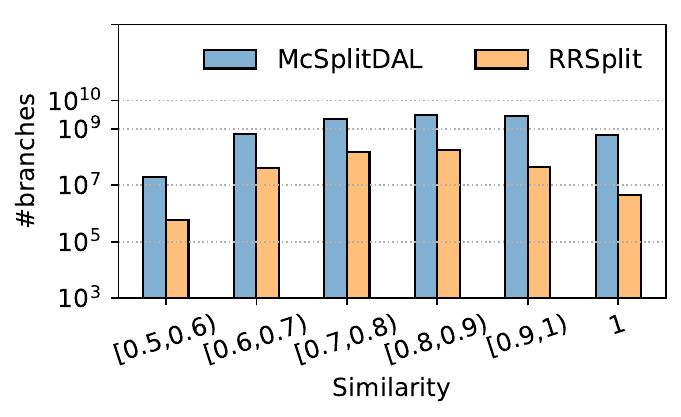}
		}
        \vspace{-0.2in}
	\caption{Comparison by varying similarities}
	\label{fig:all_vary_S}
\end{figure}

\smallskip
\noindent\textbf{Varying the similarities of  input graphs}. We define the similarity of  input graphs $Q$ and $G$, $Sim(Q,G)$, as follows.
\begin{equation}
\label{eq:sim}
    Sim(Q,G)=\frac{|S^*|}{\min\{|V_Q|,|V_G|\}},
\end{equation}
where $S^*$ is the maximum common subgraph between $Q$ and $G$. Clearly, $Sim(Q,G)$ varies from 0 to 1, and the larger the value of $Sim(Q,G)$, the higher the similarity between $Q$ and $G$. We test different problem instances as the similarity varies from 0.5 to 1 on \textsf{BI} and \textsf{LV}, and report the average running time in Figures~\ref{fig:all_vary_S}(a)-(b) and the average number of formed branches in Figures~\ref{fig:all_vary_S}(c)-(d). {\Yui The results on \textsf{CV} and \textsf{PR} show similar trends, complete details of which appear in the 
\ifx \CR\undefined
Appendix. 
\else
technical report~\cite{TR}. 
\fi
}  We can see that \texttt{RRSplit} consistently outperforms \texttt{McSplitDAL} {\chengC in} various settings, e.g., \texttt{RRSplit} runs several orders of magnitude faster and forms fewer branches than \texttt{McSplitDAL}. This demonstrates that our designed reductions are effective for pruning the redundant branches on problem instances with various similarities. Besides, we observe that both \texttt{RRSplit} and \texttt{McSplitDAL} have the running time and the number of formed branches first increase and then decrease as the similarity grows. {\revision The possible reasons are as follows. (1) The maximum common subgraphs become larger as the similarity increases according to Equation~(\ref{eq:sim}) and typically more common subgraphs will be explored for finding a large maximum common subgraph. Therefore, the running time firstly increases; (2) the upper-bound based reduction {\chengE performs} better as the similarity grows. {\chengE For example, in the setting of} $Sim(Q,G)=1$, the algorithm can be terminated directly once a common subgraph with $\min\{|V_Q|,|V_G|\}$ vertices is found. Therefore, the running time then decreases.}

%Possible reasons include (1) the number of common subgraphs (i.e., search space) first increases and then decreases as the similarity grows and/or (2) the proposed reductions performs better on those problem instances with the similarity {\cheng close to} 0.5 or 1.   

\begin{figure}[]
		\subfigure[\textsf{\Yui Varying time limits (BI)}]{
			\includegraphics[width=4.0cm]{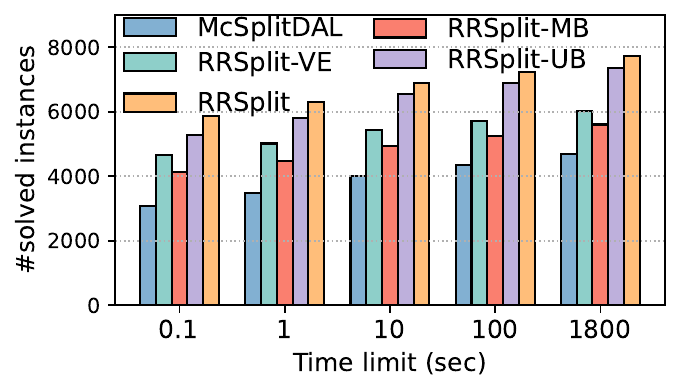}
		}	
		\subfigure[\textsf{\Yui Varying time limits (LV)}]{
			\includegraphics[width=4.0cm]{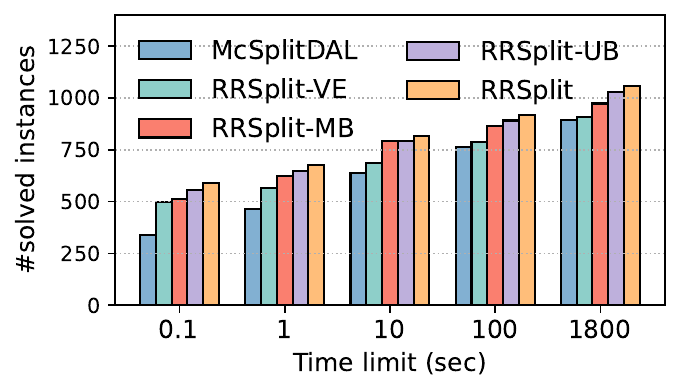}
		}
        \subfigure[\textsf{\Yui Varying limit of \#branches (BI)}]{
			\includegraphics[width=4.0cm]{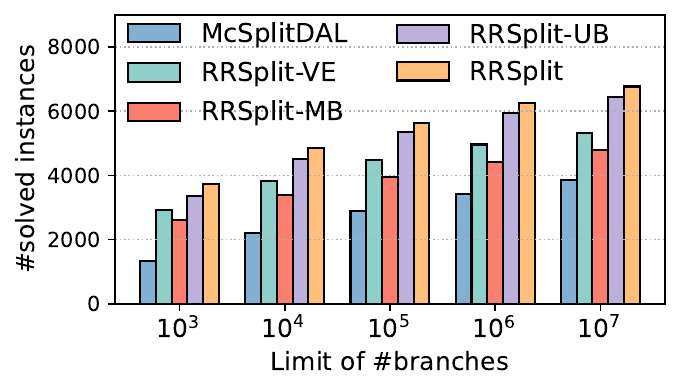}
		}	
		\subfigure[\textsf{\Yui Varying limit of \#branches (LV)}]{
			\includegraphics[width=4.0cm]{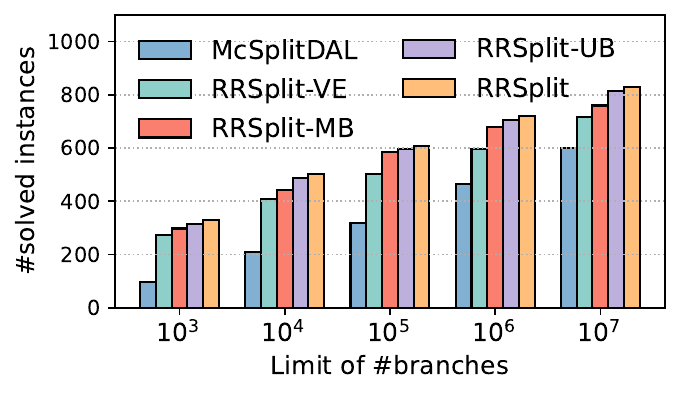}
		}
        \vspace{-0.15in}
	\caption{Comparison among various reductions}
	\label{fig:all_vary_R}
\end{figure}

{\revision
\smallskip
\noindent\textbf{Scalability test.} We test the scalability of our \texttt{RRSplit} on two large datasets, i.e., \textsf{Twitter} and \textsf{DBLP}, which are collected from different domains (http://konect.cc/). Here, \textsf{Twitter} is a social network with 465,017 vertices and 833,540 edges, and \textsf{DBLP} is a collaboration network with 317,080 vertices and 1,049,866 edges. Following existing studies~\cite{arai2023gup,jin2023circinus,sun2023efficient}, we generate the problem instances (i.e., $Q$ and $G$) as follows. Let \textsf{Twitter} or \textsf{DBLP} be the graph $G$. We first extract a set of graphs $Q$ from $G$. Specifically, we conduct a random walk on $G$ and extract a subgraph induced by the visited vertices. By varying the size of the extracted graph $Q$ (among $\{20,30,40,50,60\}$), we extract 5 sets and each of them contains 100 different graphs $Q$.
Then, we generate different problems by pairing the graph $G$ (i.e., \textsf{Twitter} or \textsf{DBLP}) with different graphs $Q$ in the set. In summary, for each dataset, we have 500 different problem instances. 

We compare our \texttt{RRSplit} with \texttt{McSplitDAL} by varying the size of $Q$, and report the average running time in Figure~\ref{fig:scalability_test}. We observe that our \texttt{RRSplit} outperforms \texttt{McSplitDAL} significantly.
% , which demonstrates the scalability of the proposed method. 
Besides, \texttt{McSplitDAL} cannot handle almost all the problem instances within the time and/or space limit (INF/OOM). This is because the implementation of \texttt{McSplitDAL} highly relies on the adjacent matrix of $Q$ and $G$, which introduces huge space and time costs when $G$ is very large. Finally, we observe that \texttt{RRSplit} has the running time increase as the size of $Q$ grows. This is also consistent with the theoretical analysis.
}

\begin{figure}[]
		\subfigure[\textsf{Running time (Twitter)}]{
			\includegraphics[width=4.0cm]{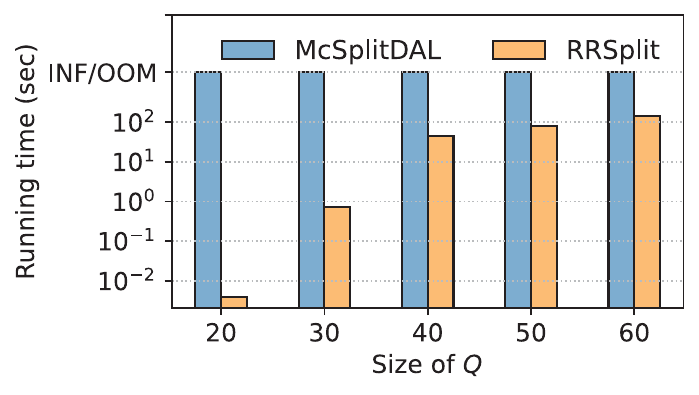}
		}	
		\subfigure[\textsf{Running time (DBLP)}]{
			\includegraphics[width=4.0cm]{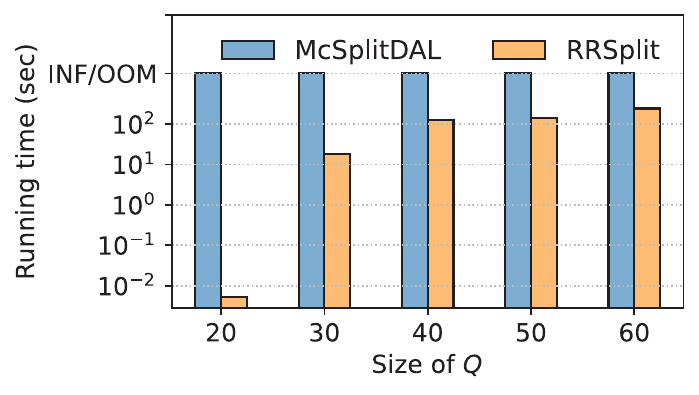}
		}
        \vspace{-0.2in}
	\caption{\revision Scalability test on large datasets}
	\label{fig:scalability_test}
\end{figure}

\subsection{Ablation studies}

We study the effects of various reductions on reducing the redundant computations. In specific, we compare \texttt{RRSplit} with three variants, namely \texttt{RRSplit-VE}: the full version without vertex-equivalence based reductions, \texttt{RRSplit-MB}: the full version without maximality based reductions and \texttt{RRSplit-UB}: the full version without the vertex-equivalence based upper bound, on \textsf{BI} and \textsf{LV}. We report the number of solved problem instances in Figure~\ref{fig:all_vary_R}(a,b) for varying the time limit and in Figure~\ref{fig:all_vary_R}(c,d) for varying the limit of number of formed branches. {\Yui The results on \textsf{CV} and \textsf{PR} show similar clues, which we put in the 
\ifx \CR\undefined
Appendix. 
\else
technical report~\cite{TR}. 
\fi
First, we can see that all four algorithms perform better than the baseline \texttt{McSplitDAL}, among which \texttt{RRSplit} performs the best. This demonstrates the effectiveness of vertex-equivalence-based reductions, maximality-based reductions and vertex-equivalence-based upper bound. Second, \texttt{RRSplit-VE} and \texttt{RRSplit-MB} {\chengB achieve} comparable performance and {\chengB both} contribute to the improvements. Specifically, we note that \texttt{RRSplit-VE} runs slightly faster than \texttt{RRSplit-MB} on \textsf{BI} while \texttt{RRSplit-MB} runs slightly faster than \texttt{RRSplit-VE} on \textsf{LV}. }  {\revision This is possibly because graphs in \textsf{BI} are relatively small biochemical networks where two vertices are more likely to be structural equivalent and thus the vertex-equivalence based reductions outperform other reductions, while graphs in \texttt{LV} are synthetic networks. }
\section{Related Work}
\label{sec:related}

\noindent\textbf{Maximum common subgraph search}. In the literature, there are quite a few studies on finding the maximum common subgraph, which solve the problem either exactly~\cite{levi1973note,mcgregor1982backtrack,abu2014maximum,krissinel2004common,suters2005new,mccreesh2016clique,vismara2008finding,zhoustrengthened,liu2020learning,liu2023hybrid,mccreesh2017partitioning} or approximately~\cite{choi2012efficient,rutgers2010approximate,xiao2009generative,zanfir2018deep,bai2021glsearch}. \underline{First}, among all those exact algorithms, they mainly focus on improving the \emph{practical} performance and most of them are backtracking (also known as branch-and-bound) algorithms~\cite{levi1973note,mcgregor1982backtrack}. Specifically, authors in~\cite{levi1973note,mcgregor1982backtrack} propose the first backtracking framework. The idea is to transform the problem of finding the maximum common subgraph between two given graphs to the problem of finding the maximum clique in the \emph{association graph}. Then, authors in~\cite{mccreesh2016clique,vismara2008finding} follow the previous framework and further improve it by employing the constraint programming techniques. However, these algorithms are all based on a large and dense association graph built from two given graphs, which thus suffer from the efficiency issue. To solve the issue, McCreesh et al.~\cite{mccreesh2017partitioning} propose a new backtracking framework, namely \texttt{McSplit}, which is not based on the maximum clique search problem. Recent works~\cite{zhoustrengthened,liu2020learning,liu2023hybrid} follow \texttt{McSplit} and improve the practical performance by optimizing the policies of branching via learning techniques. Among them, \texttt{McSplitDAL}~\cite{liu2023hybrid} runs faster than {\cheng others.}
% all previous methods. 
We note that some exact algorithms are designed 
% from the theoretical perspective
{\chengB to achieve improvements of theoretical time complexity}
~\cite{abu2014maximum,levi1973note,krissinel2004common,suters2005new}. They have gradually improved the worst-case time complexity from $O^*(1.19^{|V_Q||V_G|})$~\cite{levi1973note} to $O^*(|V_Q|^{(|V_G|+1)})$~\cite{krissinel2004common}, and {\cheng to} $O^*((|V_Q|+1)^{|V_G|})$~\cite{suters2005new},
% {\cheng which is the state-of-the-art to the best of our knowledge}. 
{\chengB which is our best-known {\YuiR worst-case} time complexity for the problem.}
However, these algorithms are of theoretical interests only and not efficient in practice. We remark that (1) our \texttt{RRSplit} not only runs faster than all previous algorithms in practice but also achieves the state-of-the-art worst case time complexity (i.e., $O^*((|V_Q|+1)^{|V_G|})$) in theory and (2) the heuristic polices proposed in~\cite{zhoustrengthened,liu2020learning,liu2023hybrid} are orthogonal to \texttt{RRSplit}. \underline{Second}, since the problem of finding the largest common subgraph is NP-hard, some researchers turn to solve it approximately in polynomial time. Some approximation algorithms include meta-heuristics~\cite{choi2012efficient,rutgers2010approximate}, spectra methods~\cite{xiao2009generative}, and learning-based methods~\cite{zanfir2018deep,bai2021glsearch}. We remark that these techniques cannot be applied to our exact algorithm directly.

\smallskip
\noindent\textbf{Subgraph matching}. Given a target graph and a query graph, subgraph matching aims to find from a target graph all those subgraphs isomorphic to a query graph. We note that maximum common subgraph search is a generalization of subgraph matching. Specifically, given two graphs $Q$ and $G$, maximum common subgraph search {\chengC would} reduce to subgraph matching if 
% one requires 
{\chengB we require}
that the found common subgraph has the size at least $|V(Q)|$ or $|V(G)|$. In recent decades, subgraph matching has been widely studied~\cite{bhattarai2019ceci,ullmann1976algorithm,sun2020rapidmatch,sun2020subgraph,shang2008taming,kim2023fast,han2013turboiso,han2019efficient,cordella2004sub,bi2016efficient,arai2023gup,jin2023circinus,sun2023efficient}. The majority of proposed solutions perform a backtracking search. Among these algorithms, the \emph{candidate filtering} technique, which is designed for removing unnecessary vertices from the target graph, has been shown to be important for improving the practical efficiency~\cite{bhattarai2019ceci,bi2016efficient,han2019efficient,han2013turboiso,kim2023fast}. The technique relies on an auxiliary data structure (e.g., a tree or a directed acyclic graph), which is obtained from the query graph (based on the implicit constraint that each vertex in the query graph must be mapped to a vertex in the found subgraph). We note that it is hard to apply candidate filtering {\cheng to find} the maximum common subgraph (since the mentioned constraint may not hold).
We remark that finding subgraphs exactly isomorphic to a query graph is too restrictive in some real applications due to the data quality issues and/or potential requirements of the fuzzy search (e.g., no result {\chengC would} be returned if there does not exist any subgraph isomorphic to a query graph). Motivated by this, we focus on finding the maximum common subgraph between two graphs in this paper.
\section{Conclusion}
\label{sec:conclusion}

In this paper, we propose a new backtracking algorithm \texttt{RRSplit} for finding the largest common subgraph.
\texttt{RRSplit} is based on our newly-designed reduction rules for reducing the redundant computations and achieves the state-of-the-art worst-case time complexity.
Extensive experiments are conducted on the widely-used graph collections, and the results demonstrate the superiority of our method. In the future, we will adapt our proposed algorithm to solve {\chengE the maximum common subgraph search problem on} other types of graphs, including vertex-labeled and edge-labeled graphs.

\section{acknowledgement}
The research of Kaiqiang Yu, Kaixin Wang and Cheng Long is supported by the Ministry of Education, Singapore, under its Academic Research Fund (Tier 2 Award MOE-T2EP20221-0013 and Tier 1 Awards (RG20/24 and RG77/21)). Any opinions, findings and conclusions or recommendations expressed in this material are those of the author(s) and do not reflect the views of the Ministry of Education, Singapore.
Lakshmanan’s research was supported in part by a grant from
the Natural Sciences and Engineering Research Council of Canada (Grant Number RGPIN-2020-05408).
Reynold Cheng was supported by the Hong Kong Jockey Club Charities Trust (Project 260920140), the University of Hong Kong (Project 2409100399), the HKU Outstanding Research Student Supervisor Award 2022-23, and the HKU Faculty Exchange Award 2024 (Faculty of Engineering).

\balance
\clearpage
\bibliographystyle{ACM-Reference-Format}
\bibliography{SIGMOD_MaxCS}

%%% -*-BibTeX-*-
%%% Do NOT edit. File created by BibTeX with style
%%% ACM-Reference-Format-Journals [18-Jan-2012].

\begin{thebibliography}{54}

%%% ====================================================================
%%% NOTE TO THE USER: you can override these defaults by providing
%%% customized versions of any of these macros before the \bibliography
%%% command.  Each of them MUST provide its own final punctuation,
%%% except for \shownote{} and \showURL{}.  The latter two
%%% do not use final punctuation, in order to avoid confusing it with
%%% the Web address.
%%%
%%% To suppress output of a particular field, define its macro to expand
%%% to an empty string, or better, \unskip, like this:
%%%
%%% \newcommand{\showURL}[1]{\unskip}   % LaTeX syntax
%%%
%%% \def \showURL #1{\unskip}           % plain TeX syntax
%%%
%%% ====================================================================

\ifx \showCODEN    \undefined \def \showCODEN     #1{\unskip}     \fi
\ifx \showISBNx    \undefined \def \showISBNx     #1{\unskip}     \fi
\ifx \showISBNxiii \undefined \def \showISBNxiii  #1{\unskip}     \fi
\ifx \showISSN     \undefined \def \showISSN      #1{\unskip}     \fi
\ifx \showLCCN     \undefined \def \showLCCN      #1{\unskip}     \fi
\ifx \shownote     \undefined \def \shownote      #1{#1}          \fi
\ifx \showarticletitle \undefined \def \showarticletitle #1{#1}   \fi
\ifx \showURL      \undefined \def \showURL       {\relax}        \fi
% The following commands are used for tagged output and should be
% invisible to TeX
\providecommand\bibfield[2]{#2}
\providecommand\bibinfo[2]{#2}
\providecommand\natexlab[1]{#1}
\providecommand\showeprint[2][]{arXiv:#2}

\bibitem[Abu-Khzam(2014)]%
        {abu2014maximum}
\bibfield{author}{\bibinfo{person}{Faisal~N Abu-Khzam}.} \bibinfo{year}{2014}\natexlab{}.
\newblock \showarticletitle{Maximum common induced subgraph parameterized by vertex cover}.
\newblock \bibinfo{journal}{\emph{Inform. Process. Lett.}} \bibinfo{volume}{114}, \bibinfo{number}{3} (\bibinfo{year}{2014}), \bibinfo{pages}{99--103}.
\newblock


\bibitem[Antelo-Collado et~al\mbox{.}(2020)]%
        {antelo2020maximum}
\bibfield{author}{\bibinfo{person}{Aurelio Antelo-Collado}, \bibinfo{person}{Ram{\'o}n Carrasco-Velar}, \bibinfo{person}{Nicol{\'a}s Garc{\'\i}a-Pedrajas}, {and} \bibinfo{person}{Gonzalo Cerruela-Garc{\'\i}a}.} \bibinfo{year}{2020}\natexlab{}.
\newblock \showarticletitle{Maximum common property: a new approach for molecular similarity}.
\newblock \bibinfo{journal}{\emph{Journal of cheminformatics}}  \bibinfo{volume}{12} (\bibinfo{year}{2020}), \bibinfo{pages}{1--22}.
\newblock


\bibitem[Arai et~al\mbox{.}(2023)]%
        {arai2023gup}
\bibfield{author}{\bibinfo{person}{Junya Arai}, \bibinfo{person}{Yasuhiro Fujiwara}, {and} \bibinfo{person}{Makoto Onizuka}.} \bibinfo{year}{2023}\natexlab{}.
\newblock \showarticletitle{GuP: Fast Subgraph Matching by Guard-based Pruning}.
\newblock \bibinfo{journal}{\emph{Proceedings of the ACM on Management of Data}} \bibinfo{volume}{1}, \bibinfo{number}{2} (\bibinfo{year}{2023}), \bibinfo{pages}{1--26}.
\newblock


\bibitem[Bai et~al\mbox{.}(2021)]%
        {bai2021glsearch}
\bibfield{author}{\bibinfo{person}{Yunsheng Bai}, \bibinfo{person}{Derek Xu}, \bibinfo{person}{Yizhou Sun}, {and} \bibinfo{person}{Wei Wang}.} \bibinfo{year}{2021}\natexlab{}.
\newblock \showarticletitle{Glsearch: Maximum common subgraph detection via learning to search}. In \bibinfo{booktitle}{\emph{International Conference on Machine Learning}}. PMLR, \bibinfo{pages}{588--598}.
\newblock


\bibitem[Balasundaram et~al\mbox{.}(2011)]%
        {balasundaram2011clique}
\bibfield{author}{\bibinfo{person}{Balabhaskar Balasundaram}, \bibinfo{person}{Sergiy Butenko}, {and} \bibinfo{person}{Illya~V Hicks}.} \bibinfo{year}{2011}\natexlab{}.
\newblock \showarticletitle{Clique relaxations in social network analysis: The maximum k-plex problem}.
\newblock \bibinfo{journal}{\emph{Operations Research}} \bibinfo{volume}{59}, \bibinfo{number}{1} (\bibinfo{year}{2011}), \bibinfo{pages}{133--142}.
\newblock


\bibitem[Bhattarai et~al\mbox{.}(2019)]%
        {bhattarai2019ceci}
\bibfield{author}{\bibinfo{person}{Bibek Bhattarai}, \bibinfo{person}{Hang Liu}, {and} \bibinfo{person}{H~Howie Huang}.} \bibinfo{year}{2019}\natexlab{}.
\newblock \showarticletitle{Ceci: Compact embedding cluster index for scalable subgraph matching}. In \bibinfo{booktitle}{\emph{Proceedings of the 2019 International Conference on Management of Data}}. \bibinfo{pages}{1447--1462}.
\newblock


\bibitem[Bi et~al\mbox{.}(2016)]%
        {bi2016efficient}
\bibfield{author}{\bibinfo{person}{Fei Bi}, \bibinfo{person}{Lijun Chang}, \bibinfo{person}{Xuemin Lin}, \bibinfo{person}{Lu Qin}, {and} \bibinfo{person}{Wenjie Zhang}.} \bibinfo{year}{2016}\natexlab{}.
\newblock \showarticletitle{Efficient subgraph matching by postponing cartesian products}. In \bibinfo{booktitle}{\emph{Proceedings of the 2016 International Conference on Management of Data}}. \bibinfo{pages}{1199--1214}.
\newblock


\bibitem[Bonnici et~al\mbox{.}(2013)]%
        {bonnici2013subgraph}
\bibfield{author}{\bibinfo{person}{Vincenzo Bonnici}, \bibinfo{person}{Rosalba Giugno}, \bibinfo{person}{Alfredo Pulvirenti}, \bibinfo{person}{Dennis Shasha}, {and} \bibinfo{person}{Alfredo Ferro}.} \bibinfo{year}{2013}\natexlab{}.
\newblock \showarticletitle{A subgraph isomorphism algorithm and its application to biochemical data}.
\newblock \bibinfo{journal}{\emph{BMC bioinformatics}}  \bibinfo{volume}{14} (\bibinfo{year}{2013}), \bibinfo{pages}{1--13}.
\newblock


\bibitem[Bunke(1997)]%
        {bunke1997relation}
\bibfield{author}{\bibinfo{person}{H. Bunke}.} \bibinfo{year}{1997}\natexlab{}.
\newblock \showarticletitle{On a relation between graph edit distance and maximum common subgraph}.
\newblock \bibinfo{journal}{\emph{Pattern recognition letters}} \bibinfo{volume}{18}, \bibinfo{number}{8} (\bibinfo{year}{1997}), \bibinfo{pages}{689--694}.
\newblock


\bibitem[Chang et~al\mbox{.}(2020)]%
        {chang2020speeding}
\bibfield{author}{\bibinfo{person}{Lijun Chang}, \bibinfo{person}{Xing Feng}, \bibinfo{person}{Xuemin Lin}, \bibinfo{person}{Lu Qin}, \bibinfo{person}{Wenjie Zhang}, {and} \bibinfo{person}{Dian Ouyang}.} \bibinfo{year}{2020}\natexlab{}.
\newblock \showarticletitle{Speeding up GED verification for graph similarity search}. In \bibinfo{booktitle}{\emph{2020 IEEE 36th International Conference on Data Engineering (ICDE)}}. IEEE, \bibinfo{pages}{793--804}.
\newblock


\bibitem[Chen et~al\mbox{.}(2019)]%
        {chen2019efficient}
\bibfield{author}{\bibinfo{person}{Xiaoyang Chen}, \bibinfo{person}{Hongwei Huo}, \bibinfo{person}{Jun Huan}, {and} \bibinfo{person}{Jeffrey~Scott Vitter}.} \bibinfo{year}{2019}\natexlab{}.
\newblock \showarticletitle{An efficient algorithm for graph edit distance computation}.
\newblock \bibinfo{journal}{\emph{Knowledge-Based Systems}}  \bibinfo{volume}{163} (\bibinfo{year}{2019}), \bibinfo{pages}{762--775}.
\newblock


\bibitem[Chiang et~al\mbox{.}(2007)]%
        {chiang2007coverage}
\bibfield{author}{\bibinfo{person}{Tony Chiang}, \bibinfo{person}{Denise Scholtens}, \bibinfo{person}{Deepayan Sarkar}, \bibinfo{person}{Robert Gentleman}, {and} \bibinfo{person}{Wolfgang Huber}.} \bibinfo{year}{2007}\natexlab{}.
\newblock \showarticletitle{Coverage and error models of protein-protein interaction data by directed graph analysis}.
\newblock \bibinfo{journal}{\emph{Genome biology}}  \bibinfo{volume}{8} (\bibinfo{year}{2007}), \bibinfo{pages}{1--14}.
\newblock


\bibitem[Choi et~al\mbox{.}(2012)]%
        {choi2012efficient}
\bibfield{author}{\bibinfo{person}{Jaeun Choi}, \bibinfo{person}{Yourim Yoon}, {and} \bibinfo{person}{Byung-Ro Moon}.} \bibinfo{year}{2012}\natexlab{}.
\newblock \showarticletitle{An efficient genetic algorithm for subgraph isomorphism}. In \bibinfo{booktitle}{\emph{Proceedings of the 14th annual conference on Genetic and evolutionary computation}}. \bibinfo{pages}{361--368}.
\newblock


\bibitem[Cordella et~al\mbox{.}(2004)]%
        {cordella2004sub}
\bibfield{author}{\bibinfo{person}{Luigi~P Cordella}, \bibinfo{person}{Pasquale Foggia}, \bibinfo{person}{Carlo Sansone}, {and} \bibinfo{person}{Mario Vento}.} \bibinfo{year}{2004}\natexlab{}.
\newblock \showarticletitle{A (sub) graph isomorphism algorithm for matching large graphs}.
\newblock \bibinfo{journal}{\emph{IEEE transactions on pattern analysis and machine intelligence}} \bibinfo{volume}{26}, \bibinfo{number}{10} (\bibinfo{year}{2004}), \bibinfo{pages}{1367--1372}.
\newblock


\bibitem[Ehrlich and Rarey(2011)]%
        {ehrlich2011maximum}
\bibfield{author}{\bibinfo{person}{Hans-Christian Ehrlich} {and} \bibinfo{person}{Matthias Rarey}.} \bibinfo{year}{2011}\natexlab{}.
\newblock \showarticletitle{Maximum common subgraph isomorphism algorithms and their applications in molecular science: a review}.
\newblock \bibinfo{journal}{\emph{Wiley Interdisciplinary Reviews: Computational Molecular Science}} \bibinfo{volume}{1}, \bibinfo{number}{1} (\bibinfo{year}{2011}), \bibinfo{pages}{68--79}.
\newblock


\bibitem[Gouda and Hassaan(2016)]%
        {gouda2016csi_ged}
\bibfield{author}{\bibinfo{person}{Karam Gouda} {and} \bibinfo{person}{Mosab Hassaan}.} \bibinfo{year}{2016}\natexlab{}.
\newblock \showarticletitle{CSI\_GED: An efficient approach for graph edit similarity computation}. In \bibinfo{booktitle}{\emph{2016 IEEE 32nd International Conference on Data Engineering (ICDE)}}. IEEE, \bibinfo{pages}{265--276}.
\newblock


\bibitem[Han et~al\mbox{.}(2019)]%
        {han2019efficient}
\bibfield{author}{\bibinfo{person}{Myoungji Han}, \bibinfo{person}{Hyunjoon Kim}, \bibinfo{person}{Geonmo Gu}, \bibinfo{person}{Kunsoo Park}, {and} \bibinfo{person}{Wook-Shin Han}.} \bibinfo{year}{2019}\natexlab{}.
\newblock \showarticletitle{Efficient subgraph matching: Harmonizing dynamic programming, adaptive matching order, and failing set together}. In \bibinfo{booktitle}{\emph{Proceedings of the 2019 International Conference on Management of Data}}. \bibinfo{pages}{1429--1446}.
\newblock


\bibitem[Han et~al\mbox{.}(2013)]%
        {han2013turboiso}
\bibfield{author}{\bibinfo{person}{Wook-Shin Han}, \bibinfo{person}{Jinsoo Lee}, {and} \bibinfo{person}{Jeong-Hoon Lee}.} \bibinfo{year}{2013}\natexlab{}.
\newblock \showarticletitle{Turboiso: towards ultrafast and robust subgraph isomorphism search in large graph databases}. In \bibinfo{booktitle}{\emph{Proceedings of the 2013 ACM SIGMOD International Conference on Management of Data}}. \bibinfo{pages}{337--348}.
\newblock


\bibitem[Hati et~al\mbox{.}(2016)]%
        {hati2016image}
\bibfield{author}{\bibinfo{person}{Avik Hati}, \bibinfo{person}{Subhasis Chaudhuri}, {and} \bibinfo{person}{Rajbabu Velmurugan}.} \bibinfo{year}{2016}\natexlab{}.
\newblock \showarticletitle{Image co-segmentation using maximum common subgraph matching and region co-growing}. In \bibinfo{booktitle}{\emph{Computer Vision--ECCV 2016: 14th European Conference, Amsterdam, The Netherlands, October 11-14, 2016, Proceedings, Part VI 14}}. Springer, \bibinfo{pages}{736--752}.
\newblock


\bibitem[Hoffmann et~al\mbox{.}(2017)]%
        {hoffmann2017between}
\bibfield{author}{\bibinfo{person}{Ruth Hoffmann}, \bibinfo{person}{Ciaran McCreesh}, {and} \bibinfo{person}{Craig Reilly}.} \bibinfo{year}{2017}\natexlab{}.
\newblock \showarticletitle{Between subgraph isomorphism and maximum common subgraph}. In \bibinfo{booktitle}{\emph{Proceedings of the AAAI Conference on Artificial Intelligence}}, Vol.~\bibinfo{volume}{31}.
\newblock


\bibitem[Jin et~al\mbox{.}(2023)]%
        {jin2023circinus}
\bibfield{author}{\bibinfo{person}{Tatiana Jin}, \bibinfo{person}{Boyang Li}, \bibinfo{person}{Yichao Li}, \bibinfo{person}{Qihui Zhou}, \bibinfo{person}{Qianli Ma}, \bibinfo{person}{Yunjian Zhao}, \bibinfo{person}{Hongzhi Chen}, {and} \bibinfo{person}{James Cheng}.} \bibinfo{year}{2023}\natexlab{}.
\newblock \showarticletitle{Circinus: Fast redundancy-reduced subgraph matching}.
\newblock \bibinfo{journal}{\emph{Proceedings of the ACM on Management of Data}} \bibinfo{volume}{1}, \bibinfo{number}{1} (\bibinfo{year}{2023}), \bibinfo{pages}{1--26}.
\newblock


\bibitem[Kann(1992)]%
        {kann1992approximability}
\bibfield{author}{\bibinfo{person}{Viggo Kann}.} \bibinfo{year}{1992}\natexlab{}.
\newblock \showarticletitle{On the approximability of the maximum common subgraph problem}. In \bibinfo{booktitle}{\emph{STACS 92: 9th Annual Symposium on Theoretical Aspects of Computer Science Cachan, France, February 13--15, 1992 Proceedings 9}}. Springer, \bibinfo{pages}{375--388}.
\newblock


\bibitem[Kim et~al\mbox{.}(2021)]%
        {kim2021versatile}
\bibfield{author}{\bibinfo{person}{Hyunjoon Kim}, \bibinfo{person}{Yunyoung Choi}, \bibinfo{person}{Kunsoo Park}, \bibinfo{person}{Xuemin Lin}, \bibinfo{person}{Seok-Hee Hong}, {and} \bibinfo{person}{Wook-Shin Han}.} \bibinfo{year}{2021}\natexlab{}.
\newblock \showarticletitle{Versatile equivalences: Speeding up subgraph query processing and subgraph matching}. In \bibinfo{booktitle}{\emph{Proceedings of the 2021 International Conference on Management of Data}}. \bibinfo{pages}{925--937}.
\newblock


\bibitem[Kim et~al\mbox{.}(2023)]%
        {kim2023fast}
\bibfield{author}{\bibinfo{person}{Hyunjoon Kim}, \bibinfo{person}{Yunyoung Choi}, \bibinfo{person}{Kunsoo Park}, \bibinfo{person}{Xuemin Lin}, \bibinfo{person}{Seok-Hee Hong}, {and} \bibinfo{person}{Wook-Shin Han}.} \bibinfo{year}{2023}\natexlab{}.
\newblock \showarticletitle{Fast subgraph query processing and subgraph matching via static and dynamic equivalences}.
\newblock \bibinfo{journal}{\emph{The VLDB journal}} \bibinfo{volume}{32}, \bibinfo{number}{2} (\bibinfo{year}{2023}), \bibinfo{pages}{343--368}.
\newblock


\bibitem[Kim(2023)]%
        {kim2023efficient}
\bibfield{author}{\bibinfo{person}{Jongik Kim}.} \bibinfo{year}{2023}\natexlab{}.
\newblock \showarticletitle{Efficient graph edit distance computation using isomorphic vertices}.
\newblock \bibinfo{journal}{\emph{Pattern Recognition Letters}}  \bibinfo{volume}{168} (\bibinfo{year}{2023}), \bibinfo{pages}{71--78}.
\newblock


\bibitem[Krissinel and Henrick(2004)]%
        {krissinel2004common}
\bibfield{author}{\bibinfo{person}{Evgeny~B Krissinel} {and} \bibinfo{person}{Kim Henrick}.} \bibinfo{year}{2004}\natexlab{}.
\newblock \showarticletitle{Common subgraph isomorphism detection by backtracking search}.
\newblock \bibinfo{journal}{\emph{Software: Practice and Experience}} \bibinfo{volume}{34}, \bibinfo{number}{6} (\bibinfo{year}{2004}), \bibinfo{pages}{591--607}.
\newblock


\bibitem[Larsen and Baumbach(2017)]%
        {larsen2017cytomcs}
\bibfield{author}{\bibinfo{person}{Simon~J Larsen} {and} \bibinfo{person}{Jan Baumbach}.} \bibinfo{year}{2017}\natexlab{}.
\newblock \showarticletitle{CytoMCS: a multiple maximum common subgraph detection tool for Cytoscape}.
\newblock \bibinfo{journal}{\emph{Journal of integrative bioinformatics}} \bibinfo{volume}{14}, \bibinfo{number}{2} (\bibinfo{year}{2017}), \bibinfo{pages}{20170014}.
\newblock


\bibitem[Levi(1973)]%
        {levi1973note}
\bibfield{author}{\bibinfo{person}{Giorgio Levi}.} \bibinfo{year}{1973}\natexlab{}.
\newblock \showarticletitle{A note on the derivation of maximal common subgraphs of two directed or undirected graphs}.
\newblock \bibinfo{journal}{\emph{Calcolo}} \bibinfo{volume}{9}, \bibinfo{number}{4} (\bibinfo{year}{1973}), \bibinfo{pages}{341--352}.
\newblock


\bibitem[Lewis(1983)]%
        {lewis1983michael}
\bibfield{author}{\bibinfo{person}{Harry~R Lewis}.} \bibinfo{year}{1983}\natexlab{}.
\newblock \showarticletitle{Michael R. $\Pi$Garey and David S. Johnson. Computers and intractability. A guide to the theory of NP-completeness. WH Freeman and Company, San Francisco1979, x+ 338 pp.}
\newblock \bibinfo{journal}{\emph{The Journal of Symbolic Logic}} \bibinfo{volume}{48}, \bibinfo{number}{2} (\bibinfo{year}{1983}), \bibinfo{pages}{498--500}.
\newblock


\bibitem[Liu et~al\mbox{.}(2020)]%
        {liu2020learning}
\bibfield{author}{\bibinfo{person}{Yanli Liu}, \bibinfo{person}{Chu-Min Li}, \bibinfo{person}{Hua Jiang}, {and} \bibinfo{person}{Kun He}.} \bibinfo{year}{2020}\natexlab{}.
\newblock \showarticletitle{A learning based branch and bound for maximum common subgraph related problems}. In \bibinfo{booktitle}{\emph{Proceedings of the AAAI Conference on Artificial Intelligence}}, Vol.~\bibinfo{volume}{34}. \bibinfo{pages}{2392--2399}.
\newblock


\bibitem[Liu et~al\mbox{.}(2023)]%
        {liu2023hybrid}
\bibfield{author}{\bibinfo{person}{Yanli Liu}, \bibinfo{person}{Jiming Zhao}, \bibinfo{person}{Chu-Min Li}, \bibinfo{person}{Hua Jiang}, {and} \bibinfo{person}{Kun He}.} \bibinfo{year}{2023}\natexlab{}.
\newblock \showarticletitle{Hybrid learning with new value function for the maximum common induced subgraph problem}. In \bibinfo{booktitle}{\emph{Proceedings of the AAAI Conference on Artificial Intelligence}}, Vol.~\bibinfo{volume}{37}. \bibinfo{pages}{4044--4051}.
\newblock


\bibitem[McCreesh et~al\mbox{.}(2016)]%
        {mccreesh2016clique}
\bibfield{author}{\bibinfo{person}{Ciaran McCreesh}, \bibinfo{person}{Samba~Ndojh Ndiaye}, \bibinfo{person}{Patrick Prosser}, {and} \bibinfo{person}{Christine Solnon}.} \bibinfo{year}{2016}\natexlab{}.
\newblock \showarticletitle{Clique and constraint models for maximum common (connected) subgraph problems}. In \bibinfo{booktitle}{\emph{International Conference on Principles and Practice of Constraint Programming}}. Springer, \bibinfo{pages}{350--368}.
\newblock


\bibitem[McCreesh et~al\mbox{.}(2017)]%
        {mccreesh2017partitioning}
\bibfield{author}{\bibinfo{person}{Ciaran McCreesh}, \bibinfo{person}{Patrick Prosser}, {and} \bibinfo{person}{James Trimble}.} \bibinfo{year}{2017}\natexlab{}.
\newblock \showarticletitle{A partitioning algorithm for maximum common subgraph problems}. In \bibinfo{booktitle}{\emph{Proceedings of the 26th International Joint Conference on Artificial Intelligence}}. \bibinfo{pages}{712--719}.
\newblock


\bibitem[McGregor(1982)]%
        {mcgregor1982backtrack}
\bibfield{author}{\bibinfo{person}{James~J McGregor}.} \bibinfo{year}{1982}\natexlab{}.
\newblock \showarticletitle{Backtrack search algorithms and the maximal common subgraph problem}.
\newblock \bibinfo{journal}{\emph{Software: Practice and Experience}} \bibinfo{volume}{12}, \bibinfo{number}{1} (\bibinfo{year}{1982}), \bibinfo{pages}{23--34}.
\newblock


\bibitem[Nguyen et~al\mbox{.}(2019)]%
        {nguyen2019applications}
\bibfield{author}{\bibinfo{person}{Thien Nguyen}, \bibinfo{person}{Dominic Yang}, \bibinfo{person}{Yurun Ge}, \bibinfo{person}{Hao Li}, {and} \bibinfo{person}{Andrea~L Bertozzi}.} \bibinfo{year}{2019}\natexlab{}.
\newblock \showarticletitle{Applications of structural equivalence to subgraph isomorphism on multichannel multigraphs}. In \bibinfo{booktitle}{\emph{2019 IEEE International Conference on Big Data (Big Data)}}. IEEE, \bibinfo{pages}{4913--4920}.
\newblock


\bibitem[Nirmala et~al\mbox{.}(2016)]%
        {nirmala2016vertex}
\bibfield{author}{\bibinfo{person}{Parisutham Nirmala}, \bibinfo{person}{Ramasubramony~Sulochana Lekshmi}, {and} \bibinfo{person}{Rethnasamy Nadarajan}.} \bibinfo{year}{2016}\natexlab{}.
\newblock \showarticletitle{Vertex cover-based binary tree algorithm to detect all maximum common induced subgraphs in large communication networks}.
\newblock \bibinfo{journal}{\emph{Knowledge and Information Systems}}  \bibinfo{volume}{48} (\bibinfo{year}{2016}), \bibinfo{pages}{229--252}.
\newblock


\bibitem[Park and Reeves(2011)]%
        {park2011deriving}
\bibfield{author}{\bibinfo{person}{Younghee Park} {and} \bibinfo{person}{Douglas Reeves}.} \bibinfo{year}{2011}\natexlab{}.
\newblock \showarticletitle{Deriving common malware behavior through graph clustering}. In \bibinfo{booktitle}{\emph{Proceedings of the 6th ACM Symposium on Information, Computer and Communications Security}}. \bibinfo{pages}{497--502}.
\newblock


\bibitem[Piao et~al\mbox{.}(2023)]%
        {piao2023computing}
\bibfield{author}{\bibinfo{person}{Chengzhi Piao}, \bibinfo{person}{Tingyang Xu}, \bibinfo{person}{Xiangguo Sun}, \bibinfo{person}{Yu Rong}, \bibinfo{person}{Kangfei Zhao}, {and} \bibinfo{person}{Hong Cheng}.} \bibinfo{year}{2023}\natexlab{}.
\newblock \showarticletitle{Computing graph edit distance via neural graph matching}.
\newblock \bibinfo{journal}{\emph{Proceedings of the VLDB Endowment}} \bibinfo{volume}{16}, \bibinfo{number}{8} (\bibinfo{year}{2023}), \bibinfo{pages}{1817--1829}.
\newblock


\bibitem[Rutgers et~al\mbox{.}(2010)]%
        {rutgers2010approximate}
\bibfield{author}{\bibinfo{person}{Jochem~H Rutgers}, \bibinfo{person}{Pascal~T Wolkotte}, \bibinfo{person}{Philip~KF H{\"o}lzenspies}, \bibinfo{person}{Jan Kuper}, {and} \bibinfo{person}{Gerard~JM Smit}.} \bibinfo{year}{2010}\natexlab{}.
\newblock \showarticletitle{An approximate maximum common subgraph algorithm for large digital circuits}. In \bibinfo{booktitle}{\emph{2010 13th Euromicro Conference on Digital System Design: Architectures, Methods and Tools}}. IEEE, \bibinfo{pages}{699--705}.
\newblock


\bibitem[Schmidt et~al\mbox{.}(2020)]%
        {schmidt2020disconnected}
\bibfield{author}{\bibinfo{person}{Robert Schmidt}, \bibinfo{person}{Florian Krull}, \bibinfo{person}{Anna~Lina Heinzke}, {and} \bibinfo{person}{Matthias Rarey}.} \bibinfo{year}{2020}\natexlab{}.
\newblock \showarticletitle{Disconnected maximum common substructures under constraints}.
\newblock \bibinfo{journal}{\emph{Journal of Chemical Information and Modeling}} \bibinfo{volume}{61}, \bibinfo{number}{1} (\bibinfo{year}{2020}), \bibinfo{pages}{167--178}.
\newblock


\bibitem[Shang et~al\mbox{.}(2008)]%
        {shang2008taming}
\bibfield{author}{\bibinfo{person}{Haichuan Shang}, \bibinfo{person}{Ying Zhang}, \bibinfo{person}{Xuemin Lin}, {and} \bibinfo{person}{Jeffrey~Xu Yu}.} \bibinfo{year}{2008}\natexlab{}.
\newblock \showarticletitle{Taming verification hardness: an efficient algorithm for testing subgraph isomorphism}.
\newblock \bibinfo{journal}{\emph{Proceedings of the VLDB Endowment}} \bibinfo{volume}{1}, \bibinfo{number}{1} (\bibinfo{year}{2008}), \bibinfo{pages}{364--375}.
\newblock


\bibitem[Solnon et~al\mbox{.}(2015)]%
        {solnon2015complexity}
\bibfield{author}{\bibinfo{person}{Christine Solnon}, \bibinfo{person}{Guillaume Damiand}, \bibinfo{person}{Colin De~La~Higuera}, {and} \bibinfo{person}{Jean-Christophe Janodet}.} \bibinfo{year}{2015}\natexlab{}.
\newblock \showarticletitle{On the complexity of submap isomorphism and maximum common submap problems}.
\newblock \bibinfo{journal}{\emph{Pattern Recognition}} \bibinfo{volume}{48}, \bibinfo{number}{2} (\bibinfo{year}{2015}), \bibinfo{pages}{302--316}.
\newblock


\bibitem[Sun and Luo(2020)]%
        {sun2020subgraph}
\bibfield{author}{\bibinfo{person}{Shixuan Sun} {and} \bibinfo{person}{Qiong Luo}.} \bibinfo{year}{2020}\natexlab{}.
\newblock \showarticletitle{Subgraph matching with effective matching order and indexing}.
\newblock \bibinfo{journal}{\emph{IEEE Transactions on Knowledge and Data Engineering}} \bibinfo{volume}{34}, \bibinfo{number}{1} (\bibinfo{year}{2020}), \bibinfo{pages}{491--505}.
\newblock


\bibitem[Sun et~al\mbox{.}(2020)]%
        {sun2020rapidmatch}
\bibfield{author}{\bibinfo{person}{Shixuan Sun}, \bibinfo{person}{Xibo Sun}, \bibinfo{person}{Yulin Che}, \bibinfo{person}{Qiong Luo}, {and} \bibinfo{person}{Bingsheng He}.} \bibinfo{year}{2020}\natexlab{}.
\newblock \showarticletitle{Rapidmatch: A holistic approach to subgraph query processing}.
\newblock \bibinfo{journal}{\emph{Proceedings of the VLDB Endowment}} \bibinfo{volume}{14}, \bibinfo{number}{2} (\bibinfo{year}{2020}), \bibinfo{pages}{176--188}.
\newblock


\bibitem[Sun and Luo(2023)]%
        {sun2023efficient}
\bibfield{author}{\bibinfo{person}{Xibo Sun} {and} \bibinfo{person}{Qiong Luo}.} \bibinfo{year}{2023}\natexlab{}.
\newblock \showarticletitle{Efficient GPU-Accelerated Subgraph Matching}.
\newblock \bibinfo{journal}{\emph{Proceedings of the ACM on Management of Data}} \bibinfo{volume}{1}, \bibinfo{number}{2} (\bibinfo{year}{2023}), \bibinfo{pages}{1--26}.
\newblock


\bibitem[Sun et~al\mbox{.}(2021)]%
        {sun2021effective}
\bibfield{author}{\bibinfo{person}{Yi Sun}, \bibinfo{person}{Ali~Kashif Bashir}, \bibinfo{person}{Usman Tariq}, {and} \bibinfo{person}{Fei Xiao}.} \bibinfo{year}{2021}\natexlab{}.
\newblock \showarticletitle{Effective malware detection scheme based on classified behavior graph in IIoT}.
\newblock \bibinfo{journal}{\emph{Ad Hoc Networks}}  \bibinfo{volume}{120} (\bibinfo{year}{2021}), \bibinfo{pages}{102558}.
\newblock


\bibitem[Suters et~al\mbox{.}(2005)]%
        {suters2005new}
\bibfield{author}{\bibinfo{person}{W~Henry Suters}, \bibinfo{person}{Faisal~N Abu-Khzam}, \bibinfo{person}{Yun Zhang}, \bibinfo{person}{Christopher~T Symons}, \bibinfo{person}{Nagiza~F Samatova}, {and} \bibinfo{person}{Michael~A Langston}.} \bibinfo{year}{2005}\natexlab{}.
\newblock \showarticletitle{A new approach and faster exact methods for the maximum common subgraph problem}. In \bibinfo{booktitle}{\emph{Computing and Combinatorics: COCOON 2005}}. Springer, \bibinfo{pages}{717--727}.
\newblock


\bibitem[Ullmann(1976)]%
        {ullmann1976algorithm}
\bibfield{author}{\bibinfo{person}{Julian~R Ullmann}.} \bibinfo{year}{1976}\natexlab{}.
\newblock \showarticletitle{An algorithm for subgraph isomorphism}.
\newblock \bibinfo{journal}{\emph{Journal of the ACM (JACM)}} \bibinfo{volume}{23}, \bibinfo{number}{1} (\bibinfo{year}{1976}), \bibinfo{pages}{31--42}.
\newblock


\bibitem[Vismara and Valery(2008)]%
        {vismara2008finding}
\bibfield{author}{\bibinfo{person}{Philippe Vismara} {and} \bibinfo{person}{Beno{\^\i}t Valery}.} \bibinfo{year}{2008}\natexlab{}.
\newblock \showarticletitle{Finding maximum common connected subgraphs using clique detection or constraint satisfaction algorithms}. In \bibinfo{booktitle}{\emph{International Conference on Modelling, Computation and Optimization in Information Systems and Management Sciences}}. Springer, \bibinfo{pages}{358--368}.
\newblock


\bibitem[Xiao et~al\mbox{.}(2009)]%
        {xiao2009generative}
\bibfield{author}{\bibinfo{person}{Bai Xiao}, \bibinfo{person}{Edwin~R Hancock}, {and} \bibinfo{person}{Richard~C Wilson}.} \bibinfo{year}{2009}\natexlab{}.
\newblock \showarticletitle{A generative model for graph matching and embedding}.
\newblock \bibinfo{journal}{\emph{Computer Vision and Image Understanding}} \bibinfo{volume}{113}, \bibinfo{number}{7} (\bibinfo{year}{2009}), \bibinfo{pages}{777--789}.
\newblock


\bibitem[Yan et~al\mbox{.}(2005)]%
        {yan2005substructure}
\bibfield{author}{\bibinfo{person}{Xifeng Yan}, \bibinfo{person}{Philip~S Yu}, {and} \bibinfo{person}{Jiawei Han}.} \bibinfo{year}{2005}\natexlab{}.
\newblock \showarticletitle{Substructure similarity search in graph databases}. In \bibinfo{booktitle}{\emph{Proceedings of the 2005 ACM SIGMOD international conference on Management of data}}. \bibinfo{pages}{766--777}.
\newblock


\bibitem[Yang et~al\mbox{.}(2023)]%
        {yang2023structural}
\bibfield{author}{\bibinfo{person}{Dominic Yang}, \bibinfo{person}{Yurun Ge}, \bibinfo{person}{Thien Nguyen}, \bibinfo{person}{Denali Molitor}, \bibinfo{person}{Jacob~D Moorman}, {and} \bibinfo{person}{Andrea~L Bertozzi}.} \bibinfo{year}{2023}\natexlab{}.
\newblock \showarticletitle{Structural Equivalence in Subgraph Matching}.
\newblock \bibinfo{journal}{\emph{IEEE Transactions on Network Science and Engineering}} (\bibinfo{year}{2023}).
\newblock


\bibitem[Zanfir and Sminchisescu(2018)]%
        {zanfir2018deep}
\bibfield{author}{\bibinfo{person}{Andrei Zanfir} {and} \bibinfo{person}{Cristian Sminchisescu}.} \bibinfo{year}{2018}\natexlab{}.
\newblock \showarticletitle{Deep learning of graph matching}. In \bibinfo{booktitle}{\emph{Proceedings of the IEEE conference on computer vision and pattern recognition}}. \bibinfo{pages}{2684--2693}.
\newblock


\bibitem[Zhou et~al\mbox{.}(2022)]%
        {zhoustrengthened}
\bibfield{author}{\bibinfo{person}{Jianrong Zhou}, \bibinfo{person}{Kun He}, \bibinfo{person}{Jiongzhi Zheng}, \bibinfo{person}{Chu-Min Li}, {and} \bibinfo{person}{Yanli Liu}.} \bibinfo{year}{2022}\natexlab{}.
\newblock \showarticletitle{A Strengthened Branch and Bound Algorithm for the Maximum Common (Connected) Subgraph Problem}. In \bibinfo{booktitle}{\emph{Proceedings of the Thirty-First International Joint Conference on Artificial Intelligence, {IJCAI} 2022}}. \bibinfo{pages}{1908--1914}.
\newblock


\end{thebibliography}
\clearpage
\appendix
\section{Additional Experimental Results}
{\Yui
In this section, we provide additional experimental results, including \emph{comparison by varying similarities} and \emph{comparison among various reductions}.

\smallskip
\noindent\textbf{Varying the similarities of two input graphs (additional results)}. We test different problem instances as the similarity varies on \textsf{CV} and \textsf{PR}. We remark that all tested problem instances in \textsf{CV} (resp. \textsf{PR}) have their similarities vary from 0.9 to 1 (resp. equal to 1). We report the average running time in Figure~\ref{fig:appendix_all_vary_S}(a) and (b) and the average number of formed branches in Figure~\ref{fig:appendix_all_vary_S}(c) and (d). The results show the similar clues to those on \textsf{BI} and \textsf{LV}. In specific, our \texttt{RRSplit} runs around 5$\times$-10$\times$ faster and forms fewer branches than \texttt{McSplitDAL}.

\smallskip
\noindent\textbf{Varying different reductions (additional results)}. We compare \texttt{RRSplit} with three variants, namely \texttt{RRSplit-VE}, \texttt{RRSplit-MB} and \texttt{RRSplit-UB}, on \textsf{CV} and \textsf{PR}. We report the number of solved problem instances in Figure~\ref{fig:appendix_all_vary_R} (a) and (b) for varying the time limit and in Figure~\ref{fig:appendix_all_vary_R} (c) and (d) for varying the limit of number of formed branches. The results on \textsf{CV} and \textsf{PR} show similar trends to those on \textsf{BI} and \textsf{LV}. First, we can see that all four algorithms performs better than \texttt{McSplitDAL}, among which \texttt{RRSplit} performs the best. This indicates the effectiveness of the proposed vertex-equivalence based reductions and maximality based reductions. Second, we note that \texttt{RRSplit-MB} and \texttt{RRSplit-VE} achieve the comparable performance. 
}
\begin{figure}[]
		\subfigure[\textsf{Running time (CV)}]{
			\includegraphics[width=4.0cm]{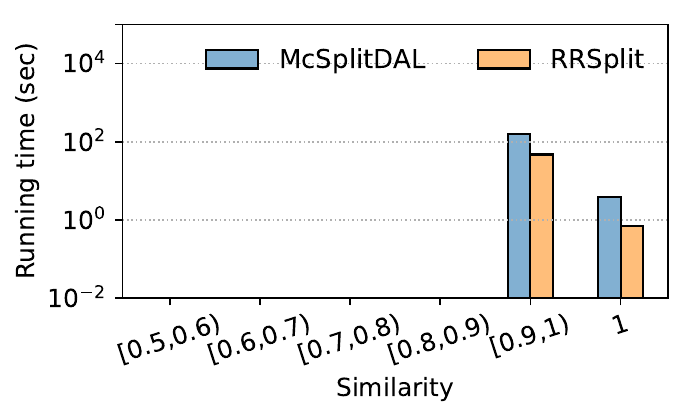}
		}	
		\subfigure[\textsf{Running time (PR)}]{
			\includegraphics[width=4.0cm]{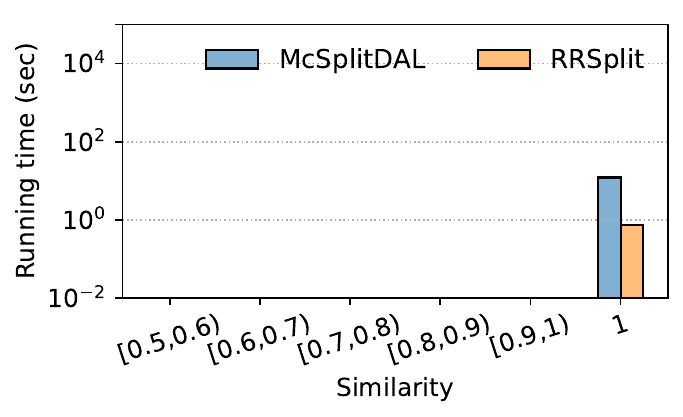}
		}
        \subfigure[\textsf{\# of branches (CV)}]{
			\includegraphics[width=4.0cm]{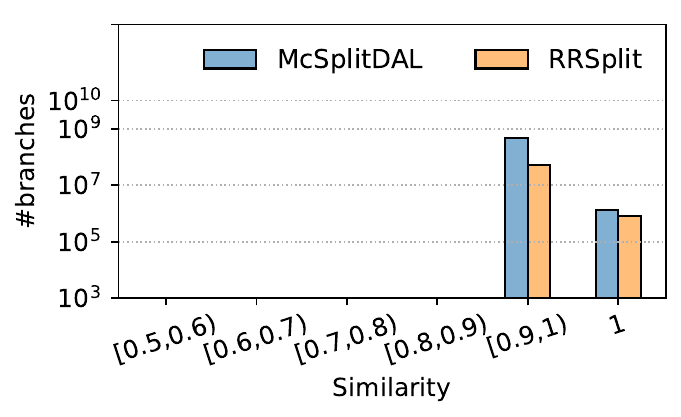}
		}	
		\subfigure[\textsf{\# of branches (PR)}]{
			\includegraphics[width=4.0cm]{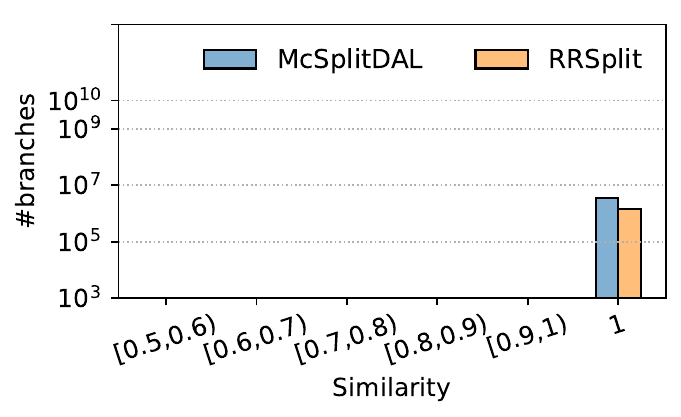}
		}
	\caption{Comparison by varying similarities (additional results)}
	\label{fig:appendix_all_vary_S}
\end{figure}

\begin{figure}[]
		\subfigure[\textsf{Varying time limits (CV)}]{
			\includegraphics[width=4.0cm]{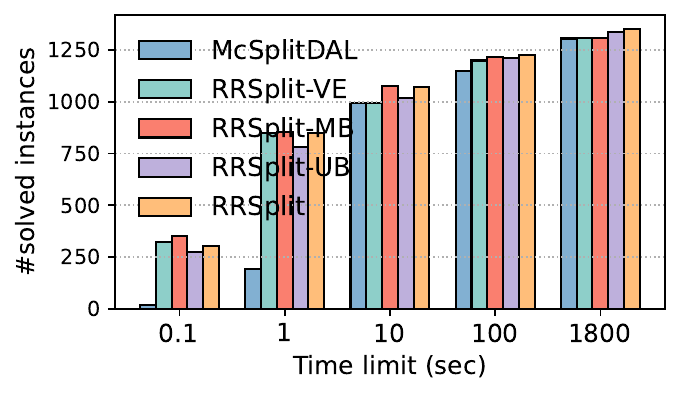}
		}	
		\subfigure[\textsf{Varying time limits (PR)}]{
			\includegraphics[width=4.0cm]{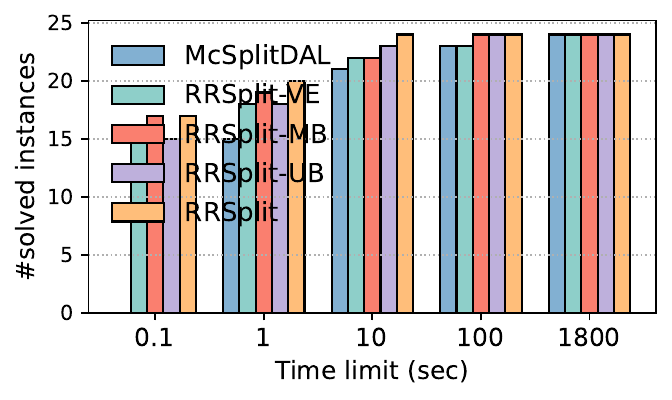}
		}
        \subfigure[\textsf{Varying limit of \#branches (CV)}]{
			\includegraphics[width=4.0cm]{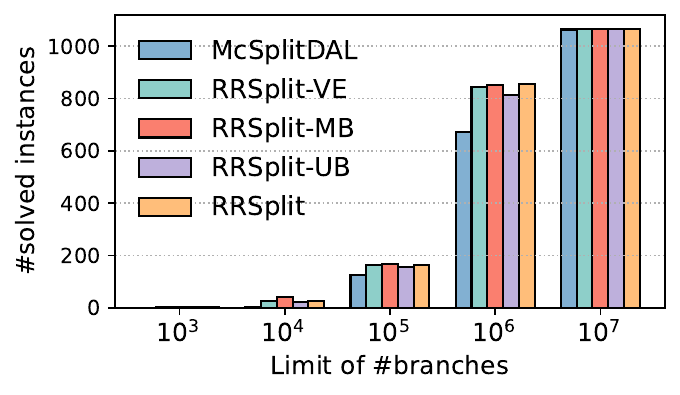}
		}	
		\subfigure[\textsf{Varying limit of \#branches (PR)}]{
			\includegraphics[width=4.0cm]{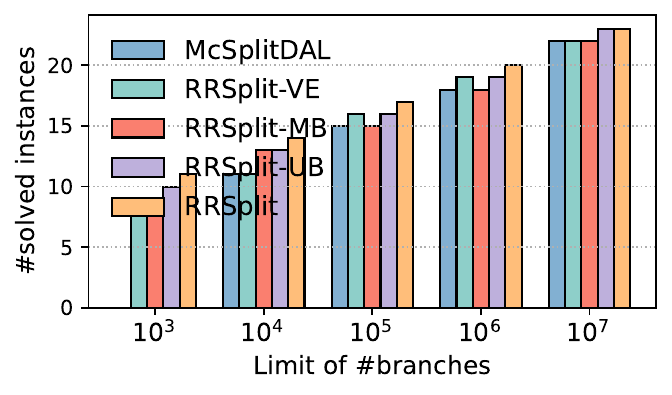}
		}
	\caption{Comparison among various reductions (additional results)}
	\label{fig:appendix_all_vary_R}
\end{figure}

\section{Additional Proofs}

\smallskip
\noindent\textbf{Lemma~\ref{lemma:reduction_for_VE1}}. \emph{
    Let $(S,C,D)$ be a branch. Common subgraph $S_{iso}$ at Equation~(\ref{eq:iso1}) has been found before the formation of $(S,C,D)$.}

\begin{proof}
    We note that the recursive branching process forms a recursion tree where each tree node corresponding to a branch. Consider the path from the initial branch $(\emptyset, V_Q\times V_G,\emptyset)$ to $(S,C,D)$, there exists an ascendant branch of $(S,C,D)$, denoted by $B_{asc}=(S_{asc},C_{asc},D_{asc})$, where $u_{equ}$ is selected as the branching vertex, since $\langle u_{equ},\phi(u_{equ}) \rangle$ is in $S$.
    We can see that there exists one sub-branch $B'_{asc}=(S'_{asc},C'_{asc},D'_{asc})$ of $B_{asc}$ formed by including $\langle u_{equ},v \rangle$, and all common subgraphs within $B'_{asc}$ has been found before the formation of $(S,C,D)$, since $\langle u_{equ},v \rangle$ is in $D$. We then show that common subgraph $S_{iso}$ can be found within $B_{asc}'$, i.e., $S'_{asc}\subseteq S_{iso}\subseteq S'_{asc}\cup C'_{asc}$. 
    \underline{First}, we have $S'_{asc}\subseteq S_{iso}$ since (1) $S'_{asc}=S_{asc}\cup\{\langle u_{equ},v \rangle\}$, (2) $S_{sub}$ is a common subgraph in $B_{asc}$ and thus $S_{asc}\subseteq S_{sub}$, (3) $S_{asc}$ does not include $\langle u_{equ},\phi(u_{equ})\rangle$ or $\langle u,v\rangle$ since they are in $C_{asc}$ and will be included to the partial solution at $B_{asc}'$ and $(S\cup \{\langle u,v \rangle\}, C\backslash u\backslash v)$, and thus (4) by combining all the above, we have  $S'_{asc}=S_{asc}\cup\{\langle u_{equ},v \rangle\}\subseteq S_{sub}\cup\{\langle u_{equ},v \rangle\} \backslash \{\langle u_{equ},\phi(u_{equ})\rangle,\langle u,v\rangle\}\subseteq S_{iso}$. 
    \underline{Second}, we have $S_{iso}\subseteq S'_{asc}\cup C'_{asc}$ based on the following two facts. 

    \begin{itemize}
        \item \textbf{Fact 1.} $S_{sub}\backslash\{\langle u_{equ},\phi(u_{equ})\rangle,\langle u,v\rangle\}\subseteq S'_{asc}\cup C'_{asc}$.
        \item \textbf{Fact 2.} $\langle u_{equ},v \rangle\in S'_{asc}$ and $\langle u,\phi(u_{equ}) \rangle\in C'_{asc}$.
    \end{itemize}
    
    Fact 1 holds since (1) $S_{sub}\subseteq S_{asc}\cup C_{asc}$ (note that $S_{sub}$ is a common subgraph in $B_{asc}$), (2) vertices $u_{equ}$ and $v$ do not appear in $S_{sub}\backslash\{\langle u_{equ},\phi(u_{equ})\rangle,\langle u,v\rangle\}$ and thus we can derive that $S_{sub}\backslash\{\langle u_{equ},\phi(u_{equ})\rangle,\langle u,v\rangle\}\subseteq (S_{asc}\cup C_{asc})\backslash u_{equ}\backslash v$ (note that $\langle u_{equ},\phi(u_{equ})\rangle$ and $\langle u,v\rangle$ are the unique vertex pairs that consist of $u_{equ}$ and $v$  in $S_{sub}$, respectively), and (3) $S'_{asc}=S_{asc}\cup \{\langle u_{equ},v\rangle\}$ and $C'_{asc}=C_{asc}\backslash u_{equ}\backslash v$ based on the branching rule.

    Fact 2 can be verified as follows. Vertex pair $\langle u_{equ},v \rangle$ is in $S'_{asc}$ since $S'_{asc}=S_{asc}\cup \{\langle u_{equ},v\rangle\}$. We note that vertices $u_{equ}$, $\phi(u_{equ})$, $u$ and $v$ appear in $C_{asc}$ since $\langle u_{equ},\phi(u_{equ})\rangle$ and $\langle u,v\rangle$ are in $C_{asc}$ discussed before. Let $C_{asc}=X_1\times Y_1 \cup  X_2\times Y_2\cup \cdots \cup  X_c\times Y_c$ where $c$ is a positive integer. 
    It is no hard to verify that, for two vertices $u$ and $u'$ (resp. $v$ and $v'$) appearing in $C_{asc}$, $u$ and $u'$ are in the same subset $X_i$ (resp. $Y_i$) of $C_{asc}$ with $1\leq i\leq c$ if and only if $u$ and $u'$ (resp. $v$ and $v'$) have the same set of neighbours and non-neighbours in $q$ (resp. $g$) according to Equation~(\ref{eq:update_candidate_set}).
    %for a subset $\langle X_i,Y_i \rangle$ with $1\leq i\leq c$, all vertices in $X_i$ (resp. $Y_i$) have the same set of neighbours and non-neighbours in $q$ (resp. $g$) since otherwise they will be split into different subsets according to Equation~(\ref{eq:update_candidate_set}); for any two different subsets
    %
    Besides, we have $X_i\cap X_j=\emptyset$ and $Y_i\cap Y_j=\emptyset$ for $1\leq i\neq j \leq c$ as discussed before. %since each vertex appearing in the candidate set is split into exactly one subset according to Equation~(\ref{eq:update_candidate_set}). 
    Based on the above, we assume that $u_{equ}$ appears in a subset $X_i,Y_i$ of $C_{asc}$ where $u_{equ}$ is in $X_i$ and $1\leq i\leq c$. We can deduce that $u$ is in $X_i$ since $u_{equ}$ and $u$ are structurally equivalent and thus they have the same set of neighbours and non-neighbours in $q$. Besides, we can deduce that vertices $\phi(u_{equ})$ and $v$ are in $Y_i$ since (1) $\langle u_{equ},\phi(u_{equ})\rangle$ and $\langle u,v\rangle$ are in $C_{asc}$, (2) $u$ and $v$ are in exactly one subset $X_i$, and thus (3) they must appear in $X_i\times Y_i$.   
\end{proof}

\smallskip
\noindent\textbf{Lemma~\ref{lemma:reduction_for_VE2}}. \emph{Let $(S,C,D)$ be a branch where $u$ is selected as the branching vertex. Common subgraph $S_{iso}$ {\chengC defined in} Equation~(\ref{eq:iso2}) has been found before the formation of $(S,C\backslash u,D)$ at the second group.  }
\begin{proof}
    \underline{First}, we note that $\langle u_{equ},\phi_{iso}(u_{equ}) \rangle$ is in $C\backslash u$ and also in $C$ since otherwise $S_{sub}$ cannot include $\langle u_{equ},\phi_{iso}(u_{equ}) \rangle$. This is because (1) $S\subseteq S_{sub}\subseteq S\cup C\backslash u$ since $S_{sub}$ is a common subgraph in the sub-branch $(S,C\backslash u,D)$ and (2) $S$ does not include $\langle u_{equ},\phi_{iso}(u_{equ}) \rangle$ since $u_{equ}$ appears in $C\backslash u$.
    \underline{Second}, we note that $\phi_{iso}(u_{equ})$ is in $Y$. Recall that $X\times Y$ is the branching set at $(S,C,D)$. This is because (1) $u_{equ}$ is in the same subset $X$ as $u$ since $u_{equ}$ and $u$ are structurally equivalent and thus have the same set of neighbours and non-neighbours in $q$, and (2) $\langle u_{equ},\phi_{iso}(u_{equ}) \rangle$ is in $C$ as discussed before.
    \underline{Third}, we can derive that there exists a sub-branch $(S\cup \{\langle u,\phi_{iso}(u_{equ}) \rangle\},C\backslash u\backslash \phi_{iso}(u_{equ}),D')$, which is formed at branch $(S,C,D)$  by including $\langle u,\phi_{iso}(u_{equ}) \rangle$ before the formation of $(S,C\backslash u,D)$, since $\phi_{iso}(u_{equ})\in Y$.
    \underline{Forth}, we show that $S_{iso}$ is in $(S\cup \{\langle u,\phi_{iso}(u_{equ}) \rangle\},C\backslash u\backslash \phi_{iso}(u_{equ}),D')$, formally, $S\cup \{\langle u,\phi_{iso}(u_{equ}) \rangle\} \subseteq S_{iso}\subseteq S\cup \{\langle u,\phi_{iso}(u_{equ}) \rangle\}\cup (C\backslash u\backslash \phi_{iso}(u_{equ}))$.
    We have $S\subseteq S_{sub}\subseteq S\cup (C\backslash u)$ since $S_{sub}$ is a common subgraph in $(S,C\backslash u,D)$. Let $S'=S\cup \{\langle u,\phi_{iso}(u_{equ}) \rangle\}$, it can be proved as below.
    \begin{eqnarray}
        && S\subseteq S_{sub}\subseteq S\cup (C\backslash u)\\
        && \Rightarrow S\subseteq S_{sub}\backslash\{\langle u_{equ}, \phi_{iso}(u_{equ})\rangle\}\subseteq S\cup (C\backslash u\backslash \phi_{iso}(u_{equ})) \label{eq:llemma_eq_1}\\
        && \Rightarrow S' \subseteq S_{iso}\subseteq S'\cup (C\backslash u\backslash \phi_{iso}(u_{equ})) \label{eq:llemma_eq_2}
    \end{eqnarray}
    Note that Equation~(\ref{eq:llemma_eq_1}) holds since $\langle u_{equ},\phi_{iso}(u_{equ}) \rangle$ is in $S$; Equation~(\ref{eq:llemma_eq_2}) is derived by including the vertex pair $\langle u, \phi_{iso}(u_{equ})\rangle$.
    %We first prove the fact that $\phi_{iso}(u_{equ})\in Y^*$ since (1) $S\cup\langle u^*,\phi_{iso}(u_{equ}) \rangle$, as a subgraph of common subgraph $S_{iso}$ (i.e., $S\cup\langle u^*,\phi_{iso}(u_{equ})\subseteq S_{iso}$), is also a common subgraph based on Definition~\ref{def:CIS}, (2) $\langle u_{equ},\phi(u_{equ}) \rangle$, which contains $\phi_{iso}(u_{equ})$, is in $C\backslash u^*$ since otherwise $S_{sub}$ cannot include $\langle u_{equ},\phi(u_{equ}) \rangle$ (note that $S\subseteq S_{sub}\subseteq S\cup C\backslash u^*$ since $S_{sub}$ is a common subgraph in the sub-branch $(S,C\backslash u^*,D)$ and $S$ does not include $\langle u_{equ},\phi(u_{equ}) \rangle$ since $u_{equ}$ appears in $C\backslash u^*$), (2) $\langle u^*,\phi_{iso}(u_{equ}) \rangle$ is in 
\end{proof}

\smallskip
\noindent\textbf{Lemma~\ref{lemma:maximality}} \emph{Let $B=(S,C,D)$ be a branch {\chengC and $\langle u,v \rangle$ be a candidate vertex pair that satisfies the condition in Equation~(\ref{eq:condition}).} There exists one largest common subgraph $S_{opt}$ in the branch $B$ such that $S_{opt}$ contains 
    {\chengC $\langle u,v \rangle$.}}
\begin{proof}
    This can be proved by construction. Let $S^*=(q^*,g^*,\phi^*)$ be one largest common subgraph to be found in $B$. Note that if $S^*$ contains the candidate vertex pair $\langle u,v \rangle$, we can finish the proof by constructing $S_{opt}$ as $S^*$. Otherwise, if $\langle u,v \rangle$ is not in $S^*$, we prove the correctness by constructing one largest common subgraph $S_{opt}$ to be found in $B$ that contains candidate vertex pair $\langle u,v \rangle$, i.e., $S\subseteq S_{opt} \subseteq S\cup C$,  $|S_{opt}|=|S^*|$ and $\langle u,v \rangle\in S_{opt}$.
    In general, there are four different cases.

    \smallskip
    \noindent\underline{\textbf{Case 1:}} $u\notin V_{q^*}$ and $v\in V_{g^*}$. In this case, there exists a vertex pair $\langle\phi^{*-1}(v),v \rangle$ in $S^*$ where $\phi^{*-1}$ is the inverse of $\phi^*$. We construct $S_{opt}$ by replacing the vertex pair $\langle\phi^{*-1}(v),v \rangle$ with $\langle u,v \rangle$, i.e.,
    \begin{equation}
        S_{opt}=S^*\backslash\{ \langle\phi^{*-1}(v),v \rangle\} \cup \{\langle u,v \rangle\}.
    \end{equation}
    Clearly, we have $S\subseteq S_{opt}\subseteq S\cup C$ (i.e., $S_{opt}$ is in $B$) since $S^*$ is in $B$ and $\langle u,v\rangle$ is in the candidate set $C$. Besides, we have $|S_{opt}|=|S^*|$ and $\langle u,v \rangle\in S_{opt}$ based on the above construction. Finally, we deduce that $S_{opt}$ is a common subgraph by showing that any two vertex pairs in $S_{opt}$ satisfy Equation~(\ref{eq:isomorphic}), i.e., $g_{opt}$ is isomorphic to $q_{opt}$ under the bijection $\phi_{opt}$. \underline{First}, $S^*\backslash\{\langle \phi^{*-1}(v),v\rangle\}$, as a subset of $S^*$, is a common subgraph and thus has any two vertex pairs inside satisfying Equation~(\ref{eq:isomorphic}) (note that any subset of a common subgraph is still a common subgraph); \underline{Second}, for each pair $\langle u',v' \rangle$ in $S$, $u$ is adjacent to $u'$ if and only if $v$ is adjacent to $v'$ (since $\langle u,v \rangle$ is a candidate pair which can form a common subgraph with $S$); \underline{Third}, for each pair $\langle u',v' \rangle$ in $S_{opt}\backslash S\backslash\{\langle \phi^{*-1}(v),v\rangle\}$, it is clear that $\langle u',v' \rangle$ is in one subset $ X\times Y$ of $\mathcal{P}(C)$ and thus $u$ is adjacent to $u'$ if and only if $v$ is adjacent to $v'$ based on Equation~(\ref{eq:condition}). Therefore, any two vertex pairs in $S_{opt}$ will satisfy the Equation~(\ref{eq:isomorphic}).

    \smallskip
    \noindent\underline{\textbf{Case 2:}} $u\in V_{q^*}$ and $v\notin V_{g^*}$. There exists a vertex pair $\langle u,\phi^*(u) \rangle$ in $S^*$. We construct $S_{opt}$ by replacing $\langle u,\phi^*(u) \rangle$ with $\langle u,v \rangle$, i.e., $S_{opt}=S^*\backslash \{\langle u,\phi^*(u) \rangle\}\cup\{\langle u,v \rangle\}$. Similar to Case 1, we can prove that $S_{opt}$ includes $\langle u,v \rangle$ and is one largest common subgraph to be found in $B$. 

    \smallskip
    \noindent\underline{\textbf{Case 3:}} $u\in V_{q^*}$ and $v\in V_{g^*}$. There exists two distinct vertex pairs $\langle u,\phi^*(u) \rangle$ and $\langle \phi^{*-1}(v),v \rangle$ in $S^*$. We construct $S_{opt}$ by replacing these two vertex pairs with $\langle \phi^{*-1}(v),\phi(u) \rangle$ and $\langle u,v \rangle$, formally,
    \begin{equation}
        S_{opt}\!\!=\!\!S^*\backslash\{\langle u,\phi^*(u) \rangle,\!\langle\phi^{*-1}(v),v \rangle\}\!\cup\!\{\langle \phi^{*-1}(v),\phi^*(u) \rangle,\!\langle u,v \rangle\}.
    \end{equation}
    Clearly, we have $S\subseteq S_{opt}\subseteq S\cup C$ (i.e., $S_{opt}$ is in $B$), $|S_{opt}|=|S^*|$ and $\langle u,v \rangle\in S_{opt}$ based on the above construction. We then deduce that $S_{opt}$ is a common subgraph  by showing that any two vertex pairs in $S_{opt}$ satisfy Equation~(\ref{eq:isomorphic}).
    \underline{First}, $S^*\backslash\{\langle u,\phi^*(u) \rangle,\langle \phi^{*-1}(v),v\rangle\}$, as a subset of $S^*$, is a common subgraph and thus has any two vertex pairs inside satisfying Equation~(\ref{eq:isomorphic});
    \underline{Second}, consider a vertex pair $\langle u',v' \rangle$ in $S^*\backslash\{\langle u,\phi^*(u) \rangle,\langle \phi^{*-1}(v),v\rangle\}$. Similar to Case 1, we can prove that $u$ is adjacent to $u'$ if and only if $v$ is adjacent to $v'$. Besides, we show that $\phi^{*-1}(v)$ is adjacent to $u'$ if and only if $\phi(u)$ is adjacent to $v'$ since (1) $(\phi^{*-1}(v),u')\in E_Q\Leftrightarrow (v,v')\in E_G$ and $(u,u')\in E_Q\Leftrightarrow (\phi^*(u),v')\in E_G$ (since the common subgraph $S^*$ contains $\{\langle u,\phi^*(u) \rangle,\langle \phi^{*-1}(v),v\rangle\}$), (2) $ (v,v')\in E_G \Leftrightarrow (u,u')\in E_Q$ (as we shown above), and thus (3) they can be combined as $(\phi^{*-1}(v),u')\in E_Q\Leftrightarrow (v,v')\in E_G \Leftrightarrow (u,u')\in E_Q \Leftrightarrow (\phi^*(u),v')\in E_G$.
    %\begin{equation}
        %(\phi^{*-1}(v),u')\in E_Q\Leftrightarrow (v,v')\in E_G \Leftrightarrow (u,u')\in E_Q \Leftrightarrow (\phi^*(u),v')\in E_G  \nonumber
    %\end{equation}
    %
    \underline{Third}, we have $(u,\phi^{*-1}(v))\in E_Q\Leftrightarrow (v,\phi^*(u))\in E_G$ since the common subgraph $S^*$ contains $\{\langle u,\phi^*(u) \rangle,\langle \phi^{*-1}(v),v\rangle\}$ and thus $(u,\phi^{*-1}(v))\in E_Q\Leftrightarrow (\phi^*(u),v)\in E_G$ (note that $(\phi^*(u),v)$ refers to the same edge as $(v,\phi^*(u))$ since the graphs $Q$ and $G$ are undirected). Therefore, any two vertex pairs in $R_{opt}$ will satisfy Equation~(\ref{eq:isomorphic}).

    \smallskip
    \noindent\underline{\textbf{Case 4:}} $u\notin V_{q^*}$ and $v\notin V_{g^*}$. We note that this case will not occur since otherwise the contradiction is derived by showing that $S^*\cup \{\langle u,v\rangle\}$ is a larger common subgraph (note that the proof is similar to Case 1 and thus be omitted).
\end{proof}

\smallskip
\noindent\textbf{Fact}. \emph{Let $(S,C,D)$ be a branch where $D=\{u_1\}\times A_1 \cup \{u_2\} \times A_2 \cup ... \cup \{u_d\}\times A_d$ and $d$ is a positive integer. For $1\leq i\neq j \leq d$, if $u_i$ and $u_j$ are structurally equivalent, we have $A_i\cap A_j = \emptyset$.}
\begin{proof}
    This can be proved by contradiction. Assume that there exists a vertex $v$ in $A_i\cap A_j$. Clearly, $u_i$ and $u_j$ are in $S$. Consider the path from the initial branch $(\emptyset,V_Q\times V_G,\emptyset)$ to the branch $(S,C,D)$ in the recursion tree. Without loss of the generality, suppose that $u_i$ is selected as the branching vertex before $u_j$ in the path. Hence, consider the ascendant branch of $(S,C,D)$, denoted by $B_{asc}=(S_{asc},C_{asc},D_{asc})$, where $u_j$ is selected as the branching vertex. We can easily deduce that $\langle u_i,v \rangle$ is in $D_{asc}$, vertex $u_j$ does not appear in $D_{asc}$ and $\langle u_j,v \rangle$ is in $C_{asc}$. For a sub-branch of $B_{asc}$ which is formed by including $\langle u_j,v \rangle$, i.e., $(S_{asc}\cup\{\langle u_j,v \rangle\}, C\backslash u_j\backslash v)$, it can be pruned by the proposed reduction at the first group since there exists a vertex pair $\langle u_i,v \rangle$ in $D_{asc}$ such that $u_i\in \Psi (u_j)$. As a result, $\langle u_j,v \rangle$ will not be included to $D$ according to the maintenance of the exclusion set, which leads to the contradiction.
\end{proof}

\end{document}